\documentclass[journal]{IEEEtran}
\usepackage[pdfencoding=auto, psdextra]{hyperref}
\usepackage{amsmath,amsfonts,amsthm}
\DeclareMathOperator*{\argmin}{arg\,min}

\newcommand{\matr}[1]{\bm{#1}}
\usepackage{bm}
\usepackage{cite}
\newcommand{\vect}[1]{\bm{#1}}
\theoremstyle{definition}
\newtheorem{theorem}{Theorem}
\newtheorem{corollary}{Corollary}[theorem] %
\newtheorem{defn}{Definition}
\newtheorem{remark}{Remark}
\newtheorem{propn}[]{Proposition}
\newtheorem{lemma}[]{Lemma}

\usepackage{bbm}
\usepackage{algorithmic}
\usepackage{calc}
\usepackage{array}
\usepackage{tabularx}
\usepackage{multirow}

\usepackage{mathtools}

\usepackage{makecell}

\usepackage[font=normalsize,labelfont=sf,textfont=sf]{subcaption}

\usepackage{textcomp}
\usepackage{stfloats}
\usepackage{url}
\usepackage{verbatim}
\usepackage{graphicx}
\usepackage{xcolor}
\def\BibTeX{{\rm B\kern-.05em{\sc i\kern-.025em b}\kern-.08em
    T\kern-.1667em\lower.7ex\hbox{E}\kern-.125emX}}
\usepackage{balance}
\usepackage{booktabs}
\usepackage{cases}
\usepackage[capitalise]{cleveref}
\usepackage{xstring}
\DeclareMathSymbol{\Theta}{\mathord}{operators}{"02}
\newcommand{\set}[1]{\mathcal{#1}}
\newcommand{\sqfrob}[1]{\left\lvert\left\lvert#1\right\rvert\right\rvert^{2}_{F}}
\newcommand{\sqnormvec}[1]{\left\lvert\left\lvert#1\right\rvert\right\rvert^{2}_{2}}

\newcommand{\trace}[1]{\textrm{tr}\left(#1\right)}
\newcommand{\rank}[1]{\textrm{rank}\left(#1\right)}
\NewDocumentCommand{\msub}{o o m}
{
    \IfValueT{#2}{\left[#3\right]_{\set{#1}, \set{#2}}}
    \IfNoValueTF{#1}{#3}{\left[#3\right]_{\set{#1}}}
}
\NewDocumentCommand{\msubgen}{o o m}
{
    \IfValueT{#2}{\left[#3\right]_{#1, #2}}
    \IfNoValueTF{#1}{#3}{\left[#3\right]_{#1}}
}
\NewDocumentCommand{\matrsub}{ o o m }
{
    \IfValueT{#2}{[\matr{#3}]_{\set{#1}, \set{#2}}}
    \IfNoValueTF{#1}{\matr{#3}}{[\matr{#3}]_{\set{#1}}}
}
\NewDocumentCommand{\matrsubpow}{ o o m m}
{
    \IfValueT{#2}{[\matr{#3}^{#4}]_{\set{#1}, \set{#2}}}
    \IfNoValueTF{#1}{\matr{#3}^{#4}}{[\matr{#3}^{#4}]_{\set{#1}}}
}
\NewDocumentCommand{\vectsub}{ o o m }
{
    \IfValueT{#2}{[\vect{#3}]_{\set{#1}, \set{#2}}}
    \IfNoValueTF{#1}{\vect{#3}}{[\vect{#3}]_{\set{#1}}}
}
\NewDocumentCommand{\vectsubgen}{ o o m }
{
    \IfValueT{#2}{[{#3}]_{{#1}, {#2}}}
    \IfNoValueTF{#1}{{#3}}{[{#3}]_{{#1}}}
}
\NewDocumentCommand{\projbl}{ o o }
{
    \IfValueT{#2}{[\matr{\Pi}_{bl(\set{K})}]_{\set{#1},\set{#2}}}
    \IfNoValueTF{#1}{\matr{\Pi}_{bl(\set{K})}}{[\matr{\Pi}_{bl(\set{K})}]_{\set{#1}}}
}
\NewDocumentCommand{\projblgen}{ o o }
{
    \IfValueT{#2}{[\matr{\Pi}_{bl(\set{K})}]_{{#1},{#2}}}
    \IfNoValueTF{#1}{\matr{\Pi}_{bl(\set{K})}}{[\matr{\Pi}_{bl(\set{K})}]_{{#1}}}
}
\NewDocumentCommand{\proj}{ o o m }
{
    \IfValueT{#2}{[\matr{\Pi}_{\set{#3}}]_{\set{#1}, \set{#2}}}    \IfNoValueTF{#1}{\matr{\Pi}_{\set{#3}}}{[\matr{\Pi}_{\set{#3}}]_{\set{#1}}}
}
\NewDocumentCommand{\projgen}{ o o m }
{
    \IfValueT{#2}{[\matr{\Pi}_{\set{#3}}]_{{#1}, {#2}}}
    \IfNoValueTF{#1}{\matr{\Pi}_{{#3}}}{[\matr{\Pi}_{{#3}}]_{{#1}}}
}
\NewDocumentCommand{\expect}{ o m }
{
    \IfNoValueTF{#1}{{\mathbb{E}}\left[{{#2}}\right] }{{\mathbb{E}}_{#1}\left[{{#2}}\right]}
}
\NewDocumentCommand{\expectnobig}{ o m }
{
    \IfNoValueTF{#1}{{\mathbb{E}}[{{#2}}] }{{\mathbb{E}}_{#1}[{{#2}}]}
}
\newcommand{\matrsubU}[1]{\matrsub[#1, K]{U}}
\newcommand{\matrsubUGen}[1]{[\matr{U}]_{#1, \set{K}}}
\newcommand{\xd}[1]{\textcolor{black}{#1}}

\begin{document}
\title{On the Impact of Sample Size in Reconstructing Noisy Graph Signals: A Theoretical Characterisation}
\author{\IEEEauthorblockN{Baskaran Sripathmanathan,} \and \IEEEauthorblockN{Xiaowen Dong,} \and \IEEEauthorblockN{Michael Bronstein}\\
\IEEEauthorblockA{University of Oxford}
}

\maketitle

\begin{abstract}
Reconstructing a signal on a graph from noisy observations of a subset of the vertices is a fundamental problem in the field of graph signal processing. This paper investigates how sample size affects reconstruction error in the presence of noise via an in-depth theoretical analysis of the two most common reconstruction methods in the literature, least-squares reconstruction (LS) and graph-Laplacian regularised reconstruction (GLR). Our theorems show that at sufficiently low signal-to-noise ratios (SNRs), under these reconstruction methods we may simultaneously decrease sample size and decrease average reconstruction error. We further show that at sufficiently low SNRs, for LS reconstruction we have a $\Lambda$-shaped error curve and for GLR reconstruction, a sample size of $\mathcal{O}(\sqrt{N})$, where $N$ is the total number of vertices, results in lower reconstruction error than near full observation. We present thresholds on the SNRs, $\tau$ and $\tau_{GLR}$, below which the error is non-monotonic, and illustrate these theoretical results with experiments across multiple random graph models, sampling schemes and SNRs. These results demonstrate that any decision in sample-size choice has to be made in light of the noise levels in the data.
\end{abstract}


\section{Introduction}
\IEEEPARstart{R}{eal}-world signals, such as brain fMRIs \cite{itani2021graph}, urban air pollution \cite{jain2014big}, and political preferences \cite{renoust2017estimating}, are often noisy and incomplete, making analysis of the signals harder. The reconstruction of these signals from limited observation is of practical importance, and can benefit from the fact that they can be treated as graph signals,  signals defined on a network domain. 
Graph signal processing (GSP) generalises the highly successful tools of sampling and reconstruction in classical signal processing by extending the classical shift operator to a graph shift operator \cite{ortega2018graph} such as the adjacency matrix \cite{EOptimalChen} or the graph Laplacian, enabling us to extrapolate the full data across the graph from observations on a subset of vertices \cite{tanaka2020sampling}. %

In the literature, the vast majority of studies on graph-based sampling focus on designing efficient sampling schemes that are approximately optimal under certain criteria \cite[Chapter 6]{pukelsheim2006optimal}, because optimal vertex choice under noise is in general NP-hard \cite{nikolov2022proportional, chamon2017greedy}. While \xd{these studies provide useful understanding in terms of sampling and reconstruction for graph signals, they} focus on the performance of sampling schemes at fixed sample sizes \cite{xie2019bayesian, puy2018random,shomorony2014sampling, tremblay2017determinantal,wang2018optimal,wang2019low, bai2020fast,  anis2016efficient}, while much less attention has been paid to understanding the impact of varying sample size on the mean squared error (MSE) of reconstruction. Sample size is an important parameter in both understanding and using sampling schemes, especially in the common setting of a fixed sample budget. %

\begin{table}[ht]
    \centering
    \setlength{\belowcaptionskip}{3pt}
  \caption{  Studies on the impact of sample size on MSE.}
    \begin{subtable}[h]{\linewidth}
    \begin{tabularx}{\linewidth}{|>{\raggedright\arraybackslash}p{2.6cm}|>{\centering\arraybackslash}p{1.2cm}|X|>{\raggedright\arraybackslash}p{1.3cm}|}%
        \hline 
        \textbf{Reference} \rule{0pt}{8pt} & \textbf{Considers Noise?} & \textbf{Reconstruction Method} & \textbf{Range of Possible Sample Sizes} \\
        \hline 
        Tremblay, 2017 \cite{tremblay2017determinantal} \rule{0pt}{8pt} & \checkmark & LS, GLR & Limited \\
        Wang, 2018 \cite{wang2018optimal} & \checkmark & LS & Limited \\
        Wang, 2019 \cite{wang2019low} & \checkmark & LS & Limited \\
        Bai, 2020 \cite{bai2020fast} & \checkmark & GLR & Limited \\
        Anis, 2016 \cite{anis2016efficient} & \checkmark & LS & Full \\
        \hline
    \end{tabularx}
    \caption{Empirical Results}

    \end{subtable}
    \begin{subtable}[h]{\linewidth}
    \begin{tabularx}{\linewidth}{|>{\raggedright\arraybackslash}p{2.6cm}|>{\centering\arraybackslash}p{1.2cm}|X|>{\raggedright\arraybackslash}p{1.3cm}|}%
        \hline
        \textbf{Reference} \rule{0pt}{8pt} & \textbf{Considers Noise?} & \textbf{Reconstruction Method} & \textbf{Range of Possible Sample Sizes} \\
        \hline
        Puy, 2018 \cite{puy2018random} \rule{0pt}{8pt} & \checkmark & LS & Limited \\
        Chamon, 2017 \cite{chamon2017greedy} & \checkmark & LS (Regularised) & Full \\
        Shomorony, 2014 \cite{shomorony2014sampling} & \texttimes & LS & Full \\
        \textbf{This Paper} & \textbf{\checkmark} & \textbf{LS, GLR} & \textbf{Full} \\
        \hline
    \end{tabularx}
    \caption{Theoretical Results}

\end{subtable}
    \label{tbl:Lit_review}
    \vspace{-0.4cm}
\end{table}
The literature studies the impact of sample size both empirically and theoretically (summarised in Table \ref{tbl:Lit_review}), and can be further divided by the setting considered: whether the observations are noisy or noiseless, and by which reconstruction method. Empirical results, linked to sampling schemes, are usually obtained in the noisy setting under least squares (LS) reconstruction or its variants \cite{ tremblay2017determinantal,wang2018optimal,wang2019low,chamon2017greedy} or graph-Laplacian regularised (GLR) reconstruction \cite{bai2020fast, tremblay2017determinantal}. These results show that MSE decreases as sample size increases in a restricted range of noise level and sample size. An exception is \cite[Fig. 1]{anis2016efficient} which empirically shows non-monotonicity of MSE with sample size under LS as it considers the {full range of possible sample sizes, including sample sizes less than bandwidth,} and a relatively high noise level. Of the three main theoretical results on the impact of sample size in the literature, two focus on slightly different settings: \cite{shomorony2014sampling} presents sample size bounds for perfect signal reconstruction in the noiseless setting, while \cite[Lemma 1]{chamon2017greedy} proves that MSE decreases as sample size increases in the noisy setting and provides bounds on the impact of sample size on MSE, but assumes a specific form of regularised LS which is Bayesian optimal. While these theoretical results provide valuable insight, the settings they are based on do not always agree with those in the empirical studies above, hence a generic understanding is still lacking. The third theoretical result \cite[Theorem 6]{puy2018random} assumes unregularised LS and noisy reconstruction, and proves an MSE upper bound that is tight and decreasing in sample size; however, this result is practically constrained to the case where sample size is greater than or equal to bandwidth \cite[Eq. 3]{puy2018random}.

In this paper, we fill the gap in the literature by providing a theoretical characterisation of the impact of sample size on the average MSE  (or MMSE  -- see Section \ref{sec:Optimality_Criteria}) across the full range of possible sample sizes in the most common settings, i.e. noisy observations and LS or GLR reconstruction. More specifically, we focus on whether under sufficiently low SNRs, decreasing sample size may actually decrease MSE. Furthermore, we investigate both the full range of sample sizes and all possible levels of noise, which is important for the application of graph sampling to domains where the noise may be greater than the signal \xd{(e.g. finance, with typically low signal-to-noise ratio \cite{nabar2023conservative}, or environmental sciences,  with inaccurate measurements from low-cost sensors \cite{bush2023impact})} or when observing more samples might not be feasible (e.g. resource-constrained settings). This breadth is only possible through our rigorous theoretical characterisation which allows us to understand behaviour at high noise levels, to characterise behaviour on arbitrarily large graphs, and to show when certain behaviours of MSE happen and why. %

Our results begin by using a Bias-Variance decomposition to explain why decreasing sample size may decrease MSE for any linear reconstruction method \xd{(Section \ref{sec:general})}, and we then specialise to LS and GLR (Section \ref{sec:every_x}). We prove that under LS reconstruction of a $k$-bandlimited signal, if the samples were chosen to be optimal in the noiseless case then we can always reduce MSE under high noise by reducing sample size from $k$ to $k-1$ {(Section \ref{sec:LS_full_band}, Theorem \ref{thm:noiseless_optimality_means_noise_sensitivity})}. We prove that under GLR, if certain graph invariants hold on a graph with $N$ vertices, reducing sample size from almost $N$ vertices to $\mathcal{O}(\sqrt{N})$ vertices will reduce MSE at high noise levels {(Section \ref{sec:GLR_full_band}, Proposition \ref{propn:GLR_simple})}, and that these invariants hold for large Erdős–Rényi graphs with high probability {(Section \ref{sec:GLR_full_band}, Proposition \ref{propn:GLR_big_N})}. Our experiments {(Section \ref{sec:experiments})} validate this for { large} Stochastic Blockmodel and Barabasi-Albert graphs as well, { demonstrating the applicability of our results to graphs with community structure and large scale-free networks}. We also investigate how sensitive our results are to different kinds of noise by presenting variants of our theoretical and empirical results under both bandlimited (Sections \ref{sec:LS_bandlimited}) \& \ref{sec:GLR_bandlimited} and full-band noise (Sections \ref{sec:LS_full_band} \& \ref{sec:GLR_full_band}). \xd{Although our results are mainly of theoretical nature, they may provide useful guidance on choosing the appropriate sample size in light of noise levels in the data.}

Our paper presents four primary contributions: 
\begin{enumerate}
    \item A theoretical characterisation of how decreasing sample size may decrease MSE, not only under LS  but also under GLR, a regularised method. 
    \item Analysis of both LS and GLR under bandlimited noise to show non-monotonicity of the MSE is not caused by just the high frequency component of the noise.
    \item Asymptotic analysis, showing how the non-monotonicity of the MSE with sample size persists as $N \to \infty$.
    \item Extensive experimental simulations illustrating the theoretical results under LS and GLR, and bandlimited noise.
\end{enumerate}
This work significantly extends a previous conference paper \cite{sripathmanathan2023impact}, which presented preliminary versions of Corollary \ref{main_ls}, Proposition \ref{propn:main_existence_LS}, and a weaker version of Theorem \ref{thm:noiseless_optimality_means_noise_sensitivity}, corresponding to the LS part of contribution 1) above. Lemmas \ref{lemma:LS_xi_1_is_rank}--\ref{lemma:LS_delta_1_improvement_means_delta_2_worse} closely follow those in the conference version \cite{sripathmanathan2023impact}.

\section{Background and Problem Formulation}

\xd{In this section, we first introduce basic definitions of graphs, and the signal and noise models we focus on. We then provide notations for sampling, discuss reconstruction methods, and evaluation criteria. 
Finally, we present our problem setting.}

\subsection{Graphs and graph signals}
\label{sec:signal_model}
A graph $\set{G}$ consists of a set of $N$ vertices, a set of edges between these vertices, and associated edge weights. We assume $\mathcal{G}$ is connected and undirected, and that the combinatorial graph Laplacian $\matr{L}$ is real positive semidefinite with $N$ distinct eigenvalues $0 = \lambda_{1} < \lambda_{2} < \ldots < \lambda_{N}$ which are also called \emph{graph frequencies} \cite{ortega2018graph}\footnote{Although we focus on the combinatorial graph Laplacian, our results on LS also hold for the normalised graph Laplacian or any graph shift operator that is positive semidefinite.}. {Write the eigendecomposition of $\matr{L}$ as $\matr{L} = \matr{U}\matr{\Lambda}\matr{U}^{T}$ where $\matr{\Lambda} = diag(\lambda_{1},\ldots,\lambda_{N})$ and the columns of $\matr{U}$ are the eigenbasis of $\matr{L}$ and form an orthonormal basis of $\mathbb{R}^{N}$.

}
The most common signal model used in the graph signal processing literature is the bandlimited signal model, where a $k$-\emph{bandlimited signal} is a linear combination of the first $k$ columns of $\matr{U}$\cite{EOptimalChen}. This is a common smoothness model for graph signals. To obtain this, we define $\projbl$ as the projection operator that $k$-bandlimits a signal (we provide a symbolic definition in Section \ref{sec:notation}).

It is rare for observed signals to be perfectly bandlimited{. While this can be modelled by assuming the underlying signal to be reconstructed is not bandlimited but rather} `approximately bandlimited signals'\cite{chen2016signal, lin2019active}, or from other more general priors \cite{tanaka2020generalized, hara2022sampling}, we take the more common approach \cite{wang2018optimal, wang2019low,bai2020fast, puy2018random, tremblay2017determinantal} of assuming {the underlying signal is $k$-bandlimited} with additive observation noise. We assume we observe a corrupted signal $\vect{y} = \vect{x} + \vect{n}$ %
where
\begin{itemize}
    \item $\vect{x}$ is a random $k$-bandlimited signal with $\expect{\vect{x}} = 0$ and $\text{Cov}(\vect{x}) = \projbl$,
    \item $\vect{n} = \sigma \cdot \vect{\epsilon}$ is noise where $\sigma > 0$, $\expect{\vect{\epsilon}} = 0$ and either
    {%
    \begin{enumerate}
        \item $\text{Cov}(\vect{\epsilon}) = \matr{I}_{N}$ (`full-band noise')
        \item $\text{Cov}(\vect{\epsilon}) = \projbl $  (`$k$-bandlimited noise')
    \end{enumerate}
    }
\end{itemize}
{ Our results apply to any noise with these properties (e.g. Gaussian or Rademacher noise).}
We refer to the $\text{Cov}(\vect{\epsilon}) = \matr{I}_{N}$ case as `full-band noise' as the associated corrupted signal $\vect{y}$ has high frequency components, whereas the other case does not. In the literature, noise levels are often described using the SNR = $\frac{\expect{\sqnormvec{ \vect{x} }}}{\expect{\sqnormvec{\vect{n}}}}$, { a ratio of norms which must be positive\footnote{It is common in the literature to express the SNR in decibels, which may be negative, while its ratio form remains positive. We will use the ratio form unless otherwise noted, so for example $-20dB$ will be written as $10^{-20/10}  > 0$.}. We note that under full-band noise we have $\sigma^{2} = \frac{k}{N \cdot \text{SNR}}$ and under $k$-bandlimited noise we have $\sigma^{2} = \frac{1}{\text{SNR}}$.}

\subsection{Notation for sampling}
\label{sec:notation}
We use the same submatrix notation as \cite{zhang2000schur}.  For any matrix $\matr{X}$ and sets $\set{A},\set{B}$, we write $\matrsub[A,B]{X}$ to be the submatrix of $\matr{X}$ with row indices in $\set{A}$ and column indices in $\set{B}$. We define the subvector $\vectsub[A]{x}$ of a vector $\vect{x}$ similarly. We define a specific shorthand for taking a principal submatrix:
\begin{equation}
    \matrsub[A]{X} = \matrsub[A,A]{X}.
\end{equation}
We let $\set{N} = \{1, \ldots, N\}$ and $\set{K} = \{1, \ldots, k\}$. We also define two pieces of notation for projections. Let
\begin{align}
    \proj{B} &= \matrsub[N,B]{I}\matrsub[B,N]{I},   &
    \projbl &= \matrsubU{N}\matrsubU{N}^{T}.
\end{align}
Finally, $\set{A} \backslash \set{B} = \{ i \in A \, \mid \, i \notin B \} $ and $\set{A}^{c} = \set{N} \backslash \set{A}$. In general, we adhere to standard set notation.

{ We will mainly use $\set{N}$ to index the vertices of $\set{G}$ and use $\set{K}$ to index the first $k$ columns of $\matr{U}$. $\projbl$ can be understood as an ideal low-pass filter.}

\subsection{Reconstruction methods}
We define a \emph{reconstruction method} (or `interpolation operator' \cite{chamon2017greedy}) to take potentially noisy observations on a vertex sample set $\set{S}$ and reconstruct the signal across all vertices. In this paper we focus on LS and GLR, which have objectives: %
\begin{flalign}
        &\text{LS:} \hspace{0.13\columnwidth}  \hat{\vect{x}} =  \argmin_{\vect{x} \in \text{span}(\matrsubU{N})} || \vectsub[S]{x} - \vect{y} ||_{2}   & \label{eq:LS_Defn} \\ 
        &\text{GLR:} \hspace{0.1\columnwidth}  \hat{\vect{x}} = \hspace{11pt} \argmin_{\vect{x} \in \mathbb{R}^{N}}  \sqnormvec{ \vectsub[S]{x} - \vect{y}} + \mu \vect{x}^{T}\matr{L}\vect{x} & \label{eq:GLR_Defn}
\end{flalign}

\noindent We summarise their differences in Table \ref{tab:tbl1}, labelling the input parameters into the reconstruction, whether they are biased and whether they require computation of the {eigenbasis} $\matrsubU{N}$.

\renewcommand{\arraystretch}{1.3}
\begin{table}[h]
\caption{The LS and GLR reconstruction methods.}    \renewcommand*{\arraystretch}{1.3}

\begin{center}
    \begin{tabular}{|l|c|c|c|}
    \hline
      \textbf{Method} & \textbf{Parameters} & \textbf{Biased} & \textbf{Requires Computing $\matrsubU{N}$} \\
    \hline
        LS  & $k$ & $\times$ & $\checkmark$ \\
        \hline
        GLR  & $\mu$ & $\checkmark$ & $\times$ \\
        \hline
    \end{tabular}
\end{center}
\label{tab:tbl1}
\end{table}
Our analysis of LS also applies to the commonly used iterative reconstruction method, Projection onto Convex Sets \cite{narang2013localized}, as POCS converges to LS.

We call a reconstruction method \emph{linear} if it is linear in its observations. For a fixed vertex sample set $\set{S}$ we can represent a linear reconstruction method by a matrix $\matr{R}_{\set{S}} \in \mathbb{R}^{N \times |\set{S}|}$.

\begin{remark}
    LS and GLR are both linear:
    \begin{align}
    \textrm{LS:\quad} \matr{R}_{\set{S}} &= \matrsubU{N}\matrsubU{S}^{\dagger} \label{eq:defn_RS:LS}\\
    \textrm{GLR:\quad} \matr{R}_{\set{S}} &= \msub[N,S]{( \proj{S} + \mu \matr{L})^{-1}} \label{eq:defn_RS:GLR}
    \end{align}
\end{remark}
\noindent where for a matrix $\matr{A}$, $\matr{A}^{\dagger}$ is its Moore-Penrose pseudoinverse.

Across all linear models under the noisy setting, LS leads to the minimum-variance unbiased estimator of $\vect{x}$ \cite{gauss1823theoria} { for $k$-bandlimited signals and full-band noise}, which theoretically justifies us focusing our analysis on LS. 
On the other hand, GLR implicitly assumes a multivariate Gaussian signal with covariance $\matr{L}^{\dagger}$\cite{dong2016learning}, which is slightly mismatched with the $k$-bandlimited signal model we focus on in this paper. Nevertheless, in practice GLR is still often used for large graphs as computing $\matrsubU{N}$ for LS is slow \cite{puy2018random, bai2020fast}. It is also a typical regression model that promotes signal smoothness and has been widely adopted in the literature \cite{belkin2004semi}. Because of these considerations we believe GLR is still a meaningful case to study.

Finally, we clarify what we mean when we consider LS with sample size less than bandwidth. { In this case }if LS is defined as the minimisation of the objective (\ref{eq:LS_Defn}) there are multiple solutions. Thus, we follow \cite{anis2016efficient} and define the LS reconstruction as the 
unique minimum-norm solution \cite[Sect. 5.5.1]{golub13}, hence (\ref{eq:defn_RS:LS}) applies regardless of sample size. 

{%
\begin{remark}    
We note that if $\vect{x}, \vect{\epsilon}$ and $\sigma$ are Gaussian, we can compute the best predictor of $\vect{x}$ in MSE terms. This is the expectation of $\vect{x}$ conditional on the observation, and is also known as the `Optimal Bayesian Reconstruction'. It has been proven in this case that decreasing sample size must increase MSE \cite{chamon2017greedy}. We do not attempt to characterise this or other noise-adaptive methods; we assume that $\matr{R}_{\set{S}}$ is not dependent on $\sigma$.
\end{remark}
}

\subsection{Optimality criteria for sampling}
\label{sec:Optimality_Criteria}
To meaningfully contrast choices of vertex sample set size and selection, we need to evaluate reconstruction performance, and we do so by certain optimality criteria. In the noiseless case, the main optimality criterion for a vertex sample set $\set{S}$ is whether it is a \emph{uniqueness set} 
 \cite{pesenson2008sampling}, that is, if we can perfectly reconstruct any $k$-bandlimited signal observed on $\set{S}$. Such a set always exists and $\set{S}$ is a uniqueness set for a bandwidth $k$ if and only if $\text{rank}(\matrsubU{S}) = k$ \cite{anis2016efficient}. 
 
 In the case of additive observation noise, there are multiple common optimality criteria \cite[Chapter 6]{pukelsheim2006optimal}:

\begin{itemize}
    \item \emph{MMSE criterion}: Minimise  average MSE. \cite{wang2018optimal,wang2019low, mfn}
    \item \emph{Confidence Ellipsoid criterion}: Minimise the confidence ellipsoid around eigenbasis co-efficients. \cite{jayawant2021doptimal, tremblay2017determinantal, mfn}
    \item \emph{WMSE criterion}: Minimise worst-case MSE. \cite{bai2020fast, EOptimalChen}
\end{itemize}

\noindent 
Under LS, these criteria have the following names and forms:
\begin{align}
    \text{\emph{(MMSE)} A-Optimality:} &\text{ minimise } \text{tr}(\matr{P}^{-1}) \label{eq:A-optimality}\\
    \text{\emph{(Conf. Ellips.)} D-Optimality:} &\text{ maximise } \text{det}(\matr{P})
 \label{eq:D-optimality}\\
    \text{\emph{(WMSE)} E-Optimality:} &\text{ maximise }
        \lambda_{min}(\matr{P}) \label{eq:E-optimality}
\end{align}
\begin{flalign}
    &\text{where }&\matr{P} = \begin{cases}
        \projbl[S] &\text{if }|\set{S}| < k \\
        \matrsubU{S}^{T}\matrsubU{S} &\text{if }|\set{S}| \geq k
    \end{cases} \qquad\qquad&
\end{flalign}  
\noindent and we define $\trace{\matr{P}^{-1}} = +\infty$ in (\ref{eq:A-optimality}) if $\matr{P}$ is not invertible.

\subsection{Problem setting}
In this paper, we are interested in a theoretical characterisation of the impact of sample size on MSE under all possible SNRs. 
For our theoretical results and experiments, we assume:
\begin{itemize}
    \item A known graph $\mathcal{G}$ which is connected and undirected.
    \item A known bandwidth $k$.
    \item A clean underlying $k$-bandlimited signal $\vect{x}$ drawn from a known distribution.
    \item Observations of $\vect{x}$ are corrupted by noise which is either:
    \begin{itemize}
        \item full-band, so we observe a non-bandlimited signal.
        \item $k$-bandlimited, so we observe a $k$-bandlimited signal.
    \end{itemize}
    \item Linear reconstruction, in particular LS and GLR.
\end{itemize}
In this paper, we study the behaviour of the MMSE criterion, that is, the MSE averaged over a known signal model and known noise model, which we write as
\begin{equation}
    \textrm{MSE}_{\set{S}} = \expect[\vect{x},\vect{\epsilon}]{\sqnormvec{\hat{\vect{x}} - \vect{x}} \smallskip \mid \set{S} \textrm{ observed} }.
\end{equation}

\section{Main Theoretical Results}
\label{main_results_sec}
\subsection{Overview and proof approaches}
\label{sec:theory-overview}
In this section, we prove theorems showing how the relationship between sample size and MSE changes with different levels of observation noise, with a focus on showing when reducing sample size reduces MSE. We first present a high level sketch of our approach. To study the effect of noise, we perform a Bias-Variance decomposition on the MSE:
\begin{equation}
    \textrm{MSE}_{\set{S}} = \underbrace{\xi_{1}(\set{S})}_{\expect{\text{bias}^{2}}} + \underbrace{\sigma^{2} \cdot \xi_{2}(\set{S})}_{\expect{\text{variance}}} 
\end{equation}
where the bias term $\xi_{1}(\set{S}) = \text{MSE}_{\set{S}}$ when $\sigma^{2} = 0$ is the MSE attributable to reconstruction of the clean signal, and $\xi_{2}(\set{S})$ can be understood as a noise-sensitivity term (see Section \ref{sec:general} for derivations).
With this decomposition, the relationship between sample size and MSE under different levels of noise reduces to how the bias and noise-sensitivity vary with respect to sample size. 

The focus of the paper is not to characterise all the cases where decreasing sample size decreases MSE, but rather to clearly show that it does happen in a wide variety of cases. In service of this, we focus on certain broad cases that are more tractable, which we call `simplifications'. For example, we only compare a sample set $\set{S}$ to a subset of it, i.e., $\set{T} \subset \set{S}$.

Our general approach per reconstruction method is:
\begin{itemize}
    \item Choose a simplification (`simplify');
    \item Under this simplification, characterise all conditions when decreasing sample size can decrease MSE (`characterise');
    \item Show that these conditions actually happen (`existence');
    \item Study these conditions as $N \to \infty$, to prove the conditions may happen on large graphs (`asymptotics').
\end{itemize}

For LS, we simplify the problem by only considering decreasing the sample size by one at a time, which we call the `single vertex' simplification. We pay particular attention to subsets sampled by sampling schemes that are optimal in the noiseless setting (Subsection \ref{sec:LS_full_band}). For GLR, we compare observing the full graph to observing a subset of the vertices. We call this the `full observation' simplification. We focus on graphs which satisfy certain graph invariants (Subsection \ref{sec:GLR_full_band}). We justify these simplifications in the relevant subsections below.
 
We then consider reconstruction under bandlimited noise (Subsections \ref{sec:LS_bandlimited} and \ref{sec:GLR_bandlimited}) to show that the reduction in MSE from reducing sample size is not sensitive to our noise model, nor due to the high frequency component of the noise.

\xd{An overview of the four steps in our general approach together with the main results are summarised in Table \ref{tbl:general_theory}.}

\begin{table}[h]
\setlength{\tabcolsep}{4pt}
\caption{Structure of theoretical results.}
\centering
\begin{tabularx}{\linewidth}{|p{1.3cm}|X|p{1.1cm}|p{1.3cm}|p{1.1cm}|p{1.3cm}|}
\hline
 & \textbf{General} &\multicolumn{2}{c|}{\textbf{LS}} & \multicolumn{2}{c|}{\textbf{GLR}} \\
\hline
{Noise} & Any & Full-band & Bandlimited & Full-band & Bandlimited \\
\hline
\emph{Simplif- ication} & $\set{S} \supset \set{T}$ & \multicolumn{2}{c|}{$\set{S}$ vs $\set{S} \backslash \{v\}$} &  \multicolumn{2}{c|}{$\set{N}$ vs $\set{S}$}  \\
\hline
\emph{Character- isation} & Thm \ref{main_general} & Corr \ref{main_ls} & Corr \ref{corr:LS_bandlimited_noise_big_variance} & \multicolumn{2}{c|}{Corr \ref{corr:main_GLR_iff}}  \\
\hline
\emph{Existence} & & Thm \ref{thm:noiseless_optimality_means_noise_sensitivity} & Corr \ref{corr:LS_bandlimited_noise_sample_only_k} & Thm \ref{thm:main_GLR_exist} & Thm \ref{thm:main_GLR_bl} \\
\hline
\emph{Asymptotics} & & Remark \ref{rmk:LS_big_N} & Remark \ref{rmk:LS_big_N_bl} & Propn \ref{propn:GLR_big_N} & Propn \ref{propn:GLR_big_N_bl} \\
\hline
\end{tabularx}
\label{tbl:general_theory}

\end{table}

\subsection{General results}
\label{sec:general}
{%
We start with general results. All results in this Section \ref{sec:general} apply to a very general signal model; we only assume that $\vect{x}$ and $\vect{\epsilon}$ are drawn from independent distributions, that $\expect{\vect{\epsilon}} = 0$, $\expect{\sqnormvec{\vect{x}}} < \infty$ and $\expect{\sqnormvec{\vect{\epsilon}}} < \infty$ . We then specialise to the signal and noise models introduced in Section \ref{sec:signal_model}, based on which we present the analysis in Section \ref{sec:every_x} and results in Sections \ref{sec:LS_full_band}, \ref{sec:GLR_full_band}, \ref{sec:LS_bandlimited} and \ref{sec:GLR_bandlimited}.

To understand the effect of changing the sample size on MSE at different levels of noise, we use a variant of the Bias-Variance decomposition  \cite{geman1992neural} on the MSE to separate out the effect of noise. Let $\hat{\vect{x}}$ be a reconstruction of the signal $\vect{x}$, then
\begin{equation}
    \textrm{MSE}_{\set{S}} = \mathbb{E}_{\vect{x}}\Biggl[  \underbrace{\sqnormvec{\vect{x} - \expect[\vect{\epsilon}]{\hat{\vect{x}}}} \rule[-1.5ex]{0pt}{0pt} }_{\text{Bias}(\hat{\vect{x}},\vect{x})^{2}} + \underbrace{\expect[\vect{\epsilon}]{\sqnormvec{ \expect[\vect{\epsilon}]{\hat{\vect{x}}} - \hat{\vect{x}}}}}_{\text{Var}(\hat{\vect{x}})} \Biggr]\label{eq:Bias_var_decomp_general}
\end{equation}
This decomposition applies to \emph{any} reconstruction $\hat{\vect{x}}$ of $\vect{x}$.
\begin{proof}
    See Appendix \ref{app:bias-variance} or \cite[Chapter 5.4.4]{Goodfellow2016DeepLearning}.
\end{proof}
}
We consider reconstructing a signal with a linear reconstruction method. We first consider a generic $\matr{R}_{\set{S}}$. 
\noindent For linear reconstruction, we have $\hat{\vect{x}} = \matr{R}_{\set{S}}\vectsubgen[\set{S}]{\vect{x} + \sigma \cdot \vect{\epsilon}}$. By assumption, $\expect{\vect{\epsilon}} = 0$, so $\expect[\vect{\epsilon}]{\hat{\vect{x}}} = \matr{R}_{\set{S}}\vectsub[S]{x}$. Therefore
\begin{align}
    \text{Bias}(\hat{\vect{x}}, \vect{x})^{2} &= \expect[\vect{x}]{\sqnormvec{\left(\matr{I} - \matr{R}_{\set{S}}\matrsub[S,N]{I}\right)\vect{x}}} 
 = \xi_{1}(\set{S}) \label{eq:def_bias_gen}\\
    \text{Var}(\hat{\vect{x}}) &= \sigma^{2} \cdot \expect[\vect{\epsilon}]{\sqnormvec{\matr{R}_{\set{S}}\vectsub[S]{\epsilon}}} = \sigma^{2} \cdot \xi_{2}(\set{S})\label{eq:def_var_gen}
\end{align}
where $\xi_{1}(\set{S})$ and $\xi_{2}(\set{S})$ are averaged over $\vect{x}$ and $\vect{\epsilon}$, so are not random. Therefore
\begin{equation}
    \textrm{MSE}_{\set{S}} = \underbrace{\xi_{1}(\set{S})}_{\expect{\textrm{Bias}(\hat{\vect{x}},\vect{x})^{2}}} + \enspace \underbrace{\sigma^{2} \cdot \xi_{2}(\set{S})}_{\expect{\textrm{Var}(\hat{\vect{x}})}}. \label{eq:xi_decomp}
\end{equation}

\noindent We will refer to $\xi_{1}(\set{S})$ as the `bias' of $\matr{R}_{\set{S}}$ and $\xi_{2}(\set{S})$ as the `noise-sensitivity' of $\matr{R}_{\set{S}}$. 

In Statistical Learning Theory, the standard Bias-Variance decomposition is used to show how increasing the complexity of a model often increases its ability to fit the data (reducing `bias') while increasing its noise-sensitivity (increasing `variance'), and that the optimum model complexity minimising MSE balances these two components. In the rest of the paper, we will use our Bias-Variance decomposition to show that while decreasing the sample size for a reconstruction method might increase bias it can also decrease noise-sensitivity hence reducing the variance, and that the optimal sample size minimising MSE balances these two components. In one sense, this is analogous to avoiding `overfitting to noise' in machine learning, where increasing the number of parameters can increase variance more than it decreases bias.

We provide the definitions and a theoretical result to quantify this. Both the `single vertex' and `full observation' simplifications considered in the paper compare an observed set $\set{S}$ to its subset $\set{T} \subset \set{S}$ which is reflected in our definitions.

\begin{defn}
Let $\set{T} \subset \set{S}$. We say that 
\begin{align}
\renewcommand{\arraystretch}{0.95}
\!\!\!\!
    \begin{array}{l l}
        \matr{R}_{\set{T}} \text{ is \emph{less biased than }} \matr{R}_{\set{S}}  & \!\!\!\!\text{if } \xi_{1}(\set{T}) < \xi_{1}(\set{S}) \\
        \matr{R}_{\set{T}} \text{ is \emph{less noise-sensitive than }} \matr{R}_{\set{S}}  & \!\!\!\!\text{if } \xi_{2}(\set{T}) < \xi_{2}(\set{S})  
    \end{array}
\end{align}

Furthermore, we say that
\begin{align}
\!\!\!\!
\renewcommand{\arraystretch}{0.95}
    \begin{array}{l l}
        \set{T} \text{ is \emph{better than} }\set{S} &\text{if } \textrm{MSE}_{\set{T}} < \textrm{MSE}_{\set{S}} \\
        \set{T} \text{ is \emph{as good or better than} }\set{S} &\text{if } \textrm{MSE}_{\set{T}} \leq \textrm{MSE}_{\set{S}} \\
        \set{T} \text{ is \emph{worse than} }\set{S} &\text{if } \textrm{MSE}_{\set{T}} > \textrm{MSE}_{\set{S}}.
    \end{array}
\end{align}
\end{defn}

\noindent For $i \in \{1,2\}$ and $\set{T}\subset \set{S}$, let
\begin{equation}
    \Delta_i(\set{S}, \set{T}) = \xi_{i}(\set{S}) - \xi_{i}(\set{S} \backslash \set{T}). 
\end{equation}
$\Delta_{1}(\set{S},\set{T}) > 0$ means $\matr{R}_{\set{S}\backslash \set{T}}$ is less biased than $\matr{R}_{\set{S}}$ and $\Delta_{2}(\set{S},\set{T}) > 0$ means $\matr{R}_{\set{S}\backslash \set{T}}$ is less sensitive to noise than $\matr{R}_{\set{S}}$.
Then, by (\ref{eq:xi_decomp}), the change in MSE from reducing sample size is 
\begin{equation}
    \text{MSE}_{\set{S}} - \text{MSE}_{\set{S} \backslash \set{T} }
    = \Delta_{1}(\set{S},\set{T}) + \sigma^{2} \cdot \Delta_{2}(\set{S},\set{T}) \label{eq:EMSE_decomp_into_delta}
\end{equation}
so $\set{S} \backslash \set{T}$ is better than $\set{S}$ if and only if
\begin{equation}
    \Delta_{1}(\set{S},\set{T}) > - \sigma^{2} \cdot \Delta_{2}(\set{S},\set{T}). \label{eq:general:better_xi_ineq}
\end{equation}
\begin{remark}
 If either $\Delta_{1}(\set{S},\set{T})$ or $\Delta_{2}(\set{S},\set{T})$ are positive, we can always pick $\sigma^{2}$ so $\set{S} \backslash \set{T}$ is better than $\set{S}$. If both $\Delta_{1}(\set{S},\set{T})$ and $\Delta_{2}(\set{S},\set{T})$ are negative then  $\set{S} \backslash \set{T}$ is never better than $\set{S}$.   
\end{remark}

The following Theorem characterises our bias/variance trade-off by computing the noise level at which an increase in bias is outweighed by an decrease in noise-sensitivity (or vice-versa) on average.

\begin{theorem}
\label{main_general}
Assume a linear reconstruction method and consider $\set{S}\supset \set{T}$. Let 
\begin{equation}   
 \tau(\set{S}, \set{T}) = { \frac{\expect{\sqnormvec{\vect{x}}}}{\expect{\sqnormvec{\vect{\epsilon}}}} } \cdot \frac{\Delta_2(\set{S}, \set{T})}{- \Delta_1(\set{S}, \set{T})} 
\end{equation}
    then $\set{S} \backslash \set{T}$ is better than $\set{S}$ if and only if one of the following conditions is met:
\begin{subnumcases}{ \label{eq:main_thm_cond} }
       \text{SNR} < \tau(\set{S}, \set{T}) &and $\Delta_{1}(\set{S},\set{T}) < 0$ \label{eq:main_thm_cond:d1neg}\\
       \text{SNR} > \tau(\set{S}, \set{T}) &and $\Delta_{1}(\set{S},\set{T}) > 0$ \label{eq:main_thm_cond:d1pos} \\
    0 < \Delta_2(\set{S},\set{T}) &and $ \Delta_{1}(\set{S},\set{T}) = 0$. \label{eq:main_thm_cond:d1zero}
\end{subnumcases}

\end{theorem}
\begin{proof}[Proof Sketch]
    We first get
    $
    { \frac{\expect{\sqnormvec{\vect{x}}}}{\expect{\sqnormvec{\vect{\epsilon}}}} } \cdot
        \Delta_{2}(\set{S},\set{T}) > -\Delta_{1}(\set{S},\set{T}) \cdot \textrm{SNR} 
    $ from (\ref{eq:general:better_xi_ineq}) 
    and then divide by $-\Delta_{1}(\set{S},\set{T})$, case-splitting on its different possible signs. See Appendix \ref{app:main_thm_proof} for a full proof.
\end{proof}

We interpret each of these conditions as follows:
\begin{enumerate}
    \item[(\ref{eq:main_thm_cond:d1neg})]
    If $\matr{R}_{\set{S} \backslash \set{T}}$ is more biased than $\matr{R}_{\set{S}}$, then $\set{S} \backslash \set{T}$ is better than $\set{S}$ if SNR is low enough (below a threshold $\tau$).
    \item[(\ref{eq:main_thm_cond:d1pos})] 
    If $\matr{R}_{\set{S} \backslash \set{T}}$ is less biased than $\matr{R}_{\set{S}}$, then $\set{S} \backslash \set{T}$ is better than $\set{S}$ if SNR is \emph{high} enough (above a threshold $\tau$).
    \item[(\ref{eq:main_thm_cond:d1zero})]
    If $\matr{R}_{\set{S}\backslash\set{T}}$ and $\matr{R}_{\set{S}}$ are equally biased and $\matr{R}_{\set{S} \backslash \set{T}}$ is less noise-sensitive than $\matr{R}_{\set{S}}$, then $\set{S} \backslash \set{T}$ is better than $\set{S}$ at every non-zero noise level.
\end{enumerate}
Theorem \ref{main_general} does not guarantee that any of the conditions for $\set{S} \backslash \set{T}$ being better than $\set{S}$ would happen; for example, $\matr{R}_{\set{S}}$ could be both less biased and less sensitive to noise than $\matr{R}_{\set{S} \backslash \set{T}}$. To show reducing sample size can reduce MSE, i.e. $\Delta_{1}$ and $\Delta_{2}$ are not always both non-positive, we look at specific reconstruction methods below.

{%

To show that the conditions in Theorem \ref{main_general} happen, we focus on showing that $\Delta_{2}$ can be positive, i.e., decreasing sample size can decrease noise-sensitivity of the reconstruction method.
This focus on the noise model means results about $\Delta_{2}$ under one signal model imply results for \emph{every} finite variance signal model, even if the signal model is not bandlimited:

\begin{propn}
\label{propn:averages_generalise_to_forall}
    Fix a mean zero noise model. Assume that $\matr{R}_{\set{S}}$ is not a funtion of $\sigma$ and suppose that we prove that $\Delta_{2}(\set{S},\set{T}) > 0$ under a specific signal model. If instead $\vect{x}$ is drawn from \emph{any} signal model where the signal is independent to the noise and  $\expect{\sqnormvec{\vect{x}}} < \infty$, then there exists a threshold $\tau'(\set{S},\set{T}) > 0$ (which is dependent on the choice of signal model) such that if
    \begin{equation}
        \textrm{SNR} < \tau'(\set{S},\set{T})
    \end{equation}
    then under that new signal model,
    \begin{equation}
        \textrm{MSE}_{\set{S} \backslash \set{T}}  < \textrm{MSE}_{\set{S}}.
    \end{equation}
\end{propn}
\begin{proof}
    See Appendix \ref{app:every_x}.
\end{proof}  

\noindent This includes the signal model where $\vect{x}$ is some fixed signal. While Proposition \ref{propn:averages_generalise_to_forall} means that our approach applies to any signal model, any result computing $\tau$ is specific to the chosen signal model. This is the focus of the rest of the section.

\subsection{Specialisation to signal and noise models}
\label{sec:every_x}
We now move to our specific signal and noise models introduced in Section \ref{sec:signal_model}. Firstly, we note that in Theorem \ref{main_general}, $\frac{\expect{\sqnormvec{\vect{x}}}}{\expect{\sqnormvec{\vect{\epsilon}}}}$
 is $\frac{k}{N}$ for full-band noise and $1$ for bandlimited noise.

For any zero-mean random vector $\vect{g}$ and compatible matrix $\matr{A}$, write $\text{Cov}(\vect{g}) = \matr{X}\matr{X}^{T}$. Using 
\begin{align}
    \expect[\vect{g}]{\sqnormvec{\matr{A}\vect{g}}} &= \expect[\vect{g}]{\trace{\matr{A}\vect{g}\vect{g}^{T}\matr{A}^{T}}} \\
    &= \trace{\matr{A}\text{Cov}(\vect{g})\matr{A}^{T}} = \sqfrob{\matr{A}\matr{X}}
\end{align} 
gives
\begin{align}
    \xi_{1}(\set{S}) &= \sqfrob{\matrsubU{N} - \matr{R}_{\set{S}}\matrsubU{S}} \label{eq:xi_1_def} \\
    \xi_{2}(\set{S}) &= \begin{cases}
        \sqfrob{\matr{R}_{\set{S}}} &\text{if full-band noise} \label{eq:xi_2_def} \\
\sqfrob{\matr{R}_{\set{S}}\matrsubU{S}} &\text{if }k\textrm{-bandlimited noise} 
    \end{cases}
\end{align}

We note that given Theorem \ref{main_general} and these terms, for any given reconstruction method and two sets $\set{S} \supset \set{T}$ we can now analytically compute if a given $\set{T}$ is better than a set $\set{S}$ under our signal model, and at what noise levels. The rest of the paper focuses on how to find such sets.
}

To show when $\Delta_{2} > 0$, we first consider a `single vertex' simplification where $\set{T} = \{v\}$. This is the approach we take for LS, and also motivates the `full observation' simplification for GLR. For a vertex $v$, as shorthand we write:
\begin{align}
    \tau(\set{S},v) = \tau(\set{S},\{v\})   
    , \quad \Delta_{i}(\set{S},v) = \Delta_{i}(\set{S},\{v\}). 
\end{align}
We start by looking at possible combinations of signs of $\Delta_{1}(\set{S},v)$ and $\Delta_{2}(\set{S},v)$, which we present in Table \ref{tbl:Delta_Behaviour}.

\begin{table}[h]
\setlength{\belowcaptionskip}{3pt}
\caption{Possible signs of $\Delta_{1}(\set{S},v)$ and $\Delta_{2}(\set{S},v)$.}
    \begin{subtable}[h]{0.48\linewidth}
        \begin{tabular}{| l | c | c |}
        \hline
         & \thead{$\Delta_{2}>0$} & $\Delta_{2} \leq 0$ \\
        \hline
        $\Delta_{1} > 0$ & $\times$ & $\times$\\
        $\Delta_{1} < 0$ & $\checkmark$ & $\times$\\
        $\Delta_{1} = 0$ & $\times$ & $\checkmark$\\ \hline
       \end{tabular}
       \caption{LS Reconstruction}
       \label{tbl:LS_Reconstruction_Delta}
    \end{subtable}
    \hspace{\fill}
    \begin{subtable}[h]{0.48\linewidth}
        \begin{tabular}{| l | c | c |}
        \hline
         & \thead{$\Delta_{2}>0$} & \thead{$\Delta_{2} \leq 0$} \\
        \hline
        $\Delta_{1} > 0$ & $\checkmark$ & $\checkmark$\\
        $\Delta_{1} < 0$ & $\checkmark$ & $\checkmark$\\
        $\Delta_{1} = 0$ & $\sim$ & $\sim$\\ \hline
       \end{tabular}
        \caption{GLR Reconstruction}
    \label{tbl:GLR_Reconstruction_Delta}
     \end{subtable}
     \label{tbl:Delta_Behaviour}
\end{table}
\noindent In Table \ref{tbl:Delta_Behaviour}, $\times$ will not happen, $\checkmark$ may happen and $\sim$ may theoretically happen but are unlikely. A more detailed explanation of the two tables, and proofs of the `$\times$' subcases, are presented in Appendix \ref{app:table_delta_proof}. We now discuss LS and GLR separately in the following sections.

\subsection{LS with full-band noise}
\label{sec:LS_full_band}
In this subsection we show how decreasing sample size by one can decrease MSE under LS. From Table \ref{tbl:LS_Reconstruction_Delta}, we see that reducing sample size never reduces bias under LS, so we focus on when reducing sample size reduces noise-sensitivity.  %

\subsubsection{Overview and simplification}
Our approach in this subsection is as follows. We consider the `single vertex' simplification. For a sample set $\set{S}$ and $v\in\set{S}$, we first characterise under exactly what conditions $\set{S} \backslash v$ is better than $\set{S}$ (Corollary \ref{main_ls}). We then show that the conditions must occur under sampling schemes which are optimal in the noiseless case (Theorem \ref{thm:noiseless_optimality_means_noise_sensitivity}). Finally, we comment on how the conditions persist as $N \to \infty$.

\subsubsection{Characterisation}
By Table \ref{tbl:LS_Reconstruction_Delta} we can eliminate conditions (\ref{eq:main_thm_cond:d1pos}) and (\ref{eq:main_thm_cond:d1zero}) in Theorem \ref{main_general} for the single-vertex simplification under LS.

We simplify Theorem \ref{main_general} to the following:

\begin{corollary}
\label{main_ls}
    Assume LS reconstruction. Then 
    \begin{equation}
        \tau(\set{S},v) = \frac{k}{N} \cdot \Delta_{2}(\set{S},v) \label{eq:main_ls_simplified_tau}
    \end{equation}
    and  $\set{S} \backslash \{v\}$ is better than $\set{S}$ if and only if
 \begin{equation}
 \label{eq:LS_corol_iff}
     \textrm{SNR} < \tau(\set{S},v).
 \end{equation}
 \end{corollary}
\begin{proof}[Proof Sketch]  Under LS, $\Delta_{1}(\set{S},v)$ can only be $-1$ or $0$. This simplifies $\tau$ in Theorem \ref{main_general} to (\ref{eq:main_ls_simplified_tau}):
\newline
    \emph{If $\Delta_{1}(\set{S},v) = 0$: }
    We have $\Delta_{2}(\set{S},v) \leq 0$ by Table \ref{tbl:LS_Reconstruction_Delta}. By Theorem \ref{main_general}, $\set{S} \backslash \{v\}$ is never better than $\set{S}$.
    
    \noindent\emph{If $\Delta_{1}(\set{S},v) = -1$: }
    We have $\Delta_{2}(\set{S},v) > 0$ by Table \ref{tbl:LS_Reconstruction_Delta} so $\tau(\set{S},v) > 0$. Theorem \ref{main_general}'s conditions reduce to this case.
    
\noindent See Appendix \ref{proof_appendix_LS} for a full proof.
\end{proof}

{ We note that the definition of $\tau$ in Corollary \ref{main_ls} is not completely equivalent to the definition of $\tau$ in Theorem \ref{main_general} -- specifically, if $\Delta_{1} = 0$ then $\tau$ in Theorem \ref{main_general} is undefined, and $\tau$ in Corollary \ref{main_ls} is negative. We make this change for ease of presentation -- with this new definition $\tau(\set{S},v) > 0$ if and only if there is a noise level where $\set{S} \backslash \{v\}$ is better than $\set{S}$. }

Corollary \ref{main_ls} says that if SNR is too low (below a threshold $\tau$ that depends on the bandwidth and the chosen samples), then using a smaller sample set improves the average reconstruction error. Note that we have not yet proven that condition (\ref{eq:LS_corol_iff}) is ever satisfied, i.e., reducing sample size can reduce MSE. We outline why it is not immediately obvious that (\ref{eq:LS_corol_iff}) can hold, and then show situations where it does hold and thus reducing sample size will reduce MSE.

\begin{remark}
    If $\Delta_{2} \leq 0$ , we have $\tau(\set{S},v) \leq 0 < $ SNR, so (\ref{eq:LS_corol_iff}) cannot hold and so $\set{S} \backslash \{v\}$ is never better than $\set{S}$ for any SNR.
\end{remark}

\subsubsection{Existence}
We first note that Corollary \ref{main_ls} { as stated} leaves room for a clever sampling scheme which picks $\set{S}_{i}$ where $\tau(\set{S}_{i},v)$ is always non-positive and so condition (\ref{eq:LS_corol_iff}) never holds, and hence $\set{S} \backslash \{v\}$ would never be better than $\set{S}$ for any $v \in \set{S}$. 

{ Corollary \ref{main_ls} uses that $\Delta_{2}(\set{S},v) > 0$ if and only if $ \rank{\matrsubUGen{\set{S}}} = 1 + \rank{\matrsubUGen{\set{S} \backslash \{v\}}}$; this condition can be understood as observing $v$ being informative in noise-free reconstruction. Note that the addition of a row to a matrix can only change the rank by $0$ or $1$. This equivalence also means that the `clever sampling scheme' must never increase said rank as it picks nodes, which is impossible -- else we could pick $\set{N}$ sequentially s.t. $\rank{\matrsubU{N}} = 0$. We now spell out these consequences formally and contextualise them. }

Most sampling schemes in the literature  construct sample sets by adding vertices one-by-one (e.g. greedy schemes, which are `near-optimal' \cite{chamon2016near}). We call such schemes `sequential'. We now show that for any graph, sequential sampling schemes always eventually add a vertex where (\ref{eq:LS_corol_iff}) can be satisfied.

\begin{propn}
   \label{propn:main_existence_LS} 
   Consider a sequential sampling scheme that constructs sample sets $\set{S}_{1}, \ldots, \set{S}_{N}$ where $\set{S}_{i} = \set{S}_{i-1} \cup \{ v_{i} \}$. Then there are exactly $k$ indices $1 \leq I_{1}, \dots, I_{k} \leq N$ where
   \begin{equation}
        \forall 1\leq j \leq k: \tau(\set{S}_{I_{j}}, v_{I_{j}}) > 0,
    \end{equation}
    and so $\set{S}_{I_{j}} \backslash \{v_{I_{j}}\}$ is better than $\set{S}_{I_{j}}$ at some SNR.
\end{propn}
\begin{proof}[Proof Sketch]
    By Table \ref{tbl:LS_Reconstruction_Delta}, $\tau \propto \Delta_{2} > 0 \iff \Delta_{1} < 0$. As $\Delta_{1}(\set{S},v) = \xi_{1}(\set{S}) - \xi_{1}(\set{S}\backslash \{v\}) \in \{0,-1\}$ under LS and $\xi_{1}(\emptyset) = k$ and $\xi_{1}(\set{N}) = 0$, we have that $\Delta_{2} > 0$ exactly $k$ times. See Appendix \ref{app:LS_satisfied} for a full proof.
\end{proof}
    Proposition \ref{propn:main_existence_LS} still leaves room for a hypothetical sequential sampling scheme where $\tau(\set{S},v) > 0$ only for the last $k$ chosen vertices, and therefore if such a scheme selects $|\set{S}| \leq N-k$ then $\set{T} \subset \set{S}$ is never better than $\set{S}$. We now show that there is a trade-off between such a property and performance in the noiseless case, namely that
    any scheme which is optimal in the noiseless case, like most deterministic schemes in the literature, could not have this property. We first define optimality in the noiseless case.  %
    \begin{defn}
        A sampling scheme is \emph{noiseless-optimal} for LS reconstruction of $k$-bandlimited signals if the first $k$ vertices it samples form a uniqueness set. That is, it finds the smallest possible uniqueness set  \cite{shomorony2014sampling}.
    \end{defn}

\begin{remark}
\label{remark:ADE_are_noiseless_optimal}
     A, D and E-optimal sampling are noiseless-optimal (see Appendix \ref{app:proof_of_remark_ADE_are_noiseless_optimal} for a proof). \end{remark} 
     
We now show such schemes find $\tau(\set{S},v) > 0$ for the first $k$ vertices they pick.
   
\begin{theorem}
\label{thm:noiseless_optimality_means_noise_sensitivity}
    Suppose we use a sequential noiseless-optimal scheme to select a vertex sample set $\set{S}_{m}$ of size $m$. For $m \leq k$:
    \begin{equation}
    \label{eq:greedy_sampling_first_k}
        \forall v \in \set{S}_{m}: \enskip \tau(\set{S}_{m},v) \geq \frac{k}{N},
    \end{equation}
    i.e., for \emph{any} vertex $v \in \set{S}_{m}$, if $\text{SNR} < \tau(\set{S}_{m},v)$ (which always holds if $\text{SNR} < \frac{k}{N}$) then $\set{S} \backslash \{v\}$ is better than $\set{S}$. For $m > k$:
    \begin{equation}
        \label{eq:greedy_sampling_over_k}
        \forall v_{+} \in \set{S}_{m} \backslash \set{S}_{k}: \enskip \tau(\set{S}_{m},v_{+}) \leq 0.
    \end{equation}
    That is, $\set{S}_{m}\backslash \{v_{+}\}$ is never better than $\set{S}_{m}$ for any of the vertices $v_{+}$ the sampling scheme adds beyond the first $k$.
\end{theorem}

\begin{proof}[Proof Sketch]
For (\ref{eq:greedy_sampling_first_k}), $\projblgen[\set{S}_{m}]$ and $\projblgen[\set{S}_{m-1}]$ are full-rank so $\Delta_{1} = -1$; we then show $\Delta_{2} \geq 1$. For (\ref{eq:greedy_sampling_over_k}), as $m>k$, $\xi_{1}(\set{S}_{k}) = \xi_{1}(\set{S}_{m})$. See Appendix \ref{app:optimal_noiseless_schemes_immediately_satisfy} for a full proof.
\end{proof}

 Theorem \ref{thm:noiseless_optimality_means_noise_sensitivity} explicitly demonstrates a case where reducing sample size reduces MSE: namely if $\text{SNR} \leq \frac{k}{N}$, when the sample is picked by an A-, D- or E-optimal scheme, and $|\set{S}| \leq k$, then reducing sample size will always reduce MSE.  On the other hand, for $|\set{S}| > k$ we transition into a regime where reducing sample size \emph{cannot} reduce MSE as long as the smaller set also has size $\geq k$.

Theorem \ref{thm:noiseless_optimality_means_noise_sensitivity} also gives us the shape of the MSE curve as sample size varies at high noise levels.  At high noise levels, $\text{MSE}_{\set{S}} \approx \sigma^{2} \cdot \xi_{2}(\set{S})$ so  $\text{MSE}_{\set{S}} - \text{MSE}_{\set{S}\backslash \{v\}} \approx \sigma^{2} \cdot \Delta_{2} = \sigma^{2} \cdot \frac{k}{N} \tau$ by Corollary \ref{main_ls}. Theorem \ref{thm:noiseless_optimality_means_noise_sensitivity} gives the sign of $\tau$, which therefore shows whether the MSE increases or decreases with sample size. %

\begin{remark}
\label{remark:LS_error_lambda_shaped}
Using a sequential noiseless-optimal sampling scheme, such as greedy A-, D- or E- optimal sampling, under sufficiently high noise leads to the MSE being $\Lambda$-shaped with regards to sample size with a peak at $|\set{S}| = k$.
\end{remark}

The intuition behind a $\Lambda$-shaped MSE under noiseless-optimal sampling is as follows. When the sample size is below $k$ and we increase it, we infer more about the signal---which can be seen in the noiseless case, as each of these samples improves our prediction---and inevitably incorporate some of the noise into our reconstruction. Beyond the first $k$, samples provide no new information about the signal in the noiseless case. These `additional samples' (corresponding to (\ref{eq:greedy_sampling_over_k})) are much like getting repetitive observations of already-seen vertices, which we can average to reduce the effect of noise. This is what causes the transition at $|\set{S}|=k$.
    
\begin{remark}
\label{remark:remove_multiple_nodes}
    If the MSE is $\Lambda$-shaped, then even if removing one vertex does not improve $\set{S}$, removing multiple vertices might decrease MSE. This happens if we reduce sample size enough to transition from the right side of the peak to the left side of the peak of $\Lambda$.
\end{remark}

\subsubsection{Asymptotics}
Finally, we discuss the large graphs case. 
{%
\begin{remark}
\label{rmk:LS_big_N}
If $\frac{k}{N}$ is fixed, the lower bound on $\tau$ in Theorem \ref{thm:noiseless_optimality_means_noise_sensitivity} is constant. Therefore, at a fixed SNR (e.g., $\frac{k}{2N}$), decreasing sample size decreases MSE on arbitrarily large graphs.
\end{remark}}

\subsection{LS with k-bandlimited noise}
\label{sec:LS_bandlimited}
We have observed that MSE can decrease when sample size decreases under LS and full band noise. This raises the question of whether the decrease is caused by some sort of interference effect between the high-frequency components of the noise and the bandlimited (low-frequency) signal. 
In Appendix \ref{app:LS_bandlimited} we show that MSE can decrease when sample size decreases under LS with $k$-bandlimited noise, demonstrating that this is not the case.

\subsection{GLR with full-band noise}
\label{sec:GLR_full_band}
In this subsection, we show how decreasing sample size can decrease MSE under GLR reconstruction and full-band noise.

\subsubsection{Overview and simplification}
We start by trying to simplify Theorem \ref{main_general}. Table \ref{tbl:GLR_Reconstruction_Delta} contains no $\times$ scenarios and so the `single vertex' simplification cannot eliminate any conditions in Theorem \ref{main_general}: surprisingly, $\matr{R}_{\set{S} \backslash \{ v \}}$ can be \emph{less} biased than $\matr{R}_{\set{S}}$ for GLR, which can be observed experimentally.
Instead, as we focus on tractability and showing that it is possible to reduce sample size to reduce MSE, rather than fully characterizing all such cases, we pick a situation where $\Delta_{1} \geq 0$ so we can simplify Theorem \ref{main_general}. Specifically, we compare the full observation set $\set{S}=\set{N}$ to a subset of it, which we call the `full observation' simplification. As it is hard to interpret what reconstruction means under full observation \cite{chen2017GLRbias}, our results should be understood as approximately showing that reducing the sample size from nearly full observation to some smaller size may reduce MSE.

Our approach is then as follows. We first characterise under exactly which conditions a sample set $\set{S} \subset \set{N}$ is better than $\set{N}$ (Corollary \ref{corr:main_GLR_iff}). We then show that these conditions must occur if certain graph invariants hold (Theorem \ref{thm:main_GLR_exist}). Finally, we analyse the parameters in Theorem \ref{thm:main_GLR_exist} to obtain a `suggested sample size' (Remark \ref{remark:mopt}) and show the conditions still occur as $N \to \infty$ (Proposition \ref{propn:GLR_big_N}). 

\subsubsection{Characterisation}
We first present the following Corollary of Theorem \ref{main_general}.

\begin{corollary}
    \label{corr:main_GLR_iff}
    Assume GLR and that $k > 1$. Consider a non-empty sample set $\set{S} \subset \set{N}$.  Then 
    \begin{equation}
        \tau(\set{N},\set{S}^{c}) = \frac{k}{N} \cdot \frac{\Delta_{2}(\set{N},\set{S}^{c})}{-\Delta_{1}(\set{N},\set{S}^{c})} \label{eq:main_GLR_simplified_tau}
    \end{equation}
    and $\set{S}$ is better than $\set{N}$ if and only if one of the following conditions is met:
        \begin{subnumcases}{ \label{eq:GLR_corol_iff_overall} }
        \text{SNR} < \tau(\set{N}, \set{S}^{c}) &and $ \msubgen[\set{S}^{c},\{2,\ldots,k\}]{\matr{U}} \neq \matr{0}$ \label{eq:GLR_corol_iff}\\
        0 < \Delta_2(\set{N},\set{S}^{c}) &and $ \msubgen[\set{S}^{c},\{2,\ldots,k\}]{\matr{U}} = \matr{0}$. \label{eq:GLR_corol_iff_weird}
\end{subnumcases}
where $\msubgen[\set{S}^{c}, \{2,\ldots,k\}]{\matr{U}} = \vect{0}$ corresponds to any $k$-bandlimited signal always being constant on all of $\set{S}^{c}$.
\end{corollary}
\begin{proof}[Proof Sketch]
    We use Cauchy-Schwartz and Lemma \ref{lemma:GLR_full_observation_MSE} to lower bound $\xi_{1}(\set{S})$ and show that either all $k$ columns of $\matrsubU{N}$ are eigenvectors of $\proj{S} + \mu\matr{L}$ (which is exactly when $\msubgen[\set{S}^{c},\{2,\ldots,k\}]{\matr{U}} = \matr{0}$) and $\Delta_{1}(\set{N},\set{S}^{c}) = 0$, corresponding to (\ref{eq:GLR_corol_iff_weird}), or $\Delta_{1}(\set{N},\set{S}^{c}) < 0$, corresponding to (\ref{eq:GLR_corol_iff}). We then apply Theorem \ref{main_general}. See Appendix \ref{app:proof_main_GLR_iff} for a full proof.
\end{proof}

We now explain the conditions in Corollary \ref{corr:main_GLR_iff}. Condition (\ref{eq:GLR_corol_iff}) corresponds to the single case we see in Corollary \ref{main_ls}. Condition (\ref{eq:GLR_corol_iff_weird}) is more of an edge case,  e.g., if $\lambda_{k} < N$ and every vertex in $\set{S}^{c}$ has degree $N-1$ \cite[Corollary 2.3]{merris1998laplacian}. In (\ref{eq:GLR_corol_iff_weird}), $\msubgen[\set{S}^{c}, \{2,\ldots,k\}]{\matr{U}} = \vect{0}$ means any $k$-bandlimited signal will be constant on all of $\set{S}^{c}$. Our proof shows that in the noiseless case this implies that observing $\set{S}^{c}$ will not improve the MSE, i.e., $\text{MSE}_{\set{S} \cup \set{T}_{c}} = \text{MSE}_{\set{S}}$ for any $\set{T}_{c} \subseteq \set{S}^{c}$. %
The other condition in (\ref{eq:GLR_corol_iff_weird}), i.e., $0 < \Delta_2(\set{N},\set{S}^{c})$, corresponds to an increase in noise-sensitivity from reconstructing from those additional vertices. Therefore, condition (\ref{eq:GLR_corol_iff_weird}) says that if $\set{S}^{c}$ reveals nothing new about the underlying signal and also makes the reconstruction more sensitive to noise, one should not observe $\set{S}^{c}$ and only observe $\set{S}$.

\subsubsection{Existence}
Like Corollary \ref{main_ls}, Corollary $\ref{corr:main_GLR_iff}$ does not show that any set $\set{S}$ is ever better than $\set{N}$, i.e., that $\tau(\set{N},\set{S}^{c}) > 0$ ever happens. { As described in Section \ref{sec:every_x}, we will find $\set{S}$ s.t. $\Delta_{2}(\set{N},\set{S}^{c}) > 0$. We illustrate our proof method in Fig.~\ref{fig:GLR_diagram}.
\begin{figure}[t]
\setlength{\abovecaptionskip}{-2pt}
    \centering
    \resizebox{0.75\width}{0.75\height}
    {\includegraphics[width=\columnwidth]{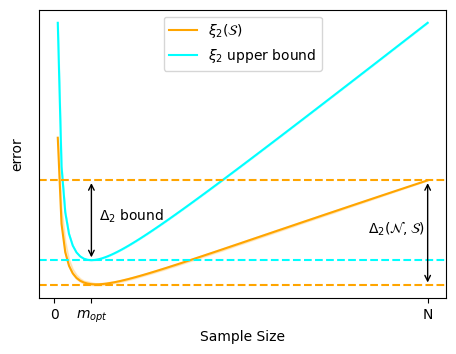}}
    \caption{Proof direction to show $\Delta_{2}(\set{N},\set{S}^{c}) > 0$.}
    \label{fig:GLR_diagram}
\end{figure}
\noindent Fundamentally, our approach rests on the fact that $\xi_{2}$, which is approximately proportional to $\text{MSE}_{\set{S}}$ at high noise levels, is  approximately U-shaped, and we can use an upper bound to approximate the minimum and thus bound $\Delta_{2}(\set{N},\set{S}^{c})$ below. Explicitly, our method to show $\Delta_{2}(\set{N},\set{S}^{c}) > 0$ is:
\begin{enumerate}
    \item Find an upper bound $B(\cdot)$ for $\xi_{2}(\set{S})$ that only depends on $\matr{L}$ and $|\set{S}|$, not on $\mu$ or the choice of $\set{S}$.
    \item Find the sample size $m_{opt}$ that minimises the bound and show $m_{opt} < N$.
    \item Identify when $\xi_{2}(\set{N}) - B(m_{opt}) > 0$ (this is the `$\Delta_{2}$ bound' in Fig. \ref{fig:GLR_diagram}). As this quantity is a lower bound for $\Delta_{2}(\set{N},\set{S}^{c})$ this shows that $\Delta_{2}(\set{N},\set{S}^{c})>0$ for any $\set{S}$ of size $m_{opt}$.
    \item Bound $\Delta_{1}$ using the same tools and use these to lower bound $\tau(\set{N}, \set{S}^{c})$
\end{enumerate}
We see that $m_{opt}$ suggests an `optimal sample size', in the sense that the $\xi_{2}$ upper bound approximates $\xi_{2}$, and we can optimise it as a surrogate. Following step 1, we now bound $\xi_{2}(\set{S})$.
\begin{lemma}
\label{lemma:GLR_xi_2_bound_main}
    Let $\lambda_{i}$ be the eigenvalues of $\matr{L}$, the combinatorial Laplacian. Let $|\set{S}| = m$. Then
    \begin{align}
        \xi_{2}(\set{S}) &\leq r\frac{N}{m} + \sum_{i=2}^{m} \omega\left(\max\left[1,\frac{\lambda_{N+2-i}}{\lambda_{i}}\right]\right)  \label{eq:xi_2_GLR_bound_strong}
        \\ 
        &\leq r\frac{N}{m} + r(m -1) . \label{eq:xi_2_GLR_bound_weak}
    \end{align}
    where we define $B(m)$ to be the RHS of (\ref{eq:xi_2_GLR_bound_strong}) and where
    \begin{align}
        \omega(x) &= \frac{1}{4}\left(\sqrt{x}+ \frac{1}{\sqrt{x}} \right)^{2}, \quad
        r = \omega\left(\frac{\lambda_{N}}{\lambda_{2}}\right)
    \end{align}
   and the summation in (\ref{eq:xi_2_GLR_bound_strong}) is 0 if $m = 0$ or $1$. 
    
\end{lemma}
\begin{proof}
    Decompose $\matr{R}_{\set{S}}$ and apply Kantorovich's inequality. See Appendix \ref{app:proof_unif_ub_xi_2}.
\end{proof}

The function $\omega$ arises from the Kantorovich inequality, commonly used in variance bounds \cite{khatri1982some, householder1965kantorovich}. $\omega$ is an increasing function and for $x>1, \omega(x) \leq x$; this means $r$ increases with $\frac{\lambda_{N}}{\lambda_{2}}$. We discuss $r$ in more detail in Appendix \ref{app:GLR_sensitivity}.

As mentioned, in step 4 we need a bound on $\xi_{1}(\set{S})$ to lower bound $\tau(\set{N},\set{S}^{c})$. The shape and accuracy of this bound do not affect the core of our argument, only how tight our approximation of $\tau$ is. Therefore, we leave it in Appendix \ref{app:proof_unif_ub_xi_1_MSE}; the term $B_{k}(m)$ it introduces in the proof is a tightening of $B(m)$ ($B_{k}(m) < B(m)$) and is explained in Appendix  \ref{app:GLR_bandlimited}.

By steps 2--4, using the bound $B(m)$ from Lemma \ref{lemma:GLR_xi_2_bound_main} yields:

\begin{theorem}
\label{thm:main_GLR_exist}
{ Let $B(m)$ be defined as in Lemma \ref{lemma:GLR_xi_2_bound_main} and let $m_{opt}$ be the sample size minimising $B$.
    \begin{flalign}
        &\text{If} \hspace{0.35\columnwidth} B(m_{opt}) < N& \label{eq:GLR_exist_thm_B_constraint}
    \end{flalign}
then $m_{opt} < \lceil\frac{N+1}{2}\rceil$. Furthermore, $\exists \mu_{ub} > 0$ and $\tau_{GLR}(\mu) > 0$ s.t. under GLR with parameter $\mu \in (0, \mu_{ub})$,
\begin{flalign}
        &\text{if} \hspace{0.3\columnwidth} \textrm{SNR} < \tau_{GLR}(\mu)&
    \end{flalign}

    \noindent then,  \emph{any} sample set $\set{S}$ of size $m_{opt}$ is better than $\set{N}$.}
\end{theorem}
\begin{proof}[Proof Sketch] 
We prove that for any $\mu$, $\xi_{1}(\set{S}) \leq k + B_{k}(m)$ and $\xi_{2}(\set{S}) \leq B(m)$ via Lemmas \ref{lemma:GLR_xi_2_bound_main} and \ref{lemma:GLR_xi_2_bound_main_bl}. We also show $\xi_{1}(\set{N}) = \sum_{i=1}^{k} \left(1 - \frac{1}{1+\mu\lambda_{i}}\right)^{2}$ and $\xi_{2}(\set{N}) = \sum_{i=1}^{N} \left({1+\mu\lambda_{i}}\right)^{-2}$.
We use these to bound $\Delta_{i}$, and finally apply Corollary \ref{corr:main_GLR_iff} as our conditions mean $\Delta_{2}(\set{N},\set{S}^{c}) > 0$.  See Appendix \ref{app:Proof_thm_main_GLR_exist} for a full proof.
\end{proof}
}
 We contrast Theorems \ref{thm:noiseless_optimality_means_noise_sensitivity} and \ref{thm:main_GLR_exist}. Theorem \ref{thm:noiseless_optimality_means_noise_sensitivity} provides necessary and sufficient conditions. Theorem  \ref{thm:main_GLR_exist}, while still useful, only provides sufficient conditions. 
 Theorem \ref{thm:noiseless_optimality_means_noise_sensitivity} applies to any graph, but only to sample sets chosen under noiseless-optimal sampling schemes, while Theorem \ref{thm:main_GLR_exist} has no requirement on sampling schemes, but only applies to graphs which fulfill certain graph invariants. %

{%
Because of the difficulty in writing down an analytic form for $m_{opt}$, the parameters in Theorem \ref{thm:main_GLR_exist} -- $m_{opt}$, $\mu_{ub}$ and $\tau_{GLR}$ -- are difficult to analyse. To better understand  we present a version of Theorem \ref{thm:main_GLR_exist} instead based on  
the weaker bound (\ref{eq:xi_2_GLR_bound_weak}). This bound is a linear transformation of $\frac{N}{m} + m$ and so is clearly U-shaped and has an minimum at $m=\sqrt{N}$, and thus following steps 2, 3 and 4 yields the following proposition:

\begin{propn}
\label{propn:GLR_simple}
    Let $\bar{\lambda} = \frac{1}{N}\sum_{i=1}^{N}\lambda_{i}$ be the mean eigenvalue of $\matr{L}$ and let $r$ be defined as in Lemma \ref{lemma:GLR_xi_2_bound_main}.
    \begin{flalign}
    &\text{If} \hspace{0.35\columnwidth} 2r\sqrt{N} < N,& \label{eq:weak_GLR_constraint}\\[0.25em]
    &\text{then let} \hspace{0.12\columnwidth} \mu_{ub\_weak} = \bar{\lambda}^{-1}\left(\sqrt[4]{N}\cdot ({2r})^{-\frac{1}{2}} - 1 \right) \\
    &\text{and} \hspace{0.15\columnwidth} \tau_{GLR\_weak} =  \frac{\sqrt{N}(1 + \mu\bar{\lambda})^{-2} - 2r}{\sqrt{N} + 2r \cdot \frac{N}{k} }
\end{flalign}
    then under GLR with parameter $\mu \in (0,\mu_{ub\_weak})$, if
    \begin{equation}
        \text{SNR} < \tau_{GLR\_weak}(\mu)
    \end{equation}
    then \emph{any} sample set $\set{S}$ of size $\lceil\sqrt{N}\rceil$ is better than $\set{N}$.
\end{propn}
\begin{proof}[Proof Sketch]
    We bound $\xi_{1}(\set{S})$ and $\xi_{2}(\set{S})$ in terms of $r\left(\frac{N}{m} + m - 1\right)$, which is minimised at $m = \sqrt{N}$. We explicitly compute $\xi_{i}(\set{N})$ and apply Corollary \ref{corr:main_GLR_iff}. See Appendix \ref{app:proof_propn_GLR_simple} for a full proof.
\end{proof}

Proposition \ref{propn:GLR_simple} is sufficient to show that decreasing sample size can decrease MSE on a restricted class of graph models (e.g. Erdős–Rényi graphs). 
Proposition \ref{propn:GLR_simple} involves several parameters: $r$, $2r\sqrt{N}$, $\mu_{ub\_weak}$ and $\tau_{GLR\_weak}$ which can be interpreted as graph properties (see Appendix \ref{app:GLR_sensitivity} for a detailed explanation and sensitivity analysis).

Intuitively, we expect the equivalent parameters in Proposition \ref{propn:GLR_simple} and Theorem \ref{thm:main_GLR_exist} to behave similarly; for $m_{opt}$ to behave like $\left\lceil \sqrt{N} \right\rceil$, for $B(m_{opt})$ to behave like $2r\sqrt{N}$ and for $\mu_{ub}$ and $\tau_{GLR}$ to behave like their weak counterparts. We can explicitly show that $m_{opt}$ is $\mathcal{O}(\sqrt{N})$:

\begin{remark}
\label{remark:mopt}
If condition (\ref{eq:GLR_exist_thm_B_constraint}) in Theorem \ref{thm:main_GLR_exist} holds
    then $m_{opt} \in \left[\left\lfloor\sqrt{N}\right\rfloor, \left\lceil\sqrt{rN}\right\rceil \right]$ and $m_{opt} \leq \left\lceil\frac{N+1}{2}\right\rceil$.
\end{remark}
\begin{proof}
    See Appendix \ref{app:proof_of_remark_GLR_mopt}.
\end{proof}

}

The proof of Theorem \ref{thm:main_GLR_exist} leads to an upper bound of $\text{MSE}_{\set{S}}$. { We present this upper bound to link the experiments in Section \ref{sec:experiments} with Fig. \ref{fig:GLR_diagram}, which can be considered to be the case where $\sigma^{2} \to \infty$.
\begin{corollary}
\label{corr:unif_ub_xi_1_MSE}
    Let $B(m)$ be defined as in Lemma \ref{lemma:GLR_xi_2_bound_main} and $B_{k}(m)$ be defined as it will be in Lemma \ref{lemma:GLR_xi_2_bound_main_bl}. 
    
    For a sample set $\set{S}$ of size $m$,
    \begin{align}
    \textrm{MSE}_{\set{S}} &\leq (k + B_{k}(m)) + \sigma^{2} \cdot B(m). \label{eq:unif_ub_MSE}
    \end{align}
\end{corollary}
\begin{proof}
By Lemma \ref{lemma:GLR_xi_2_bound_main} $\xi_{2}(\set{S}) \leq B(m)$. By Lemma \ref{lemma:unif_ub_xi_1_GLR} in Appendix \ref{app:proof_unif_ub_xi_1_MSE}, $\xi_{1}(\set{S}) \leq k + B_{k}(m)$. Combining these using (\ref{eq:xi_decomp}) gives the bound.
\end{proof}
}

\subsubsection{Asymptotics}
Finally, we consider the large graph case.

{ While our sensitivity analysis suggests that $\tau_{GLR} \to 0$ as $N$ gets large, this is an oversimplification from assuming indepence of parameters in our analysis. We present a Proposition to show that that $\tau_{GLR} \not\to 0$ as $N \to \infty$ for ER graphs.
\begin{propn}
\label{propn:GLR_big_N}
For Erdős–Rényi graphs, as $N \to \infty$, condition (\ref{eq:GLR_exist_thm_B_constraint}) in Theorem \ref{thm:main_GLR_exist} holds w.h.p. and $\tau_{GLR}$ tends to a positive constant even under optimal choice of $\mu$ \cite{chen2017GLRbias}.
\end{propn}
\begin{proof}
    Full statement and proof in Appendix 
 \ref{app:Proof_GLR_big_N}.
\end{proof}
}
{ This shows that reducing sample size can reduce MSE on arbitrarily large graphs under GLR even with optimal $\mu$. 
}

\subsection{GLR with k-bandlimited noise}
\label{sec:GLR_bandlimited}
Once more, one might ask whether the MSE increasing with sample size under GLR is caused by some sort of interference effect between the high-frequency components of the noise and the bandlimited (low-frequency) signal. We present a variant of Theorem \ref{thm:main_GLR_exist} in Appendix \ref{app:GLR_bandlimited} to disprove this.
\vspace{-0.1cm}

\section{Experiments}
\label{sec:experiments}
\label{experiments_sec}

Our theoretical results show how the relationship between sample size and MSE changes with the level of noise, focusing on how reducing sample size will reduce MSE if the SNR is below a threshold. In this section, we demonstrate  the applicability and validity of these results via empirical experiments. 

We first demonstrate the applicability of our results with plots of the thresholds $\tau(\set{S},v)$, $\tau_{GLR}$ and $\tau_{GLR\_bl}$ against sample size (Figs. \ref{LS_SNR_Threshold_plots_all} and \ref{GLR_Threshold_plots}). These plots show concrete SNR values for the thresholds, giving a practical understanding of how high noise levels need to be for reducing sample size to reduce MSE for different random graph models and parameters. Additionally, we empirically tabulate the probabilities that graphs from each model  satisfy the conditions of our theorems (Table \ref{tbl:empirical_probabilities_conditions}), helping the reader evaluate the impact of our theorems across different applications. We then demonstrate the validity of our results by plotting $\textrm{MSE}_{\set{S}}$ against sample size (Figs. \ref{LS_ER_MSE_fig} and  \ref{GLR_ER_MSE_fig}) at SNRs below, near and above the derived thresholds, showing that the behaviour of $\textrm{MSE}_{\set{S}}$ follows our theoretical results. We finally present similar results on real-world datasets, validating the applicability of our results to real-world datasets.

\vspace{-0.2cm}
\subsection{Experimental setup}
We now present the setup of the experiments. All results are presented with 90\% confidence intervals and all experiments use the combinatorial Laplacian $\matr{L}$ and its eigenbasis.
\subsubsection{Synthetic graph generation}
We consider each of the following unweighted random graph models:
\begin{itemize}
    \item Erdős–Rényi (ER) with edge probability $p=0.8$ (experiments with other values of $p$ show similar results)
    \item Barabási-Albert (BA) with a preferential attachment to 3 vertices at each step of its construction
    \item Stochastic Blockmodel (SBM) with intra- and inter-cluster edge probabilities of $0.7 \text{ and }0.1$ respectively
\end{itemize}
We consider 10 instantiations of each model for plots, and 1000 instantiations of each model to assess the probability the graph invariant conditions in our Theorems are met.

We present threshold plots for graphs with 500, 1000, 2000 and 3000 vertices. We only present MSE plots for graphs with 500 vertices  {  (like \cite[Fig 8]{bai2020fast})} as they are intended as an accompaniment to our threshold plots and theorems to demonstrate their validity, and a single graph size suffices. 

\subsubsection{\xd{Synthetic} signal generation}
We set the bandwidth  $k = \lfloor \frac{N}{10} \rfloor$, per \cite{bai2020fast}.  We consider these SNRs for full-band noise:
\begin{itemize}
    \item{\makebox[1cm][l]{LS:}  $10^{-1}, 10^{2}, 10^{10}$ (i.e. $-10dB, 20dB, 100dB$)}
    \item{\makebox[1cm][l]{GLR:} $10^{-1}, 0.5, 10^{10}$ (i.e. $-10dB, -3dB, 100dB$)}
\end{itemize}
and the following SNRs for bandlimited noise:
\begin{itemize}
    \item{\makebox[1cm][l]{LS:}  $10^{-1}, 1, 10^{10}$ (i.e. $-10dB, 0dB, 100dB$)}
    \item{\makebox[1cm][l]{GLR:} $10^{-2}, 0.5, 10^{10}$ (i.e. $-20dB, -3dB, 100dB$)}
\end{itemize}

These SNRs are chosen to demonstrate that there are three regimes for MSE with distinctive properties---the high noise regime, the transitionary regime and the approximately noiseless regime---and that $\tau$ captures when the regimes change. Suitable values of SNRs to demonstrate this vary between reconstruction methods and noise types, hence our choices.

To test the MSE in reconstructing signals from samples on random graph models, we generate 200 signals by sampling $\vect{y} = \vect{x}_{raw} + \sigma \cdot \epsilon_{raw}$ where:
\begin{enumerate}
    \item $\vect{x}_{raw} \sim \mathcal{N}(\vect{0}, \matr{\Pi}_{bl(\set{K})})$ 
    \item[2a)] If full-band noise, $\vect{\epsilon}_{raw} \sim \mathcal{N}(\vect{0}, \matr{I}_{N})$, $\sigma = \sqrt{\frac{k}{{N \cdot\text{SNR}}}}$
    \item[2b)]  If bandlimited noise, $\vect{\epsilon}_{raw} \sim \mathcal{N}(\vect{0}, \projbl )$, $\sigma = \frac{1}{\sqrt{\text{SNR}}}$ 
\end{enumerate}

\subsubsection{Real-world datasets}

We also consider two real-world datasets as in \cite{zhi2023gaussian}. The first is an FMRI dataset, where the original graph consists of 4465 nodes corresponding to voxels of the cerebellum, with 292 Blood-Oxygen-Level-Dependent (BOLD) signals derived from FMRI. We sample a connected subgraph of 367 nodes via Neighbourhood Sampling \cite{hamilton2017inductive}. The second is a Weather dataset, where a $10$-nearest neighbours graph of 45 cities in Sweden is constructed, with 95 signals derived from the temperature. See \cite{venkitaraman2020gaussian, behjat2016signal, zhi2023gaussian} for details on graph construction and signal generation in both cases. 
We generate $\vect{x}$ by $k$-bandlimiting the original signals, where For FMRI we set $k=36$ and for Weather dataset we set $k=8$. We otherwise follow the above methodology in Synthetic Signal Generation for generating signals with full-band noise.

\subsubsection{Sample-set selection}

We generate sample sets greedily using exact analytic forms and by exactly computing $\matr{\Pi}_{bl(\set{K})}$.
\begin{itemize}
    \item[LS:] We use (\ref{eq:A-optimality})-(\ref{eq:E-optimality}) to exactly compute the MMSE {\cite{wang2018optimal,wang2019low, mfn}}, Confidence Ellipsoid {\cite{jayawant2021doptimal, tremblay2017determinantal,mfn}} and WMSE criteria { \cite{bai2020fast}}, which are deterministic and guaranteed to be noiseless-optimal.  We also look at Weighted Random sampling \cite{puy2018random}, which is neither deterministic nor guaranteed to be noiseless-optimal.
    \item[GLR:] We exactly compute the MMSE criterion, which is a function of SNR and noise type, and the WMSE criterion {\cite{EOptimalChen}}, which is not. We also consider uniform random sampling.
\end{itemize}
Note that sampling schemes in the literature tend to differ from ours mainly in that they approximate our setup for computational efficiency reasons; e.g. approximating the projection matrix $\projbl$ with a polynomial in $\matr{L}$ \cite{wang2018optimal}, and approximating optimality criteria \cite{bai2020fast}. As these differences are for efficiency reasons, we do not expect them to matter in our experiments.

\subsubsection{Parameters of reconstruction methods} We consider LS with $k = \lfloor \frac{N}{10} \rfloor$ (as above) and GLR with $\mu \in \{10^{-4}, 10^{-2}, 1 \}$.

\subsection{Experimental results on synthetic graphs}

\begin{figure*}
    \centering
    \begin{subfigure}{0.6\columnwidth}
    \resizebox{\width}{0.62\columnwidth}
    {\includegraphics[width=\columnwidth]{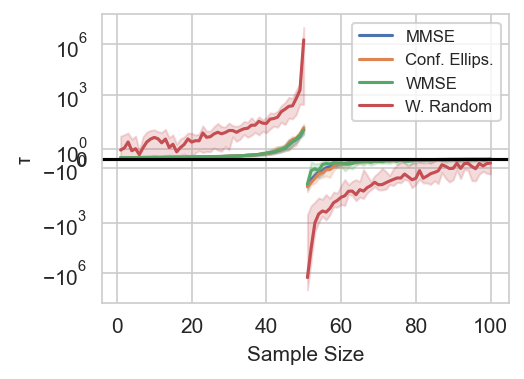}}
    \caption{Erdős–Rényi, 500 vertices} 
    \label{snr_ER}
    \end{subfigure}
    \hfill
    \begin{subfigure}{0.6\columnwidth}
    \resizebox{\width}{0.62\columnwidth}{
    \includegraphics[width=\columnwidth]{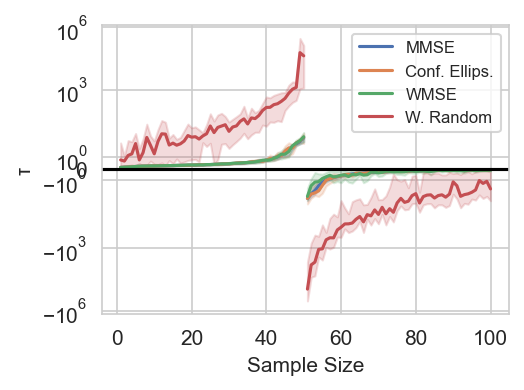}}
    \caption{Barabási-Albert, 500 vertices}%
    \label{snr_BA}%
    \end{subfigure}
    \hfill%
    \begin{subfigure}{0.6\columnwidth}
    \resizebox{\width}{0.62\columnwidth}{
    \includegraphics[width=\columnwidth]{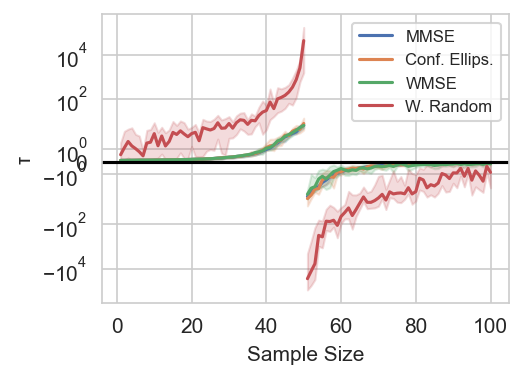}}
    \caption{SBM, 500 vertices}%
    \label{snr_SBM}%
    \end{subfigure}%
    \hfill
    \begin{subfigure}{0.6\columnwidth}
    \resizebox{\width}{0.62\columnwidth}{
    \includegraphics[width=\columnwidth]{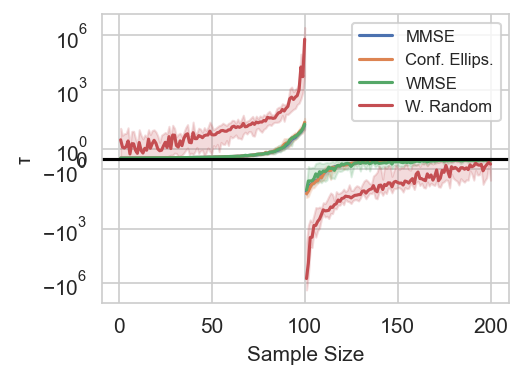}}
    \caption{Erdős–Rényi, 1000 vertices} 
    \label{snr_ER_1000}
    \end{subfigure}
    \hfill
    \begin{subfigure}{0.6\columnwidth}
    \resizebox{\width}{0.62\columnwidth}{
    \includegraphics[width=\columnwidth]{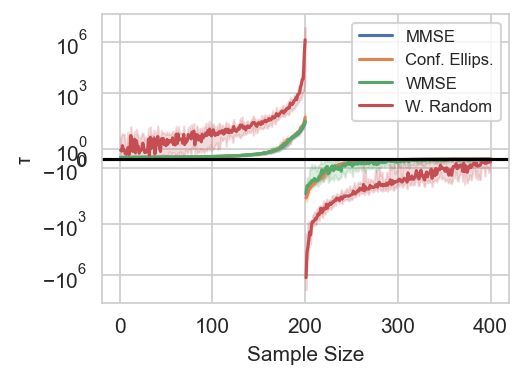}}
    \caption{Erdős–Rényi, 2000 vertices}%
    \label{snr_ER_2000}%
    \end{subfigure}
    \hfill%
    \begin{subfigure}{0.6\columnwidth}
    \resizebox{\width}{0.62\columnwidth}{
    \includegraphics[width=\columnwidth]{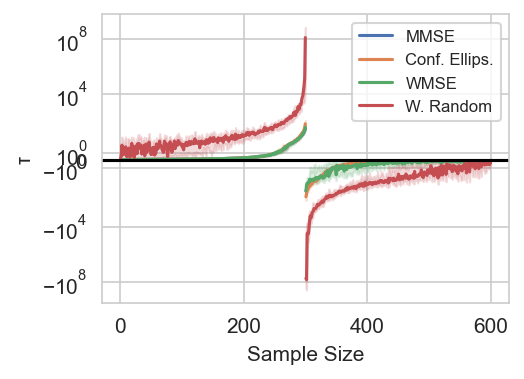}}
    \caption{Erdős–Rényi, 3000 vertices}%
    \label{snr_ER_3000}%
    \end{subfigure}%
    \caption{$\tau(\set{S},v)$ for different random graph models and different $N$ under LS and full-band noise (bandwidth = $\frac{\# \text{ vertices}}{10}$).}
\label{LS_SNR_Threshold_plots_all}
\end{figure*}

\begin{figure*}%
    \centering
    \begin{subfigure}{0.6\columnwidth}
    \resizebox{\width}{0.62\columnwidth}{
    \includegraphics[width=\columnwidth]{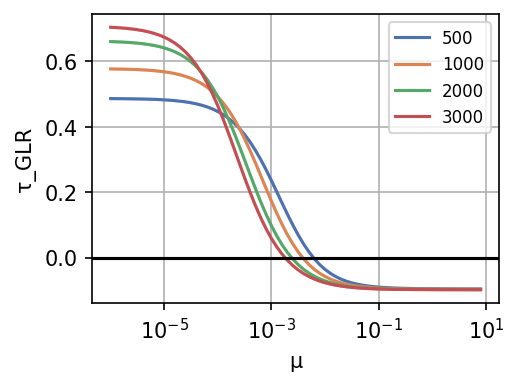}}
    \caption{Erdős–Rényi ($\tau_{GLR}$)}
    \label{tau_GLR_er}
    \end{subfigure}
    \hfill
    \begin{subfigure}{0.6\columnwidth}
    \resizebox{\width}{0.62\columnwidth}{
    \includegraphics[width=\columnwidth]{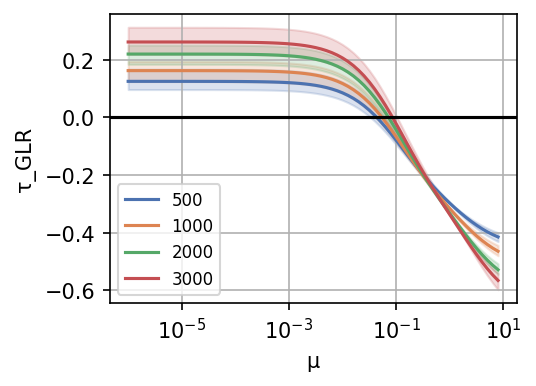}}
    \caption{Barabási-Albert ($\tau_{GLR}$)}%
    \label{tau_GLR_BA}%
    \end{subfigure}
    \hfill%
    \begin{subfigure}{0.6\columnwidth}
    \resizebox{\width}{0.62\columnwidth}{
    \includegraphics[width=\columnwidth]{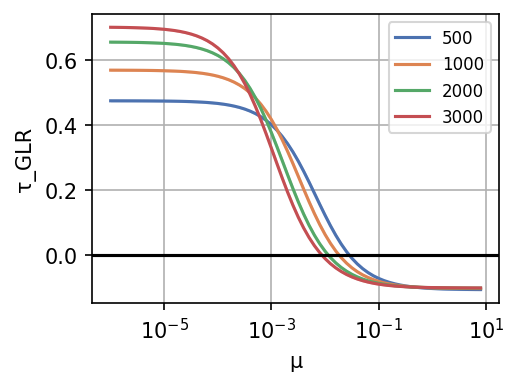}}
    \caption{SBM ($\tau_{GLR}$)}%
    \label{tau_GLR_SBM}%
    \end{subfigure}%
    \hfill
    \caption{$\tau_{GLR}$  for different random graph models (\#vertices = colour, bandwidth = $\frac{\text{\# vertices}}{10}$).}
\label{GLR_Threshold_plots}
\end{figure*}

\begin{figure*}%
    \centering
    \begin{subfigure}{0.6\columnwidth}
    \resizebox{\width}{0.62\columnwidth}{
    \includegraphics[width=\columnwidth]{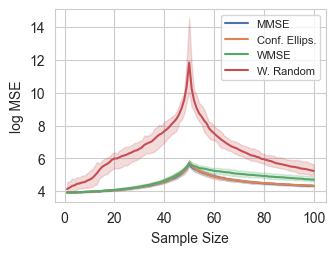}}
    \caption{Full-band noise, SNR = $10^{-1}$}
    \label{MSE_subfiga}
    \end{subfigure}\hfill
    \begin{subfigure}{0.6\columnwidth}
    \resizebox{\width}{0.62\columnwidth}{
    \includegraphics[width=\columnwidth]{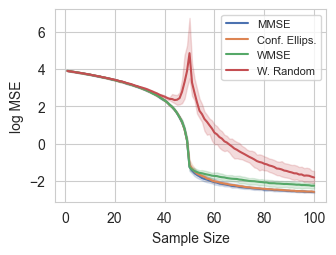}}
    \caption{Full-band noise, SNR = $10^{2}$}%
    \label{MSE_subfigb}%
    \end{subfigure}\hfill%
    \begin{subfigure}{0.6\columnwidth}
    \resizebox{\width}{0.62\columnwidth}{
    \includegraphics[width=\columnwidth]{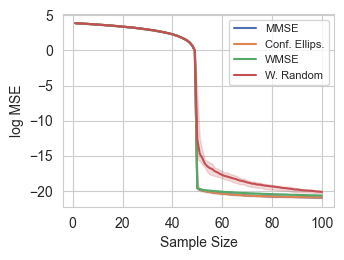}}
    \caption{Full-band noise, SNR = $10^{10}$}%
    \label{MSE_subfigc}%
    \end{subfigure}%
    \caption{Average MSE under LS on ER graphs (\#vertices=500, bandwidth = 50).}
\label{LS_ER_MSE_fig}
\end{figure*}

\begin{figure*}%
    \centering
    \begin{subfigure}{0.6\columnwidth}
    \resizebox{\width}{0.62\columnwidth}{
    \includegraphics[width=\columnwidth]{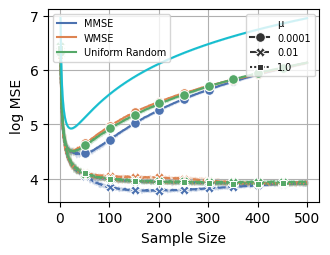}}
    \caption{Full-band noise, SNR = $10^{-1}$}
    \label{GLR_MSE_subfiga}
    \end{subfigure}\hfill
    \begin{subfigure}{0.6\columnwidth}
    \resizebox{\width}{0.62\columnwidth}{
    \includegraphics[width=\columnwidth]{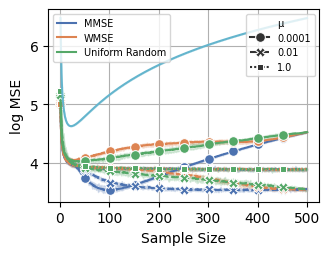}}
    \caption{Full-band noise, SNR = $\frac{1}{2}$}%
    \label{GLR_MSE_subfigb}%
    \end{subfigure}\hfill%
    \begin{subfigure}{0.6\columnwidth}
    \resizebox{\width}{0.62\columnwidth}{
    \includegraphics[width=\columnwidth]{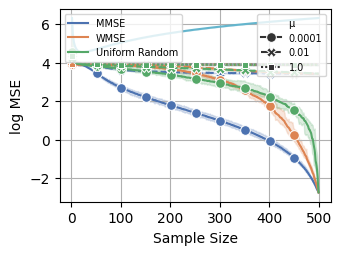}}
    \caption{Full-band noise, SNR = $10^{10}$}%
    \label{GLR_MSE_subfigc}%
    \end{subfigure}%
    \caption{Average MSE under GLR on ER graphs (\#vertices=500, bandwidth = 50). Line without markers is an upper bound.}
\label{GLR_ER_MSE_fig}
\end{figure*}

For full-band noise, we present threshold plots for all graphs and MSE plots for ER graphs in the main body of the paper. MSE plots for BA and SBM graphs are presented in Appendix \ref{plot_appendix}. For bandlimited noise, we present threshold plots, MSE plots and discussion of how the experimental and theoretical results correspond in Appendix \ref{app:Experiments_Bandlimited}.
{ We also present Table \ref{tbl:theory_experiment_correspondence} showing how our plots correspond to our theoretical results, which also corresponds to the summary of theoretical results in Table \ref{tbl:general_theory}.}

\begin{table}[h]
\setlength{\tabcolsep}{3pt}
\caption{Correspondence between theoretical and empirical results.}
\centering
\begin{tabularx}{\linewidth}{|X|p{1.5cm}|p{1.5cm}|p{1.5cm}|p{1.7cm}|p{1.7cm}|}
\hline
\multirow{2}{*}{} 
  & \multicolumn{2}{c|}{\centering\textbf{LS}} 
  & \multicolumn{2}{c|}{\centering\textbf{GLR}} \\
\cline{2-5}
 & Full-band & Bandlimited & Full-band & Bandlimited \\
\hline
\textbf{Character- isation} 
 & \makecell[tl]{Corr~\ref{main_ls} \\
   (Figs.~\ref{snr_ER}-\ref{snr_SBM} \\
   \& \ref{MSE_subfiga}-\ref{MSE_subfigc})}
 & \makecell[tl]{Corr~\ref{corr:LS_bandlimited_noise_big_variance} \\
   (Figs.~\ref{bandlimited_MSE_subfiga} \\-\ref{bandlimited_MSE_subfigc})}
 & \multicolumn{2}{c|}{\makecell[tc]{Corr~\ref{corr:main_GLR_iff} \\
   (Figs.~\ref{GLR_MSE_subfiga}-\ref{GLR_MSE_subfigc})}} \\
\hline
\textbf{Existence}
 & \makecell[tl]{Thm~\ref{thm:noiseless_optimality_means_noise_sensitivity} \\
   (Figs.~\ref{snr_ER}-\ref{snr_SBM} \\
   \& \ref{MSE_subfiga}-\ref{MSE_subfigc})}
 & \makecell[tl]{Corr~\ref{corr:LS_bandlimited_noise_sample_only_k} \\
   (Figs.~\ref{bandlimited_MSE_subfiga} \\-\ref{bandlimited_MSE_subfigc})}
 & \makecell[tl]{Thm~\ref{thm:main_GLR_exist} \\
   (Figs.~\ref{tau_GLR_er}-\ref{tau_GLR_SBM} \\
   \& \ref{GLR_MSE_subfiga}-\ref{GLR_MSE_subfigc} \\
   \& Table~\ref{tbl:empirical_probabilities_conditions})}
 & \makecell[tl]{Thm~\ref{thm:main_GLR_bl} \\
   (Figs.~\ref{tau_GLR_bl_er}-\ref{tau_GLR_bl_SBM} \\
   \& \ref{bandlimited_GLR_MSE_subfiga}-\ref{bandlimited_GLR_MSE_subfigc} \\
   \& Table~\ref{tbl:empirical_probabilities_conditions_bl})} \\
\hline
\textbf{Asymptotics}
 & \makecell[tl]{Rmk~\ref{rmk:LS_big_N} \\
   (Figs.~\ref{snr_ER_1000}-\ref{snr_ER_3000})}
 & \makecell[tl]{Rmk~\ref{rmk:LS_big_N_bl}}
 & \makecell[tl]{Propn~\ref{propn:GLR_big_N} \\
   (Figs.~\ref{tau_GLR_er}-\ref{tau_GLR_SBM})}
 & \makecell[tl]{Propn~\ref{propn:GLR_big_N_bl} \\
   (Figs.~\ref{tau_GLR_bl_er}-\ref{tau_GLR_bl_SBM})} \\
\hline
\end{tabularx}
\label{tbl:theory_experiment_correspondence}
\end{table}

\subsubsection{$\tau$ plots (LS)}

Fig. \ref{LS_SNR_Threshold_plots_all}  shows $\tau(\set{S},v)$ as sample size varies for sequential sampling methods under LS, where $v$ is the latest node added to $\set{S}$. 
     For ER graphs, for sample size smaller than the bandwidth, $\tau(\set{S},v) > 0$ and beyond that $\tau(\set{S},v) \leq 0$. The maximum of $\tau(\set{S},v)$ observed is approximately $10^{6}$ ($60dB$) for weighted random sampling, and approximately 10 ($10dB$) for the deterministic sampling methods. The confidence intervals for weighted random sampling is much wider than for the deterministic sampling methods.
     Next, we observe the same phenomenon as with ER for BA and SBM graphs, with maxima of approximately $10^5$ ($50dB$) for weighted random sampling and maxima of approximately 10 ($10dB$) for the deterministic sampling methods.
Finally, Figs. \ref{snr_ER_1000}-\ref{snr_ER_3000} show the same phenomenon as Fig. \ref{snr_ER} happens for ER graphs at sizes of 1000, 2000 and 3000 vertices.

 We now correlate our experiments and our theoretical results. As Corollary \ref{main_ls} is necessary and sufficient, the sign of $\tau(\set{S},v)$ tells us exactly when removing a vertex improves $\set{S}$. Therefore $\tau$ being negative is an informative statement, telling us that $\set{S} \backslash \{v\}$ is \emph{never} better than $\set{S}$. Concretely, if SNR is below the maximum $\tau(\set{S},v)$ observed, then when sample size equals bandwidth we can reduce sample size to reduce MSE.

 Theorem \ref{thm:noiseless_optimality_means_noise_sensitivity} proves that noiseless-optimal methods (MMSE, Confidence Ellipsoid and WMSE in our experiments) will have $\tau(\set{S},v) > 0$ for sample sizes smaller than the bandwidth, and then $\tau(\set{S},v) \leq 0$ afterwards, and Fig. \ref{LS_SNR_Threshold_plots_all}  validates this. Note that even though this pattern holds for Weighted Random Sampling in our experiments, Theorem \ref{thm:noiseless_optimality_means_noise_sensitivity} does not guarantee it always holds for Weighted Random Sampling.

 While we prove that $\tau(\set{S},v) \not\to 0$ as $N \to \infty$, we conjecture the stronger claim that at a sample size equal to bandwidth, $\tau(\set{S},v)$ might actually increase with graph size, which is supported (but not proven) by Figs. \ref{snr_ER_1000}-\ref{snr_ER_3000}.

\subsubsection{$\tau_{GLR}$ plots (GLR)}
In Figs. \ref{tau_GLR_er}-\ref{tau_GLR_SBM}, we plot $\tau_{GLR}$, where if $\textrm{SNR} < \tau_{GLR}$, then there is a sample size $m_{opt} < N$ where \emph{any} sample set of size $m_{opt}$ is better than $\set{N}$. Unlike with LS, our theorems are only sufficient so $\text{SNR} > \tau_{GLR}$ is uninformative. %
    For ER graphs, we see that $\tau_{GLR}$ is decreasing in $\mu$, and that $\tau_{GLR} > 0$ for sufficiently small $\mu$ for all graph sizes. The maximum value in this case for $\tau_{GLR}$ ranges between 0.4 ($-4dB$) and 0.7 ($-1.5dB$). %
    We see a similar pattern to ER graphs for BA graphs, the main differences being that $\tau_{GLR} > 0$ for larger values of $\mu$ and that the maxmimum of $\tau_{GLR}$ is  approximately {0.2 ($-7dB$)}.
    $\tau_{GLR}$ for SBM graphs behaves very similarly to ER graphs in our experiments.
Note that the confidence intervals for ER and SBM graphs are so tight as to not be clearly seen in Fig. \ref{GLR_Threshold_plots}, while being much wider for BA graphs. As  graph properties, when we sample from a random graph model, $r$ is a random variables. We observe wider confidence intervals {when $\mathcal{G}$ is sampled from the BA graph model as $\text{Var}_{\mathcal{G}}(r)$ is much higher than when we sample from the ER or SBM graph models.}

As $\tau_{GLR} > 0$, Figs. \ref{tau_GLR_er}-\ref{tau_GLR_SBM} show we can reduce sample size to reduce MSE for all examined graph models. We see $\tau_{GLR}$ is only positive for small enough $\mu$, motivating $\mu_{ub}$ in Theorem \ref{thm:main_GLR_exist}. The value of $\mu$ below which $\tau_{GLR} > 0$ is at least $\mu_{ub}$.

Finally, Proposition \ref{propn:GLR_big_N} proves that $\tau_{GLR} \not\to 0$ as $N \to \infty$ for $\mu = \frac{c}{\lambda_{2}}$,  or $\frac{c}{\lambda_{N}}$, or $\frac{c}{\sqrt{\lambda_{2}\lambda_{N}} }$ on ER graphs, i.e. for decreasing $\mu$ as $N$ increases. We note that Figs. \ref{tau_GLR_er}-\ref{tau_GLR_SBM}  provide empirical evidence that this might hold for all graph models tested.

\subsubsection{MSE plots (LS)}
The MSE plots demonstrate the validity of our theoretical results linking MSE and sample size.

Figs. \ref{MSE_subfiga}-\ref{MSE_subfigc} show log MSE against sample size for LS under full-band noise. %
    For high noise (a), we see MSE increases with sample size no larger than bandwidth, and decreases afterwards -- that is, it is $\Lambda$-shaped, as described in Remark \ref{remark:LS_error_lambda_shaped}. %
    In (b), for our three deterministic noiseless-optimal sampling schemes (orange/green/blue), MSE is decreasing in sample size. For weighted random sampling, we see for sample sizes no larger than bandwidth, MSE first decreases and then increases, attaining a maximum with sample size at bandwidth, and then decreases with sample size.
    In (c), the almost noiseless case, we see MSE is decreasing in sample size for all sampling schemes, with a large drop when sample size is at bandwidth.

Comparing Figs. \ref{MSE_subfiga}-\ref{MSE_subfigc} to Fig. \ref{snr_ER}, Fig. \ref{MSE_subfiga} corresponds to when $\text{SNR} < \tau(\set{S},v)$, Fig. \ref{MSE_subfigc} corresponds to $\text{SNR} > \tau(\set{S},v)$ and Fig. \ref{MSE_subfigb} corresponds to when SNR lies between some values of $\tau(\set{S},v)$ for weighted random sampling.  
We see that MSE increases with sample size when $\text{SNR} < \tau(\set{S},v)$ and decreases otherwise, proving the validity of Corollary \ref{main_ls}.
Fig. \ref{MSE_subfiga} (green, orange, blue curves) shows that for low SNRs, optimal sampling schemes  lead MSE to monotonically increase with each additional sample until the sample size reaches the bandwidth, illustrating Theorem \ref{thm:noiseless_optimality_means_noise_sensitivity} and Remark \ref{remark:LS_error_lambda_shaped}. We also validate Remark \ref{remark:remove_multiple_nodes}: if we are slightly above the bandwidth ($50$ for Fig. \ref{MSE_subfiga}), i.e., to the right of the peak, then reducing sample size by one does not reduce MSE; however, if we significantly reduce sample size, i.e., transitioning from just right of the peak to left of the peak, we can reduce MSE. %

Interestingly, Fig. \ref{MSE_subfiga} shows that at a low SNR of $10^{-1}$, the optimal sample size under LS is zero. Although this might appear counter-intuitive at a first glance, it makes concrete the idea that reconstruction does not work if there is too much noise: at this noise level letting $\vect{\hat{x}} = \vect{0}$ rather than fitting with LS with any number of observed vertices will result in a lower MSE on average. We can also formalise this in terms of our Bias-Variance decomposition; a $\vect{0}$ reconstruction has high bias but zero variance, and reconstructing from a non-zero number of samples has lower bias but positive variance. At a high enough noise the variance term in the MSE will dominate, and the MSE from taking $\hat{\vect{x}} = \vect{0}$ will be lower than reconstructing from any non-zero number of samples. The same reasoning applies to why the MSE at a smaller number of samples (e.g. 5 samples) is better than a larger number (e.g. 100 samples).

On the other hand, for high SNRs (Fig. \ref{MSE_subfigc}), MSE decreases monotonically as sample size increases for all sampling schemes, showing Corollary \ref{main_ls} is necessary and sufficient. %
Finally, Fig. \ref{MSE_subfigb} illustrates the situation between the two cases.

\subsubsection{MSE plots (GLR)}
As with LS, the MSE plots demonstrate the validity of our theoretical results linking MSE and sample size.
Figs. \ref{GLR_MSE_subfiga}-\ref{GLR_MSE_subfigc} show how MSE changes with sample size for different values of $\mu$ under full-band noise, along with an upper bound (light blue) which is not dependent on $\mu$. This bound is approximately U-shaped in all cases.
    For $\text{SNR} = 10^{-1}$, we see for $\mu=10^{-4}$, under all sampling schemes, MSE is minimised at a sample size around 16 to 22. For the MMSE sampling scheme with $\mu=0.01$, MSE is minimised at a sample size of approximately 200. In all other cases where $\text{SNR} = 10^{-1}$, MSE is minimised at full observation $(|\set{S}|=500)$.
    For $\text{SNR}=\frac{1}{2}$ and $\mu=10^{-4}$, MSE is minimised at a sample size a bit less than 100. For larger $\mu$, we see MSE is minimised at full observation.
    In the nearly noiseless case, MSE decreases with sample size under all parameter choices. In all cases where the MSE is minimised at a sample size less than $N$, the MSE is approximately U-shaped like our bound { and 
 our diagram Fig. \ref{fig:GLR_diagram}}.

Figs. \ref{GLR_MSE_subfiga}-\ref{GLR_MSE_subfigc} illustrate Corollary \ref{corr:main_GLR_iff}, Theorem \ref{thm:main_GLR_exist} and Corollary \ref{corr:unif_ub_xi_1_MSE} in the following ways. First, the MSE upper bound corresponds to Corollary \ref{corr:unif_ub_xi_1_MSE}, and we can see it is greater than any observed MSE at each sample size. The sample size which minimises our upper bound is $m_{opt}$ (Remark \ref{remark:mopt}) and as $r \in (1,1.01]$ in our Erdős–Rényi experiments, $m_{opt} \in [22,23]$, which empirically well approximates the sample size that minimises MSE in our low $\mu$ and low SNR experiments. %
Second, Figs. \ref{GLR_MSE_subfiga}-\ref{GLR_MSE_subfigc} show that MSE can decrease with increasing sample size, and Fig. \ref{GLR_MSE_subfigc} shows that at high SNR full observation is best, illustrating the necessary and sufficient nature of Corollary \ref{corr:main_GLR_iff}.
Finally, Figs. \ref{GLR_MSE_subfiga}-\ref{GLR_MSE_subfigc} illustrate the dependence on SNR and $\mu$ in Theorem \ref{thm:main_GLR_exist}, i.e. at low $\mu$ and SNR the optimal sample size is less than $N$, but at high enough $\mu$ or SNR this no longer holds.  

Figs. \ref{GLR_MSE_subfiga}-\ref{GLR_MSE_subfigc} also demonstrate some limitations of the characterisation in Corollary \ref{corr:unif_ub_xi_1_MSE} and Theorem \ref{thm:main_GLR_exist}. We see from Fig. \ref{GLR_MSE_subfigb} that even though the MSE at $m_{opt}$ is lower than at $N$, our bound never goes below the maximum observed MSE, so is too loose to show this. This is partly because the Theorem and the Corollary bound $\xi_{2}(\set{S})$ only as a function of $|\set{S}|$, ignoring $\mu$ and the composition of $\set{S}$. This limitation corresponds to a gap between $\tau_{GLR}$ and $\tau(\set{N},\set{S}^{c})$, demonstrating $\tau_{GLR}$ is a lower bound for $\tau(\set{N},\set{S}^{c})$ rather than an exact characterisation.

\subsubsection{Checking conditions}
While our theorems for LS apply to all graphs, Theorem \ref{thm:main_GLR_exist} for GLR relies on conditions around graph invariants. We sample graphs from each random graph model to empirically show the probability the conditions of Theorem \ref{thm:main_GLR_exist} are met at a sample size of $m_{opt}$ for some $\mu > 0$:
\begin{table}[h!btp]
\caption{Probability theorem conditions are met.}
    \vspace{-0.2cm}
    \begin{center}
        \begin{tabular}{|l|c|c|c|}
     \hline
       & \textbf{ER} & \textbf{SBM} & \textbf{BA} \\ 
     \hline
     Theorem \ref{thm:main_GLR_exist} conditions met & 100\% & 100\% & 99.9\% \\ \hline
        \end{tabular}
    \end{center}
    \label{tbl:empirical_probabilities_conditions}
    \vspace{-0.4cm}
\end{table}

Proposition \ref{propn:GLR_big_N} shows the conditions in Theorem \ref{thm:main_GLR_exist} hold w.h.p. for ER graphs as $N \to \infty$. However, the proposition does not say whether the conditions hold for a graph of a given size, or other graph models. The results in Table \ref{tbl:empirical_probabilities_conditions} show empirically that the conditions hold under full-band noise (Theorem \ref{thm:main_GLR_exist}) for ER, BA and SBM graphs with $500$ vertices. This outlines the applicability of our theorems.

If the conditions on our Theorems are not met, they provide no information about the shape of the MSE. However, Figs. \ref{bandlimited_GLR_BA_MSE_fig} and  \ref{bandlimited_GLR_SBM_MSE_fig} in Appendix \ref{plot_appendix} show empirically that for BA and SBM graphs under bandlimited noise with $\text{SNR} \in \{10^{-2}, \frac{1}{2}\}$ under GLR with $\mu \in 
\{10^{-2}, 10^{-4}\}$, even if the conditions of Theorem \ref{thm:main_GLR_bl} are not met, reducing sample size from ${N}$ to below ${N}$ reduces MSE under the presented sampling schemes. We leave further investigation of this as future work. %

\subsection{Experimental results on real-world datasets}

\begin{figure*}[p]
\centering
\begin{minipage}{0.48\textwidth}
    \centering
    \begin{subfigure}[t]{0.45\columnwidth}
    \resizebox{\width}{0.62\columnwidth}
    {\includegraphics[width=\columnwidth]{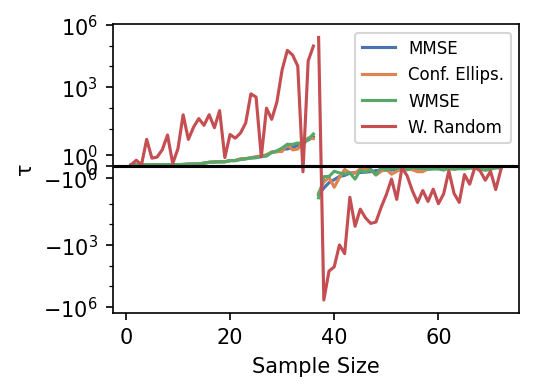}}
    \caption{FMRI} 
    \label{snr_FMRI}
    \end{subfigure}
    \begin{subfigure}[t]{0.45\columnwidth}
    \resizebox{\width}{0.62\columnwidth}{
    \includegraphics[width=\columnwidth]{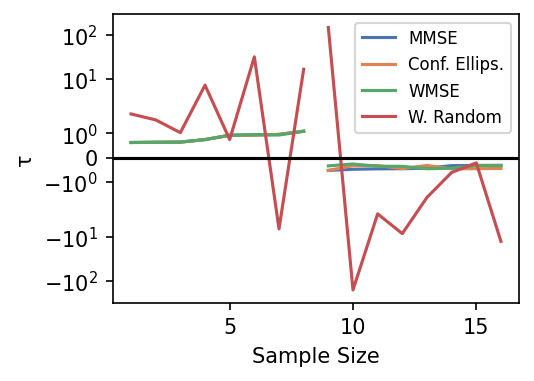}}
    \caption{Weather}%
    \label{snr_Weather}%
    \end{subfigure}
    \caption{$\tau(\set{S},v)$ for real-world data under LS \& full-band noise.}
    \label{LS_SNR_Threshold_plots_all_real}
\end{minipage}
\hfill
\begin{minipage}{0.48\textwidth}
    \centering
    \begin{subfigure}[t]{0.45\columnwidth}
    \resizebox{\width}{0.62\columnwidth}{
    \includegraphics[width=\columnwidth]{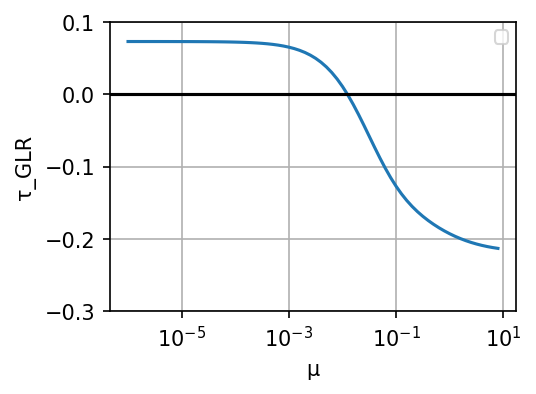}}
    \caption{FMRI}
    \label{tau_GLR_fmri}
    \end{subfigure}
    \begin{subfigure}[t]{0.45\columnwidth}
    \resizebox{\width}{0.62\columnwidth}{
    \includegraphics[width=\columnwidth]{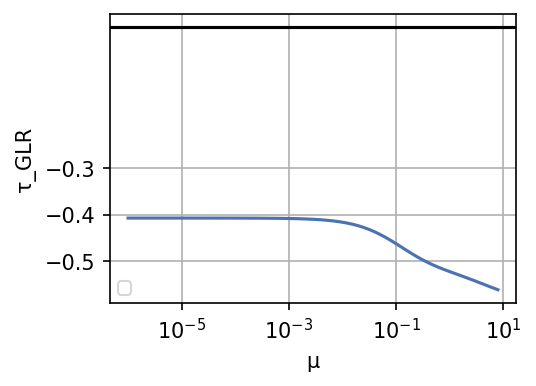}}
    \caption{Weather}%
    \label{tau_GLR_weather}%
    \end{subfigure}
    \caption{$\tau_{GLR}$ for real-world data.}
    \label{GLR_Threshold_plots_real}
\end{minipage}
\end{figure*}

\begin{figure*}[p]
    \centering
    \begin{subfigure}{0.6\columnwidth}
    \resizebox{\width}{0.62\columnwidth}{
    \includegraphics[width=\columnwidth]{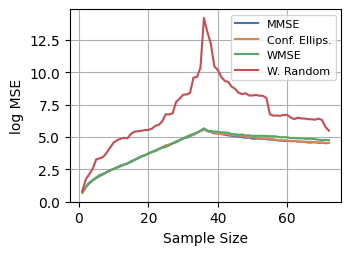}}
    \caption{FMRI, SNR = $10^{-1}$}
    \label{fmri_MSE_subfiga}
    \end{subfigure}\hfill
    \begin{subfigure}{0.6\columnwidth}
    \resizebox{\width}{0.62\columnwidth}{
    \includegraphics[width=\columnwidth]{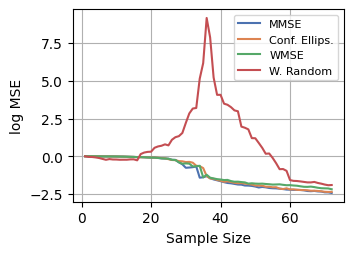}}
    \caption{FMRI, SNR = $10^{2}$}%
    \label{fmri_MSE_subfigb}%
    \end{subfigure}\hfill%
    \begin{subfigure}{0.6\columnwidth}
    \resizebox{\width}{0.62\columnwidth}{
    \includegraphics[width=\columnwidth]{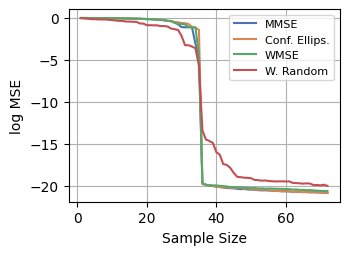}}
    \caption{FMRI, SNR = $10^{10}$}%
    \label{fmri_MSE_subfigc}%
    \end{subfigure}%
    \hfill
    \begin{subfigure}{0.6\columnwidth}
    \resizebox{\width}{0.62\columnwidth}{
    \includegraphics[width=\columnwidth]{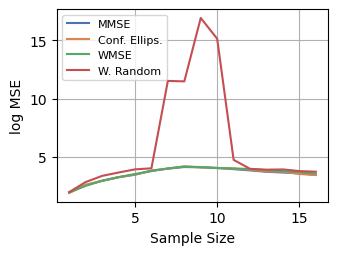}}
    \caption{Weather, SNR = $10^{-1}$}
    \label{weather_MSE_subfiga}
    \end{subfigure}\hfill
    \begin{subfigure}{0.6\columnwidth}
    \resizebox{\width}{0.62\columnwidth}{
    \includegraphics[width=\columnwidth]{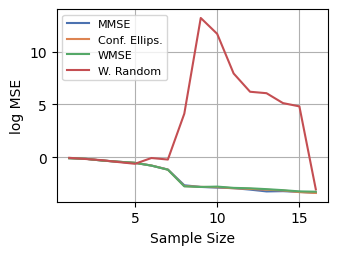}}
    \caption{Weather, SNR = $1$}%
    \label{weather_MSE_subfigb}%
    \end{subfigure}\hfill%
    \begin{subfigure}{0.6\columnwidth}
    \resizebox{\width}{0.62\columnwidth}{
    \includegraphics[width=\columnwidth]{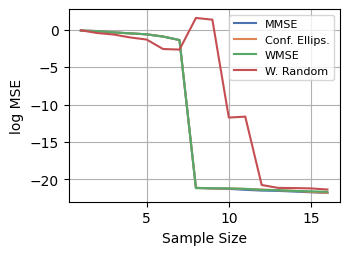}}
    \caption{Weather, SNR = $10^{10}$}%
    \label{weather_MSE_subfigc}%
    \end{subfigure}%
    \caption{Average MSE under LS on real-world data and full-band noise.}
\label{LS_real_MSE_fig}
\end{figure*}

\begin{figure*}[p]%
    \centering
    \begin{subfigure}{0.6\columnwidth}
    \resizebox{\width}{0.62\columnwidth}{
    \includegraphics[width=\columnwidth]{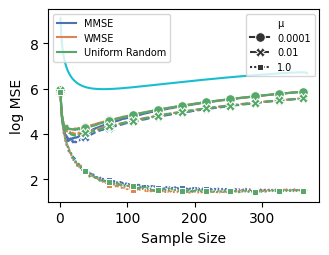}}
    \caption{FMRI, SNR = $10^{-1}$}
    \label{fmri_GLR_MSE_subfiga}
    \end{subfigure}\hfill
    \begin{subfigure}{0.6\columnwidth}
    \resizebox{\width}{0.62\columnwidth}{
    \includegraphics[width=\columnwidth]{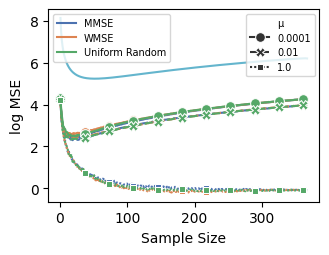}}
    \caption{FMRI, SNR = $\frac{1}{2}$}%
    \label{fmri_GLR_MSE_subfigb}%
    \end{subfigure}\hfill%
    \begin{subfigure}{0.6\columnwidth}
    \resizebox{\width}{0.62\columnwidth}{
    \includegraphics[width=\columnwidth]{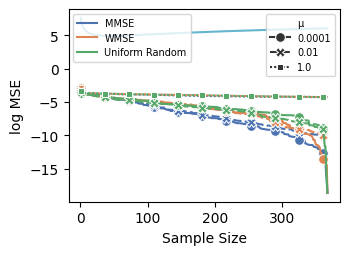}}
    \caption{FMRI, SNR = $10^{10}$}%
    \label{fmri_GLR_MSE_subfigc}%
    \end{subfigure}%
    \hfill
    \begin{subfigure}{0.6\columnwidth}
    \resizebox{\width}{0.62\columnwidth}{
    \includegraphics[width=\columnwidth]{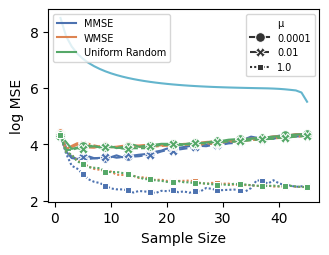}}
    \caption{Weather, SNR = $10^{-2}$}
    \label{weather_GLR_MSE_subfiga}
    \end{subfigure}\hfill
    \begin{subfigure}{0.6\columnwidth}
    \resizebox{\width}{0.62\columnwidth}{
    \includegraphics[width=\columnwidth]{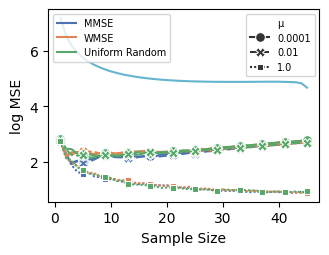}}
    \caption{Weather, SNR = $\frac{1}{2}$}%
    \label{weather_GLR_MSE_subfigb}%
    \end{subfigure}\hfill%
    \begin{subfigure}{0.6\columnwidth}
    \resizebox{\width}{0.62\columnwidth}{
    \includegraphics[width=\columnwidth]{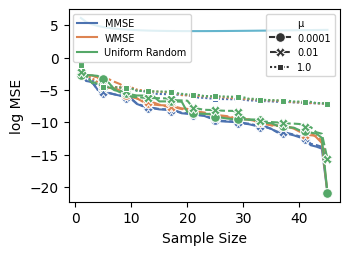}}
    \caption{Weather, SNR = $10^{10}$}%
    \label{weather_GLR_MSE_subfigc}%
    \end{subfigure}%
    \caption{Average MSE under GLR on real-world data and full-band noise. Line without markers is an upper bound.}
\label{GLR_real_MSE_fig}
\end{figure*}

We first discuss the plots of $\tau$ and $\tau_{GLR}$ for real-world datasets. In practice, one cannot know for sure what the signal model of a real-world signal is; however, computation of $\tau$ is dependent on our choice of theoretical signal model. We have therefore computed $\tau$ under the signal model assumptions given in Section \ref{sec:signal_model}.  

{%
\subsubsection{$\tau$ plots (LS)}
We examine Fig. \ref{LS_SNR_Threshold_plots_all_real}. By Proposition \ref{propn:averages_generalise_to_forall}, the sign of $\tau(\set{S},v)$ we have plotted is provably correct for \emph{any} signal model $\tau$. 
We validate the actual value of $\tau$ with MSE experiments (Figs. \ref{LS_real_MSE_fig} \& \ref{GLR_real_MSE_fig}).

\subsubsection{$\tau$ plots (GLR)}
By Theorem \ref{thm:main_GLR_exist}, $\tau_{GLR}$ being positive means $\tau(\set{N},\set{S}^{C}) >0$ under the bandlimited signal model. As Theorem \ref{thm:main_GLR_exist} functions by showing $\Delta_{2} >0$, by Proposition \ref{propn:averages_generalise_to_forall}, $\tau_{GLR}$ being positive under one signal model means $\tau(\set{N},\set{S}^{C})$ is positive under any signal model. We use this to interpret Fig. \ref{GLR_Threshold_plots_real}.
We can see from Fig. \ref{tau_GLR_fmri} that if $\mu<0.01$ then at some noise level, any set of size $m_{opt}$ is better than $\set{N}$. Fig. \ref{tau_GLR_weather} which shows $\tau_{GLR} < 0$ for all $\mu$ is completely uninformative -- $\tau_{GLR}$ is a lower bound for $\tau(\set{N},\set{S}^{C})$, and a negative lower bound cannot tell us whether the bounded quantity is or isn't positive. We validate this in the MSE section.

\subsubsection{MSE plots (LS)}
At high noise levels, Figs. \ref{weather_MSE_subfiga} and \ref{fmri_MSE_subfiga} show a $\Lambda$-shaped MSE, validating Proposition \ref{propn:averages_generalise_to_forall} applied to Remark \ref{remark:LS_error_lambda_shaped}. 

In both plots in Fig. \ref{LS_SNR_Threshold_plots_all_real}, all of the lines are above $10^{-1}$. This corresponds to Fig. \ref{weather_MSE_subfiga} and Fig. \ref{fmri_MSE_subfiga} being $\Lambda$-shaped. In Fig. \ref{LS_SNR_Threshold_plots_all_real}, all of the lines are under $10^{10}$. This correponds to Fig. \ref{weather_MSE_subfigc} and Fig. \ref{weather_MSE_subfigc} being decreasing, which we observe.

In Fig. \ref{snr_FMRI}, we see the red line (Weighted Random sampling) is broadly above $10^2$ and the other lines (MMSE, WMSE and Confidence Ellipsoid Sampling) are below $10^2$. This corresponds to the red MSE line being $\Lambda$-shaped and the other lines decreasing in Fig. \ref{fmri_MSE_subfigb}. 

In Fig \ref{snr_Weather} all lines are below $10^{2}$ and so we expect at $\text{SNR}=10^{2}$ for all lines to be decreasing. For MMSE, WMSE and Confidence Ellipsoid sampling this is consistent with Fig. \ref{weather_MSE_subfigb}; however, it is clear our covariance assumption has underestimated $\tau$ for Weighted Random Sampling as the red line is $\Lambda$-shaped in Fig. \ref{weather_GLR_MSE_subfigb}.

We see that the real-world results are broadly in line with the synthetic experiments, validating our use of the signal model presented in Section \ref{sec:signal_model} to calculate $\tau$ for LS.

\subsubsection{MSE plots (GLR)}
All four figures in Fig. \ref{GLR_real_MSE_fig} corresponding to $\text{SNR} \leq \frac{1}{2}$ and $\mu < 1$ show the minimum MSE at sample size of significantly less than $N$, which validates Theorem \ref{thm:main_GLR_exist} combined with Proposition \ref{propn:averages_generalise_to_forall}. This is expected for the FMRI graph for $\mu<0.01$, as $\tau_{GLR} > 0$ and we had proof that decreasing sample size from $N$ would decrease MSE. However, it is the case for the Weather graph, even though our computation of $\tau_{GLR}$ was negative; this is not unexpected as Theorem \ref{thm:main_GLR_exist} is sufficient but not necessary.
}

\section{Discussion}
In this paper we studied the impact of sample size on linear reconstruction of noisy $k$-bandlimited graph signals. We showed theoretically and experimentally, in the same settings as much of the sample set selection literature, that reconstruction error is not always monotonic in sample size, i.e., at sufficiently low SNRs, reconstruction error can sometimes be improved by \emph{reducing} sample size. %
Our finding reveals that existing results in the literature for the noiseless setting may not necessarily generalise to the noisy case,  even when considering regularised reconstruction methods. It also demonstrates the need to consider both optimal sample size selection and reconstruction methods at the same time, and motivates assessment of noise levels in datasets to do so. %

\xd{One practical implication of our theoretical results is that it can provide useful guidance on how to make use of the sampling budget available. For example, if one is aware of potentially high noise present in the observed signals, to minimise MSE under LS reconstruction, it might be preferable to not use all of the provided sampling budget.}

\xd{Another use case could be a practical sampling algorithm, for both the LS and GLR reconstructions. For LS reconstruction, if one has chosen $k$ or fewer nodes by some noiseless-optimal scheme, one can calculate $\tau(\set{S},v)$ for each element of $\set{S}$, remove $v$ if $\text{SNR} < \tau(\set{S},v)$, and repeat until it is no longer possible to find a node $v$ in our remaining set $\set{S}_{rem}$ where $\text{SNR} < \tau(\set{S}_{rem},v)$.}

\xd{For the GLR reconstruction, with a fixed $\mu$ and belief that the signal model in Section \ref{sec:every_x} applies (the latter could be assessed based on past observations and is required for computation of the SNR threshold $\tau$), then our results suggest the following sampling algorithm can be considered. Given the graph structure, before observing any nodes:
\begin{enumerate}
    \item Check whether condition \ref{eq:GLR_exist_thm_B_constraint} in Theorem \ref{thm:main_GLR_exist} applies (this only depend on graph structure);
    \item If they do, compute $\tau_{GLR}$;
    \item If it is believed that $\text{SNR} < \tau_{GLR}$ holds, sample $m_{opt}$ (approximately $\sqrt{N}$) nodes at random and reconstruct from those nodes.
\end{enumerate}}

We finally remark that future work includes extending the analysis on GLR to the normalised graph Laplacian, providing bounds on $\xi_2$ for LS, analysing other graph models such as Ring graphs or studying a more detailed ``early-stopping'' mechanism in sequential sampling schemes that do not use the full sample budget.

\bibliographystyle{IEEEtran}
\bibliography{sample}

@Book{golub13,
  author = "Golub, Gene H. and Van Loan, Charles F.",
  title = "Matrix Computations",
  publisher = "John Hopkins University Press",
  year = 2013,
  edition = "4th",
  isbn = "978-1-4214-0794-4",
        }

@article{wang2018optimal,
  title={A-optimal sampling and robust reconstruction for graph signals via truncated neumann series},
  author={Wang, Fen and Wang, Yongchao and Cheung, Gene},
  journal={IEEE Signal Processing Letters},
  volume={25},
  number={5},
  pages={680--684},
  year={2018},
  publisher={IEEE}
  }

@article{wang2019low,
  title={Low-complexity graph sampling with noise and signal reconstruction via Neumann series},
  author={Wang, Fen and Cheung, Gene and Wang, Yongchao},
  journal={IEEE Transactions on Signal Processing},
  volume={67},
  number={21},
  pages={5511--5526},
  year={2019},
  publisher={IEEE}
}

@article{pesenson2008sampling,
  title={Sampling in Paley-Wiener spaces on combinatorial graphs},
  author={Pesenson, Isaac},
  journal={Transactions of the American Mathematical Society},
  volume={360},
  number={10},
  pages={5603--5627},
  year={2008}
}

@article{bai2020fast,
  title={Fast graph sampling set selection using gershgorin disc alignment},
  author={Bai, Yuanchao and Wang, Fen and Cheung, Gene and Nakatsukasa, Yuji and Gao, Wen},
  journal={IEEE Transactions on signal processing},
  volume={68},
  pages={2419--2434},
  year={2020},
  publisher={IEEE}
}

@article{chen2017GLRbias,
  title={Bias-variance tradeoff of graph laplacian regularizer},
  author={Chen, Pin-Yu and Liu, Sijia},
  journal={IEEE Signal Processing Letters},
  volume={24},
  number={8},
  pages={1118--1122},
  year={2017},
  publisher={IEEE}
}

@article{puy2018random,
  title={Random sampling of bandlimited signals on graphs},
  author={Puy, Gilles and Tremblay, Nicolas and Gribonval, R{\'e}mi and Vandergheynst, Pierre},
  journal={Applied and Computational Harmonic Analysis},
  volume={44},
  number={2},
  pages={446--475},
  year={2018},
  publisher={Elsevier}
}

@book{pukelsheim2006optimal,
  title={Optimal design of experiments},
  author={Pukelsheim, Friedrich},
  year={2006},
  publisher={SIAM}
}

@article{jayawant2021doptimal,
  title={Practical graph signal sampling with log-linear size scaling},
  author={Jayawant, Ajinkya and Ortega, Antonio},
  journal={Signal Processing},
  volume={194},
  pages={108436},
  year={2022},
  publisher={Elsevier}
}

@inproceedings{tremblay2017determinantal,
  title={Graph sampling with determinantal processes},
  author={Tremblay, Nicolas and Amblard, Pierre-Olivier and Barthelm{\'e}, Simon},
  booktitle={2017 25th European Signal Processing Conference (EUSIPCO)},
  pages={1674--1678},
  year={2017},
  organization={IEEE}
}

@article{EOptimalChen,
  title={Discrete Signal Processing on Graphs: Sampling Theory},
  author={Chen, Siheng and Varma, Rohan and Sandryhaila, Aliaksei and Kova{\v{c}}evi{\'c}, Jelena},
  journal={IEEE transactions on signal processing},
  volume={63},
  number={24},
  pages={6510--6523},
  year={2015},
  publisher={IEEE}
}

@article{mfn,
  title={Signals on graphs: Uncertainty principle and sampling},
  author={Tsitsvero, Mikhail and Barbarossa, Sergio and Di Lorenzo, Paolo},
  journal={IEEE Transactions on Signal Processing},
  volume={64},
  number={18},
  pages={4845--4860},
  year={2016},
  publisher={IEEE}
}

@article{tanaka2020sampling,
  title={Sampling signals on graphs: From theory to applications},
  author={Tanaka, Yuichi and Eldar, Yonina C and Ortega, Antonio and Cheung, Gene},
  journal={IEEE Signal Processing Magazine},
  volume={37},
  number={6},
  pages={14--30},
  year={2020},
  publisher={IEEE}
}

@article{anis2016efficient,
  title={Efficient sampling set selection for bandlimited graph signals using graph spectral proxies},
  author={Anis, Aamir and Gadde, Akshay and Ortega, Antonio},
  journal={IEEE Transactions on Signal Processing},
  volume={64},
  number={14},
  pages={3775--3789},
  year={2016},
  publisher={IEEE}
}

@book{gauss1823theoria,
  title={Theoria combinationis observationum erroribus minimis obnoxiae},
  author={Gauss, Carl-Friedrich},
  year={1823},
  publisher={Henricus Dieterich}
}

@inproceedings{shomorony2014sampling,
  title={Sampling large data on graphs},
  author={Shomorony, Han and Avestimehr, A Salman},
  booktitle={2014 IEEE Global Conference on Signal and Information Processing (GlobalSIP)},
  pages={933--936},
  year={2014},
  organization={IEEE}
}

@article{chamon2017greedy,
  title={Greedy sampling of graph signals},
  author={Chamon, Luiz FO and Ribeiro, Alejandro},
  journal={IEEE Transactions on Signal Processing},
  volume={66},
  number={1},
  pages={34--47},
  year={2017},
  publisher={IEEE}
}

@inproceedings{narang2013localized,
  title={Localized iterative methods for interpolation in graph structured data},
  author={Narang, Sunil K and Gadde, Akshay and Sanou, Eduard and Ortega, Antonio},
  booktitle={2013 IEEE Global Conference on Signal and Information Processing},
  pages={491--494},
  year={2013},
  organization={IEEE}
}

@article{ortega2018graph,
  title={Graph signal processing: Overview, challenges, and applications},
  author={Ortega, Antonio and Frossard, Pascal and Kova{\v{c}}evi{\'c}, Jelena and Moura, Jos{\'e} MF and Vandergheynst, Pierre},
  journal={Proceedings of the IEEE},
  volume={106},
  number={5},
  pages={808--828},
  year={2018},
  publisher={IEEE}
}

@book{bhatia2013matrix,
  title={Matrix analysis},
  author={Bhatia, Rajendra},
  volume={169},
  year={2013},
  publisher={Springer Science \& Business Media}
}

@article{chen2016signal,
  title={Signal recovery on graphs: Fundamental limits of sampling strategies},
  author={Chen, Siheng and Varma, Rohan and Singh, Aarti and Kova{\v{c}}evi{\'c}, Jelena},
  journal={IEEE Transactions on Signal and Information Processing over Networks},
  volume={2},
  number={4},
  pages={539--554},
  year={2016},
  publisher={IEEE}
}

@inproceedings{lin2019active,
  title={Active sampling for approximately bandlimited graph signals},
  author={Lin, Sijie and Xie, Xuan and Feng, Hui and Hu, Bo},
  booktitle={ICASSP 2019-2019 IEEE International Conference on Acoustics, Speech and Signal Processing (ICASSP)},
  pages={5441--5445},
  year={2019},
  organization={IEEE}
}

@inproceedings{itani2021graph,
  title={A graph signal processing framework for the classification of temporal brain data},
  author={Itani, Sarah and Thanou, Dorina},
  booktitle={2020 28th European Signal Processing Conference (EUSIPCO)},
  pages={1180--1184},
  year={2021},
  organization={IEEE}
}

@inproceedings{renoust2017estimating,
  title={Estimating political leanings from mass media via graph-signal restoration with negative edges},
  author={Renoust, Benjamin and Cheung, Gene and Satoh, Shin'Ichi},
  booktitle={2017 IEEE International Conference on Multimedia and Expo (ICME)},
  pages={1009--1014},
  year={2017},
  organization={IEEE}
}

@article{jain2014big,
  title={Big data+ big cities: Graph signals of urban air pollution [exploratory sp]},
  author={Jain, Rishee K and Moura, Jose MF and Kontokosta, Constantine E},
  journal={IEEE Signal Processing Magazine},
  volume={31},
  number={5},
  pages={130--136},
  year={2014},
  publisher={IEEE}
}

@article{nikolov2022proportional,
  title={Proportional volume sampling and approximation algorithms for A-optimal design},
  author={Nikolov, Aleksandar and Singh, Mohit and Tantipongpipat, Uthaipon},
  journal={Mathematics of Operations Research},
  volume={47},
  number={2},
  pages={847--877},
  year={2022},
  publisher={INFORMS}
}

@article{nordstrom2011convexity,
  title={Convexity of the inverse and Moore--Penrose inverse},
  author={Nordstr{\"o}m, Kenneth},
  journal={Linear algebra and its applications},
  volume={434},
  number={6},
  pages={1489--1512},
  year={2011},
  publisher={Elsevier}
}

@article{barahona2002synchronization,
  title={Synchronization in small-world systems},
  author={Barahona, Mauricio and Pecora, Louis M},
  journal={Physical review letters},
  volume={89},
  number={5},
  pages={054101},
  year={2002},
  publisher={APS}
}

@inproceedings{sripathmanathan2023impact,
  title={On the Impact of Sample Size in Reconstructing Graph Signals},
  author={Sripathmanathan, Baskaran and Dong, Xiaowen and Bronstein, Michael},
  booktitle={2023 International Conference on Sampling Theory and Applications (SampTA)},
  pages={1--6},
  year={2023},
  organization={IEEE}
}

@article{zhang2000schur,
  title={Schur complements and matrix inequalities in the L{\"o}wner ordering},
  author={Zhang, Fuzhen},
  journal={Linear Algebra and Its Applications},
  volume={321},
  number={1-3},
  pages={399--410},
  year={2000},
  publisher={Elsevier}
}

@article{jiang2012low,
  title={Low eigenvalues of Laplacian matrices of large random graphs},
  author={Jiang, Tiefeng},
  journal={Probability Theory and Related Fields},
  volume={153},
  pages={671--690},
  year={2012},
  publisher={Springer}
}

@article{householder1965kantorovich,
  title={The Kantorovich and some related inequalities},
  author={Householder, AS},
  journal={SIAM Review},
  volume={7},
  number={4},
  pages={463--473},
  year={1965},
  publisher={SIAM}
}

@article{merris1998laplacian,
  title={Laplacian graph eigenvectors},
  author={Merris, Russell},
  journal={Linear algebra and its applications},
  volume={278},
  number={1-3},
  pages={221--236},
  year={1998},
  publisher={Elsevier}
}

@article{geman1992neural,
  title={Neural networks and the bias/variance dilemma},
  author={Geman, Stuart and Bienenstock, Elie and Doursat, Ren{\'e}},
  journal={Neural computation},
  volume={4},
  number={1},
  pages={1--58},
  year={1992},
  publisher={MIT Press One Rogers Street, Cambridge, MA 02142-1209, USA journals-info~…}
}

@inproceedings{chamon2016near,
  title={Near-optimality of greedy set selection in the sampling of graph signals},
  author={Chamon, Luiz FO and Ribeiro, Alejandro},
  booktitle={2016 IEEE Global Conference on Signal and Information Processing (GlobalSIP)},
  pages={1265--1269},
  year={2016},
  organization={IEEE}
}

@inproceedings{nabar2023conservative,
  title={Conservative Predictions on Noisy Financial Data},
  author={Nabar, Omkar and Shroff, Gautam},
  booktitle={Proceedings of the Fourth ACM International Conference on AI in Finance},
  pages={427--435},
  year={2023}
}

@article{tanaka2020generalized,
  title={Generalized sampling on graphs with subspace and smoothness priors},
  author={Tanaka, Yuichi and Eldar, Yonina C},
  journal={IEEE Transactions on Signal Processing},
  volume={68},
  pages={2272--2286},
  year={2020},
  publisher={IEEE}
}

@inproceedings{hara2022sampling,
  title={Sampling set selection for graph signals under arbitrary signal priors},
  author={Hara, Junya and Tanaka, Yuichi},
  booktitle={ICASSP 2022-2022 IEEE International Conference on Acoustics, Speech and Signal Processing (ICASSP)},
  pages={5732--5736},
  year={2022},
  organization={IEEE}
}

@article{dong2016learning,
  title={Learning Laplacian matrix in smooth graph signal representations},
  author={Dong, Xiaowen and Thanou, Dorina and Frossard, Pascal and Vandergheynst, Pierre},
  journal={IEEE Transactions on Signal Processing},
  volume={64},
  number={23},
  pages={6160--6173},
  year={2016},
  publisher={IEEE}
}

@inproceedings{xie2019bayesian,
  title={Bayesian design of sampling set for bandlimited graph signals},
  author={Xie, Xuan and Yu, Junhao and Feng, Hui and Hu, Bo},
  booktitle={2019 IEEE Global Conference on Signal and Information Processing (GlobalSIP)},
  pages={1--5},
  year={2019},
  organization={IEEE}
}

@book{Goodfellow2016DeepLearning,
author = {Goodfellow, Ian and Bengio, Yoshua and Courville, Aaron},
address = {Cambridge, Massachusetts},
booktitle = {Deep learning},
copyright = {"The online version of the book ... will remain available online for free."},
isbn = {9780262035613},
language = {eng},
publisher = {MIT Press},
title = {Deep learning },
year = {2016},
}

@article{khatri1982some,
  title={Some generalizations of Kantorovich inequality},
  author={Khatri, CG and Rao, C Radhakrishna},
  journal={Sankhy{\=a}: The Indian Journal of Statistics, Series A},
  pages={91--102},
  year={1982},
  publisher={JSTOR}
}

@article{baksalary1991generalized,
  title={Generalized matrix versions of the Cauchy-Schwarz and Kantorovich inequalities},
  author={Baksalary, Jerzy K and Puntanen, Simo},
  journal={Aequationes Mathematicae},
  volume={41},
  pages={103--110},
  year={1991},
  publisher={Springer}
}

@article{zhi2023gaussian,
  title={Gaussian processes on graphs via spectral kernel learning},
  author={Zhi, Yin-Cong and Ng, Yin Cheng and Dong, Xiaowen},
  journal={IEEE Transactions on Signal and Information Processing over Networks},
  volume={9},
  pages={304--314},
  year={2023},
  publisher={IEEE}
}

@article{bush2023impact,
  title={The impact of COVID-19 public health restrictions on particulate matter pollution measured by a validated low-cost sensor network in Oxford, UK},
  author={Bush, Tony and Bartington, Suzanne and Pope, Francis D and Singh, Ajit and Thomas, G Neil and Stacey, Brian and Economides, George and Anderson, Ruth and Cole, Stuart and Abreu, Pedro and others},
  journal={Building and Environment},
  volume={237},
  pages={110330},
  year={2023},
  publisher={Elsevier}
}

@article{liu1997kantorovich,
  title={Kantorovich and Cauchy-Schwarz inequalities involving positive semidefinite matrices, and efficiency comparisons for a singular linear model},
  author={Liu, Shuangzhe and Neudecker, Heinz},
  journal={Linear algebra and its applications},
  volume={259},
  pages={209--221},
  year={1997},
  publisher={Elsevier}
}

@inproceedings{venkitaraman2020gaussian,
  title={Gaussian processes over graphs},
  author={Venkitaraman, Arun and Chatterjee, Saikat and Handel, Peter},
  booktitle={ICASSP 2020-2020 IEEE International Conference on Acoustics, Speech and Signal Processing (ICASSP)},
  pages={5640--5644},
  year={2020},
  organization={IEEE}
}

@article{behjat2016signal,
  title={Signal-adapted tight frames on graphs},
  author={Behjat, Hamid and Richter, Ulrike and Van De Ville, Dimitri and S{\"o}rnmo, Leif},
  journal={IEEE Transactions on Signal Processing},
  volume={64},
  number={22},
  pages={6017--6029},
  year={2016},
  publisher={IEEE}
}

@article{hamilton2017inductive,
  title={Inductive representation learning on large graphs},
  author={Hamilton, Will and Ying, Zhitao and Leskovec, Jure},
  journal={Advances in neural information processing systems},
  volume={30},
  year={2017}
}

@article{belkin2004semi,
  title={Semi-supervised learning on Riemannian manifolds},
  author={Belkin, Mikhail and Niyogi, Partha},
  journal={Machine learning},
  volume={56},
  pages={209--239},
  year={2004},
  publisher={Springer}
}
\clearpage

\appendices
\section{Theoretical Results for LS under $k$-Bandlimited Noise}
\label{app:LS_bandlimited}
We have observed that MSE can decrease when sample size decreases under LS and full band noise. This raises the question of whether the decrease is caused by some sort of interference effect between the high-frequency components of the noise and the bandlimited (low-frequency) signal. In this subsection, by showing that MSE can decrease when sample size decreases under LS with $k$-bandlimited noise, we demonstrate that this is not the case.

\subsubsection{Simplification}
For coherence with the full-band case, we use the single vertex simplification.
\subsubsection{Characterisation}
We first show that under $k$-bandlimited noise the choice of sample set $\set{S}$ only influences the MSE through $\textrm{rank}(\matrsubU{S})$.
\begin{lemma}
\label{lemma:LS_bandlimited_noise_MSE}
$        \textrm{MSE}_{\set{S}} = k + (\sigma^2 - 1)\textrm{rank}(\matrsubU{S}) $.
\end{lemma}

\noindent Lemma \ref{lemma:LS_bandlimited_noise_MSE} makes it easy to prove the following variants of Corollary \ref{main_ls} and Theorem \ref{thm:noiseless_optimality_means_noise_sensitivity} for bandlimited noise.

\begin{corollary}
\label{corr:LS_bandlimited_noise_big_variance}
Let $v \in \set{S}$, then $\set{S}\backslash \{v\}$ is as good or better than $\set{S}$ if and only if

\begin{equation}
\textrm{SNR} \leq 1.
\end{equation}

\end{corollary}
We can understand this as being like Corollary \ref{main_ls} with $\tau_{LS\_bl} = 1$.
\noindent The criterion in Corollary \ref{corr:LS_bandlimited_noise_big_variance} does not depend on vertex choice, unlike Corollary \ref{main_ls}. Therefore under LS and bandlimited noise, whether $\set{S} \backslash \{v\}$ is as good or better than $\set{S}$ is not contingent on which vertices are in $\set{S}$.

\subsubsection{Existence}
Next, our variant of Theorem \ref{thm:noiseless_optimality_means_noise_sensitivity} for bandlimited noise concerns when MSE changes at all, rather than when it reduces.
\begin{corollary}
\label{corr:LS_bandlimited_noise_sample_only_k}
    Suppose we use a sequential noiseless-optimal scheme to select a vertex sample set $\set{S}_{m}$ of size $m$. For $m \leq k$: 
\begin{equation}
    \forall v \in \set{S}_{m}: \quad \tau(\set{S}_{m},v) = \tau_{LS\_bl} = 1,
\end{equation}
i.e., for \emph{any} $v \in \set{S}_{m}$, $\set{S}_{m} \backslash \{v\}$ is better than $\set{S}_{m}$ if and only if $\textrm{SNR} < \tau_{LS\_bl}$.
For $m > k$: 
    \begin{equation}
        \forall \text{SNR}, \enskip \forall v \in \set{S}_{m}: \quad \textrm{MSE}_{\set{S}_{m}} = \textrm{MSE}_{\set{S}_{m} \backslash \{ v \}}
    \end{equation}
\end{corollary}

\noindent We prove Lemma \ref{lemma:LS_bandlimited_noise_MSE} and Corollaries \ref{corr:LS_bandlimited_noise_big_variance} \& \ref{corr:LS_bandlimited_noise_sample_only_k} in Appendix \ref{app:LS_Bandlimited_Noise_proofs}.

{%
\subsubsection{Asymptotics}
Finally, we discuss the large graph case.
\begin{remark}
\label{rmk:LS_big_N_bl}
    As $\tau_{LS\_bl}$ is not a function of the graph, we find that at a fixed SNR $< \tau_{LS\_bl} = 1$, for graphs of any size (even arbitrarily large), we can reduce sample size to improve MSE.
\end{remark}
}

\section{Theoretical Results for GLR under $k$-Bandlimited Noise}
\label{app:GLR_bandlimited}
Once more, one might ask whether the MSE increasing with sample size under GLR is caused by some sort of interference effect between the high-frequency components of the noise and the bandlimited (low-frequency) signal. { We present a variant of Theorem \ref{thm:main_GLR_exist} to disprove this.

\subsubsection{Simplification} As with the full-band GLR case, we consider when a set $\set{S} \subset \set{N}$ is better than $\set{N}$.
\subsubsection{Characterisation}
We first note that Corollary \ref{corr:main_GLR_iff} also applies to the bandlimited case with the appropriate definition of $\Delta_{2}$, providing a full characterisation of when a set $\set{S}$ is better than $\set{N}$ under GLR reconstruction with bandlimited noise.

\subsubsection{Existence}
 We now present a bandlimited variant of Lemma \ref{lemma:GLR_xi_2_bound_main}. It can be understood as Lemma \ref{lemma:GLR_xi_2_bound_main} with $\lambda_{N}$ replaced by $\lambda_{k}$.

\begin{lemma}
\label{lemma:GLR_xi_2_bound_main_bl}
    Let $\lambda_{i}$ be the eigenvalues of $\matr{L}$ and let $\omega$ be defined as in Lemma \ref{lemma:GLR_xi_2_bound_main}. Let
    \begin{align}
        r_{bl} &= \omega\left( \frac{\lambda_{k}}{\lambda_{2}}\right) \\
        B_{k}(m) &= r_{bl} \frac{N}{m} + \sum^{m}_{i=2}\omega\left(\max\left[1,\frac{\lambda_{k+2-i}}{\lambda_{i}}\right]\right)
    \end{align}
    where the summation is 0 if $m=0$ or 1. Then
    \begin{align}
        \xi_{2,bl}(\set{S}) &\leq 1 + B_{k}(m) \leq 1 + r_{bl} \left( \frac{N}{m} + m - 1 \right).
    \end{align}
\end{lemma}
\begin{proof}
    We prove a bandlimited variant of the Kantorovich inequality and use the same structure as the proof of Lemma \ref{lemma:GLR_xi_2_bound_main}. See Appendix \ref{app:GLR_bandlimited_lemma_proof}.
\end{proof}

A version of Theorem \ref{thm:main_GLR_exist} follows in the same manner as with full-band noise:

\begin{theorem}
\label{thm:main_GLR_bl}
{ Let $B_{k}(m)$ be defined as in Lemma \ref{lemma:GLR_xi_2_bound_main_bl}.  Let $m_{opt\_bl}$ be the sample size minimising $B_{k}(m)$. 
    \begin{flalign}
        &\text{If} \hspace{0.35\columnwidth} B_{k}(m_{opt\_bl}) < k-1& \label{eq:GLR_exist_thm_B_constraint_bl}
    \end{flalign}
then $m_{opt\_bl} \leq \lceil \frac{N+1}{2} \rceil$. Furthermore $\exists \mu_{ub\_bl} > 0$  s.t. under GLR with parameter $\mu \in (0, \mu_{ub\_bl})$, $\exists \tau_{GLR\_bl}(\mu) > 0$ where
\begin{flalign}
        &\text{if} \hspace{0.3\columnwidth} \textrm{SNR} < \tau_{GLR\_bl}(\mu)&
    \end{flalign}

    \noindent then \emph{any} sample set $\set{S}$ of size $m_{opt\_bl}$ is better than $\set{N}$.}
\end{theorem}
\begin{proof}
    Apply Lemma \ref{lemma:GLR_xi_2_bound_main_bl} and follow similar proof steps to Theorem \ref{thm:main_GLR_exist}. See Appendix \ref{app:GLR_bandlimited_thm} for details.
\end{proof}

Given the similarity of the bounds and proof approaches between Theorems \ref{thm:main_GLR_exist} and \ref{thm:main_GLR_bl}, the same analysis of parameters applies to Theorem \ref{thm:main_GLR_bl} as Theorem \ref{thm:main_GLR_exist}. Thus we expect the parameters in Theorem \ref{thm:main_GLR_bl} -- $m_{opt\_bl}$, $B_{k}(m_{opt\_bl})$, $\mu_{ub\_bl}$ and $\tau_{GLR\_bl}$ to behave similarly to their non-bandlimited counterparts. This is reinforced by the asymptotics, which we will see in Proposition \ref{propn:GLR_big_N_bl}.

We again present an MSE upper bound to link our experiments in Section \ref{sec:experiments} to Fig \ref{fig:GLR_diagram}:
\begin{corollary}
\label{corr:unif_ub_xi_1_MSE_bl}
    Let $B_{k}(m)$ be defined as in Lemma \ref{lemma:GLR_xi_2_bound_main_bl}. For a sample set $\set{S}$ of size $m$,
    \begin{align}
    \textrm{MSE}_{\set{S}} &\leq (k - 1) + (1 + \sigma^{2}) \cdot (1 + B_{k}(m)). \label{eq:unif_ub_MSE}
    \end{align}
\end{corollary}
\begin{proof}
By Lemma \ref{lemma:GLR_xi_2_bound_main_bl}, $\xi_{2}(\set{S}) \leq B(m)$. By Lemma \ref{lemma:unif_ub_xi_1_GLR} in Appendix \ref{app:proof_unif_ub_xi_1_MSE}, $\xi_{1}(\set{S}) \leq (k-1) + \xi_{2,bl}(\set{S}) \leq k + B_{k}(m)$. Combining these using (\ref{eq:xi_decomp}) gives the desired bound.
\end{proof}
}

\subsubsection{Asymptotics}
Finally, we consider the case of large graphs. We prove a variant of Proposition \ref{propn:GLR_big_N}.
\begin{propn}
\label{propn:GLR_big_N_bl}
As $N \to \infty$, condition (\ref{eq:GLR_exist_thm_B_constraint_bl}) in Theorem \ref{thm:main_GLR_bl} holds w.h.p. . Furthermore,
\begin{align}
        r &\overset{p}{\to} 1, &
        {m_{opt\_bl}}\cdot{{N}}^{-\frac{1}{2}} &\overset{p}{\to} 1, &
    \mu_{ub\_bl}\lambda_{2} &\overset{p}{\to} +\infty.
\end{align}
    Assume $\frac{k}{N}$ is fixed and choose $\mu = \frac{c}{\lambda_{2}}$,  or $\frac{c}{\lambda_{N}}$, or $\frac{c}{\sqrt{\lambda_{2}\lambda_{N}} }$ for optimal bias-variance trade-off 
 at $\set{S} = \set{N}$ \cite{chen2017GLRbias}, then
    \begin{align}
        \tau_{GLR\_bl} \to (1+2c)^{-1}.
    \end{align}
\end{propn}
\begin{proof}
    See Appendix \ref{app:Proof_GLR_big_N_bl}.
\end{proof}

{ We see that the parameters asymptotically behave the same as the non-bandlimited parameters for large Erdős–Rényi graphs.}

\section{Experimental Results under $k$-Bandlimited Noise}
\label{app:Experiments_Bandlimited}
In this section we provide results on synthetic datasets under $k$-bandlimtied noise.

\begin{figure*}%
    \centering
    \begin{subfigure}{0.6\columnwidth}
    \resizebox{\width}{0.62\columnwidth}{
    \includegraphics[width=\columnwidth]{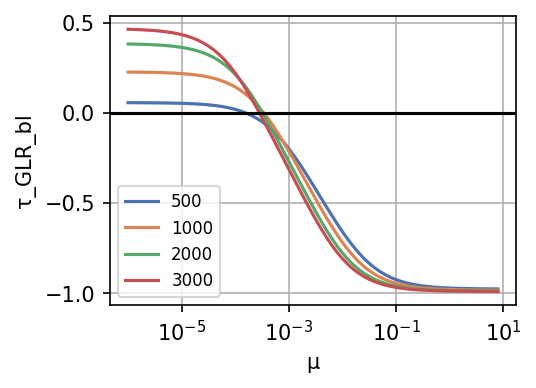}}
    \caption{Erdős–Rényi ($\tau_{GLR\_bl}$)}
    \label{tau_GLR_bl_er}
    \end{subfigure}
    \hfill
    \begin{subfigure}{0.6\columnwidth}
    \resizebox{\width}{0.62\columnwidth}{
    \includegraphics[width=\columnwidth]{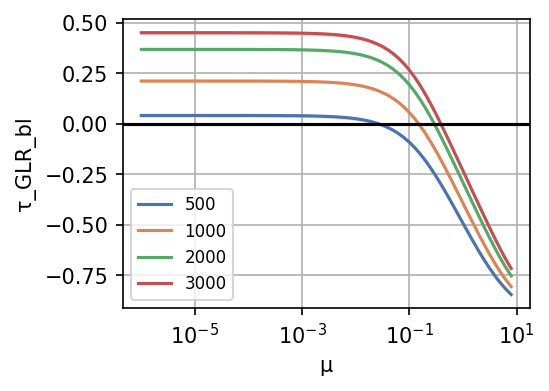}}
    \caption{Barabási-Albert ($\tau_{GLR\_bl}$)}%
    \label{tau_GLR_bl_BA}%
    \end{subfigure}
    \hfill%
    \begin{subfigure}{0.6\columnwidth}
    
    \resizebox{\width}{0.62\columnwidth}{
    \includegraphics[width=\columnwidth]{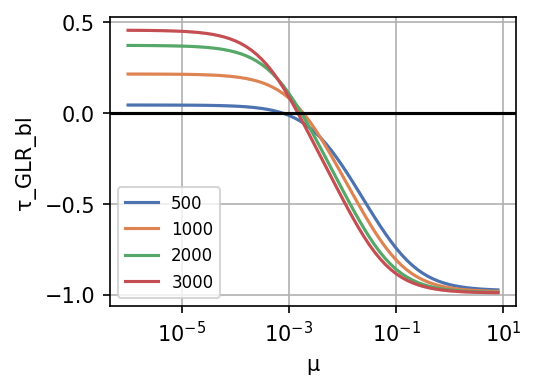}}
    \caption{SBM ($\tau_{GLR\_bl}$)}%
    \label{tau_GLR_bl_SBM}%
    \end{subfigure}%
    \caption{$\tau_{GLR\_bl}$ for different random graph models (\#vertices = colour, bandwidth = $\frac{\text{\# vertices}}{10}$).}
\label{GLR_Threshold_plots_bl}
\end{figure*}

\begin{figure*}%
    \centering
    \begin{subfigure}{0.6\columnwidth}
    \resizebox{\width}{0.62\columnwidth}{
    \includegraphics[width=\columnwidth]{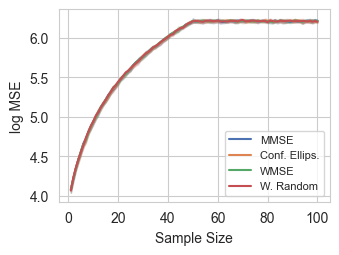}}
    \caption{Bandlimited noise, SNR = $10^{-1}$}
    \label{bandlimited_MSE_subfiga}
    \end{subfigure}\hfill
    \begin{subfigure}{0.6\columnwidth}
    \resizebox{\width}{0.62\columnwidth}{
    \includegraphics[width=\columnwidth]{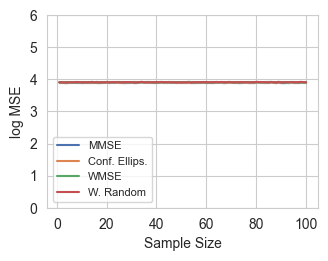}}
    \caption{Bandlimited noise, SNR = $1$}%
    \label{bandlimited_MSE_subfigb}%
    \end{subfigure}\hfill%
    \begin{subfigure}{0.6\columnwidth}
    \resizebox{\width}{0.62\columnwidth}{
    \includegraphics[width=\columnwidth]{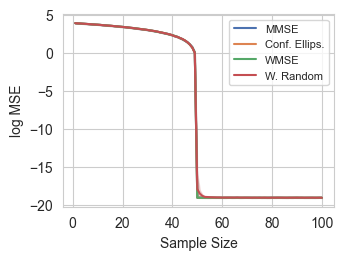}}
    \caption{Bandlimited noise, SNR = $10^{10}$}%
    \label{bandlimited_MSE_subfigc}%
    \end{subfigure}%
    \caption{Average MSE under LS on ER graphs (\#vertices=500, bandwidth = 50).}
\label{LS_ER_MSE_fig_bl}
\end{figure*}

\begin{figure*}%
    \centering
    \begin{subfigure}{0.6\columnwidth}
    \resizebox{\width}{0.62\columnwidth}{
    \includegraphics[width=\columnwidth]{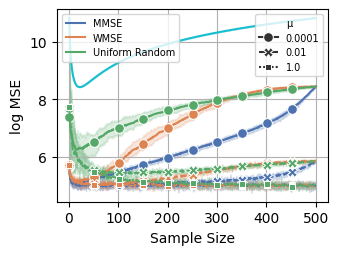}}
    \caption{Bandlimited noise, SNR = $10^{-2}$}
    \label{bandlimited_GLR_MSE_subfiga}
    \end{subfigure}\hfill
    \begin{subfigure}{0.6\columnwidth}
    \resizebox{\width}{0.62\columnwidth}{
    \includegraphics[width=\columnwidth]{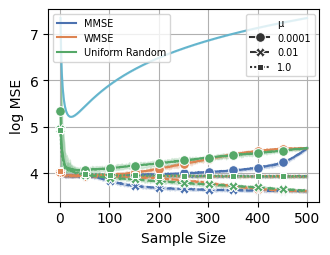}}
    \caption{Bandlimited noise, SNR = $\frac{1}{2}$}%
    \label{bandlimited_GLR_MSE_subfigb}%
    \end{subfigure}\hfill%
    \begin{subfigure}{0.6\columnwidth}
    \resizebox{\width}{0.62\columnwidth}{
    \includegraphics[width=\columnwidth]{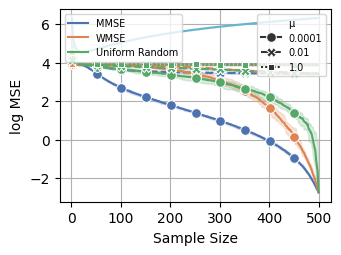}}
    \caption{Bandlimited noise, SNR = $10^{10}$}%
    \label{bandlimited_GLR_MSE_subfigc}%
    \end{subfigure}%
    \caption{Average MSE under GLR on ER graphs (\#vertices=500, bandwidth = 50). Line without markers is an upper bound.}
\label{GLR_ER_MSE_fig_bl}
\end{figure*}

\subsection{\texorpdfstring{$\tau$ and $\tau_{GLR}$ plots}{\texttau and \texttau\_GLR plots}}

\subsubsection{LS}
We do not provide a plot as in this case, as our theoretical results provide the magnitude of $\tau$: by Corollary \ref {corr:LS_bandlimited_noise_sample_only_k}, $\tau(\set{S},v) = 1$ for $|\set{S}| \leq k$ under sequential noiseless-optimal sampling. Corollary \ref{corr:LS_bandlimited_noise_sample_only_k} overall says something stronger, i.e. we can always reduce MSE by reducing sample size if $\text{SNR} < 1$ at any sample size.

\subsubsection{GLR}
Our purpose in showing bandlimited variants is to show that reducing sample size can reduce MSE even under bandlimited noise. Figs. \ref{tau_GLR_bl_er}- \ref{tau_GLR_bl_SBM} show the same overall trend as Figs. \ref{tau_GLR_er}-\ref{tau_GLR_SBM}, i.e. $\tau_{GLR\_bl}$ can be positive for ER, { BA} and SBM graphs, giving evidence to our claim. As with the full-band case, Fig. \ref{tau_GLR_bl_er} validates our asymptotic result for ER graphs (Proposition \ref{propn:GLR_big_N_bl}).

\subsection{MSE plots}
The MSE plots demonstrate the validity of our theoretical results linking MSE and sample size.

\subsubsection{LS}
Figs. \ref{bandlimited_MSE_subfiga}-\ref{bandlimited_MSE_subfigc} show how MSE behaves as sample size increases. Fig. \ref{bandlimited_MSE_subfiga} shows that under high noise, MSE increases with sample size when sample size is no larger than bandwidth, and is constant beyond that. Fig. \ref{bandlimited_MSE_subfigb} shows at an SNR of $\tau_{LS\_bl} = 1$, MSE is constant with sample size. Finally, Fig. \ref{bandlimited_MSE_subfigc} shows that in the almost noiseless case MSE decreases with sample size. Fig. \ref{bandlimited_MSE_subfiga} demonstrates Corollary \ref{corr:LS_bandlimited_noise_big_variance} by showing that for $\textrm{SNR} < \tau_{LS\_bl}=1$ the MSE is increasing with sample size. In all cases, MSE remains unchanged for sample sizes exceeding the bandlimit $k$, demonstrating Corollary \ref{corr:LS_bandlimited_noise_sample_only_k}.

\subsubsection{GLR}
Our observations about Figs. \ref{GLR_MSE_subfiga}-\ref{GLR_MSE_subfigc}  and Theorem \ref{thm:main_GLR_exist} also apply to Figs. \ref{bandlimited_GLR_MSE_subfiga}-\ref{bandlimited_GLR_MSE_subfigc} and Theorem \ref{thm:main_GLR_bl}; specifically, for sufficiently small $\mu$ and SNR, MSE is minimised at sample sizes well below $N$ under bandlimited noise and $\tau_{GLR\_bl}$ is a lower bound for $\tau(\set{N},\set{S}^{c})$ rather than an exact characterisation.

\subsection{Checking conditions}
As in the full-band case, our theorems on GLR rely on conditions around graph invariants.  We sample graphs from each random graph model to empirically show the probability the conditions of Theorem \ref{thm:main_GLR_bl} are met at a sample size of $m_{opt}$ for some $\mu > 0$:
\begin{table}[h]
\caption{Probability theorem conditions are met.}
    \begin{center}
        \begin{tabular}{|l|c|c|c|}
     \hline
       & \textbf{ER} & \textbf{SBM} & \textbf{BA} \\ 
     \hline
     Theorem \ref{thm:main_GLR_bl} conditions met & 100\% & 100\% & 99.4\% \\ \hline
        \end{tabular}
    \end{center}
    \label{tbl:empirical_probabilities_conditions_bl}
\end{table}

As with full-band noise, Proposition \ref{propn:GLR_big_N_bl} shows the conditions in Theorem \ref{thm:main_GLR_bl} hold w.h.p. for ER graphs as $N \to \infty$, and Table \ref{tbl:empirical_probabilities_conditions_bl} shows empirically that the conditions hold under $k$-bandlimited noise noise (Theorem \ref{thm:main_GLR_exist}) for all ER and SBM graphs tested with $500$ vertices, and most BA graphs tested.

\section{Bias-Variance Decomposition}
\label{app:bias-variance}
We present a proof of the Bias-Variance decomposition we use. We assume that $\vect{x}$ and $\vect{\epsilon}$ are independent. Let $\hat{\vect{x}}$ be any reconstruction of the signal $\vect{x}$, then
\begin{align}
    \textrm{MSE}_{\set{S}} &= \expect[\vect{x}]{\expect[\vect{\epsilon}]{\sqnormvec{\vect{x} - \hat{\vect{x}}}}}\\
    &= \expect[\vect{x}]{\expect[\vect{\epsilon}]{\sqnormvec{\vect{x} - \expect[\vect{\epsilon}]{\hat{\vect{x}}} + \expect[\vect{\epsilon}]{\hat{\vect{x}}} - \hat{\vect{x}}}}}\\
    &= \expect[\vect{x}]{2 {\expect[\vect{\epsilon}]{ {\left(\vect{x} - \expect[\vect{\epsilon}]{\hat{\vect{x}}} \right)^{T}} {\left(\expect[\vect{\epsilon}]{\hat{\vect{x}}} - \hat{\vect{x}}\right) }}} } \label{eq:bias_variance_decomp:cross_terms} \\
    &+ \expect[\vect{x}]{{\sqnormvec{\vect{x} - \expect[\vect{\epsilon}]{\hat{\vect{x}}}}} + {\expect[\vect{\epsilon}]{\sqnormvec{ \expect[\vect{\epsilon}]{\hat{\vect{x}}} - \hat{\vect{x}}}}}} \label{eq:bias_var_decomp1}.
\end{align}

\noindent Note that by the properties of expectation,
\begin{flalign}
 &&   \expect[\vect{\epsilon}]{\expect[\vect{\epsilon}]{\hat{\vect{x}}}^{T}{\hat{\vect{x}}}} = {\expect[\vect{\epsilon}]{\hat{\vect{x}}}^{T}\expect[\vect{\epsilon}]{\hat{\vect{x}}}} &= \expect[\vect{\epsilon}]{\expect[\vect{\epsilon}]{\hat{\vect{x}}}^{T}\expect[\vect{\epsilon}]{\hat{\vect{x}}}} \\
 \text{so} &&  \expect[\vect{\epsilon}]{ { \expect[\vect{\epsilon}]{\hat{\vect{x}}}^{T}} {\left(\expect[\vect{\epsilon}]{\hat{\vect{x}}} - \hat{\vect{x}}\right) }} &= 0
\end{flalign}
and as the value of the signal $\vect{x}$ is independent to the value of $\vect{\epsilon}$, 
\begin{align}
 \expect[\vect{\epsilon}]{ {\vect{x}^{T}} {\left(\expect[\vect{\epsilon}]{\hat{\vect{x}}} - \hat{\vect{x}}\right) }} = \vect{x}^{T}\expect[\vect{\epsilon}]{  {\expect[\vect{\epsilon}]{\hat{\vect{x}}} - \hat{\vect{x}} }} = 0 \\
 {\expect[\vect{\epsilon}]{ {\left(\vect{x} - \expect[\vect{\epsilon}]{\hat{\vect{x}}} \right)^{T}} {\left(\expect[\vect{\epsilon}]{\hat{\vect{x}}} - \hat{\vect{x}}\right) }}} = 0
\end{align}
for any value of $\vect{x}$. Therefore
\begin{equation}
    \textrm{MSE}_{\set{S}} = \expect[\vect{x}]{\underbrace{\sqnormvec{\vect{x} - \expect[\vect{\epsilon}]{\hat{\vect{x}}}}}_{\text{Bias}(\hat{\vect{x}},\vect{x})^{2}} + \underbrace{\expect[\vect{\epsilon}]{\sqnormvec{ \expect[\vect{\epsilon}]{\hat{\vect{x}}} - \hat{\vect{x}}}}}_{\text{Var}(\hat{\vect{x}})}}.
\end{equation}

Note that if $\vect{x}$ was fixed (which can be understood as being drawn from a distribution with one possible value), $\vect{\epsilon}$ and $\vect{x}$ must be independent, and therefore
\begin{align}    
\textrm{MSE}_{\set{S}}\big\vert_{\vect{x}} &= \expect[\vect{\epsilon}]{\sqnormvec{\hat{\vect{x}} - \vect{x}} \smallskip \mid \enspace \set{S} \textrm{ observed} } \\
&= {\underbrace{\sqnormvec{\vect{x} - \expect[\vect{\epsilon}]{\hat{\vect{x}}}}}_{\text{Bias}(\hat{\vect{x}},\vect{x})^{2}} + \underbrace{\expect[\vect{\epsilon}]{\sqnormvec{ \expect[\vect{\epsilon}]{\hat{\vect{x}}} - \hat{\vect{x}}}}}_{\text{Var}(\hat{\vect{x}})}}.
\end{align}

\section{Extending Results to Other Settings - Proof of Proposition \ref{propn:averages_generalise_to_forall}}
\label{app:every_x}
In this proof we refer to the distribution of $\vect{x}$ as the signal model and the distribution of $\vect{\epsilon}$ as the noise model.

We apply Theorem \ref{main_general}; we note that all results in Section \ref{sec:general} don't require $\expect{\vect{x}} = 0$, so apply when $\vect{x}$ is a fixed nonzero signal.  By (\ref{eq:def_var_gen}),
\begin{align}
    \Delta_{2}(\set{S},\set{T}) &= \xi_{2}(\set{S}) - \xi_{2}(\set{S} \backslash \set{T}) \\
    &= \sqfrob{\matr{R}_{\set{S}}\vectsub[S]{\epsilon}} - \sqfrob{\matr{R}_{\set{S}}\msubgen[\set{S} \backslash \set{T}]{\vect{\epsilon}}}
\end{align}
therefore $\Delta_{2}(\set{S},\set{T})$ is not a function of the signal model. %

Pick a signal model $\vect{x}$ with  $\expect[\vect{x}]{\sqnormvec{\vect{x}}} < \infty$. Then let
\begin{align}
    \tau'(\set{S},\set{T}) = \begin{dcases}
  \frac{\expect{\sqnormvec{\vect{x}}}} {\expect{\sqnormvec{\vect{\epsilon}}}} \cdot \frac{\Delta_2(\set{S}, \set{T})}{- \Delta_1(\set{S}, \set{T})}  &\text{if }\Delta_{1} < 0 \\
    \infty &\text{otherwise} 
    \end{dcases}
\end{align}

As ${\Delta_2(\set{S}, \set{T})} > 0$, $\tau'(\set{S},\set{T}) > 0$.

We now apply Theorem \ref{main_general}. We look at the three cases
\subsubsection{$\Delta_{1} < 0$}
In this case, $\tau'$ is $\tau$ in Theorem \ref{main_general}, and the Theorem directly gives the proposition.
\subsubsection{$\Delta_{1} > 0$}
This case in Theorem \ref{main_general} results in negative $\tau$, and thus $\text{SNR} \geq 0 > \tau $ is always fulfilled, and is equivalent to $\text{SNR} < \infty$. Thus we set $\tau' = \infty$.

\subsubsection{$\Delta_{1} = 0$}
As $\Delta_{2} > 0$ by assumption, this case always holds and reduces to the case where $\Delta_{1} > 0$.

Therefore under this new signal model if $\text{SNR} < \tau'(\set{S},\set{T})$, $\text{MSE}_{\set{S}} > \text{MSE}_{\set{S} \backslash \set{T}}$.

\section{Proof of Table \ref{tbl:Delta_Behaviour}}
\label{app:table_delta_proof}
{%
For GLR reconstruction, Table \ref{tbl:GLR_Reconstruction_Delta} does not rule out any option. In experiments with Erdős–Rényi graphs, all options marked as $\checkmark$ in Figure \ref{tbl:GLR_Reconstruction_Delta} can be observed at some sample size without too much difficulty under random or WMSE sampling.

The options marked $\sim$ do not turn up in such experiments, but may happen for very specific values of $\mu$ dependent on the graph and sampling scheme.

For LS reconstruction, we decompose the pattern in Table \ref{tbl:LS_Reconstruction_Delta} into the following statements: 
\begin{itemize}
    \item $\Delta_{1} \leq 0$
    \item $\Delta_{1} < 0 $ if and only if $ \Delta_{2} > 0$
\end{itemize}
}
\subsection{Under LS reconstruction, \texorpdfstring{$\Delta_{1} \leq 0$}{\textDelta\textoneinferior =< 0}}
\label{app:delta_1_non_positive}
For LS we have:  \[\matr{R}_{\set{S}} = \matrsubU{N}[\matr{U}]_{\set{S},\set{K}}^{\dagger}.\]

\begin{lemma}
    \label{lemma:LS_xi_1_is_rank}
    For LS, 
    \begin{align}
        \xi_{1}(\set{S}) = k - \text{rank}([\matr{U}]_{\set{S},\set{K}}),\label{eq:Lemma_xi_1_is_rank:rk}\\
        \Delta_{1}(\set{S},v) \in \{0, -1\} \label{eq:Lemma_xi_1_is_rank:discrete}.
    \end{align}
\end{lemma}
\begin{proof}
As $\sqfrob{\matrsubU{N}\matr{A}} = \sqfrob{\matr{A}}$ for any matrix $\matr{A}\in \mathbb{R}^{k \times k}$, 
    \begin{align}
        \xi_{1}(\set{S}) &= || \matrsubU{N} - \matr{R}_{\set{S}}[\matr{U}]_{\set{S},\set{K}} ||^{2}_{F} \\
        &= || \matrsubU{N} - \matrsubU{N}([\matr{U}]_{\set{S},\set{K}})^{\dagger}[\matr{U}]_{\set{S},\set{K}} ||^{2}_{F} \\
        &= || \matr{I}_{k} -([\matr{U}]_{\set{S},\set{K}})^{\dagger}[\matr{U}]_{\set{S},\set{K}}||^{2}_{F}
    \end{align}
    Let $\matr{\Pi} = ([\matr{U}]_{\set{S},\set{K}})^{\dagger}[\matr{U}]_{\set{S},\set{K}}$. $\matr{\Pi}$ is of the form $\matr{A}^{\dagger}\matr{A}$, so is a symmetric orthogonal projection onto the range of $([\matr{U}]_{\set{S},\set{K}})^{T}$ \cite[p.~290]{golub13}. Orthogonal projections are idempotent ($\matr{\Pi} = \matr{\Pi}^{2}$) hence have eigenvalues which are $0$ or $1$, and therefore $\text{tr}(\matr{\Pi}) = \text{rank}(([\matr{U}]_{\set{S},\set{K}})^{T}) = \text{rank}([\matr{U}]_{\set{S},\set{K}})$. We then have:
    \begin{align}
        \xi_{1}(\set{S}) &= ||\matr{I}_{k} - \matr{\Pi}||^{2}_{F} \\
        &= \text{tr}((\matr{I}_{k} - \matr{\Pi})(\matr{I}_{k} - \matr{\Pi})^{T})\\
        & =\text{tr}((\matr{I}_{k} - \matr{\Pi})(\matr{I}_{k} - \matr{\Pi})) \\
        &= \text{tr}(\matr{I}_{k} - 2\matr{\Pi} + \matr{\Pi}^{2}) \\
        &= \text{tr}(\matr{I}_{k} - \matr{\Pi}) \\
        &= \text{tr}(\matr{I}_{k}) - \text{tr}(\matr{\Pi}) \\
        &= k - \text{rank}([\matr{U}]_{\set{S},\set{K}})
    \end{align}
    proving (\ref{eq:Lemma_xi_1_is_rank:rk}). We now prove (\ref{eq:Lemma_xi_1_is_rank:discrete}). Removing a vertex from $\set{S}$ removes a row from $[\matr{U}]_{\set{S},\set{K}}$, reducing the rank by $0$ or $1$, so 
\begin{align}
\Delta_{1}(\set{S},v) &= \xi_{1}(\set{S}) - \xi_{1}(\set{S} \backslash \{v\})  \\
&= - \text{rank}([\matr{U}]_{\set{S},\set{K}}) + \text{rank}([\matr{U}]_{\set{S} \backslash \{v\}, \set{K}})  \\
&\in \{0, -1\}. 
\end{align} 
\end{proof}

Non-positivity of $\Delta_{1}$ immediately follows from Lemma \ref{lemma:LS_xi_1_is_rank}

\subsection{Under LS reconstruction, \texorpdfstring{$\Delta_{1} < 0$ if and only if $ \Delta_{2} > 0$
}{\textDelta\textoneinferior < 0 if and only if \textDelta\texttwoinferior > 0}}
\label{app:Delta_LS_opposite_signs}
We first need the following lemmas.

\begin{lemma}
\label{lemma:simplify_xi_2}
\begin{equation}
    \xi_{2}(\set{S}) = \sum_{\lambda_{i}^{\set{S}} \neq 0} { \frac{1}{\lambda_{i}^{\set{S}}}  }
\end{equation}
    where $\lambda_{i}^{\set{S}}$ is the $i^{th}$ eigenvalue of $\projbl[S]$.
\end{lemma}
\begin{proof} As
\begin{align}
\xi_{2}(\set{S}) = ||\matr{R}_{\set{S}}||^{2}_{F}
= ||\matrsubU{N}[\matr{U}]_{\set{S},\set{K}}^{\dagger}||^{2}_{F} = ||[\matr{U}]_{\set{S},\set{K}}^{\dagger}||^{2}_{F} \label{eq:xi_2_is_U_pseudo_frob}
\end{align}
    $\xi_{2}(\set{S})$ is the sum of the squares of the singular values of $([\matr{U}]_{\set{S},\set{K}})^{\dagger}$ \cite[Corollary 2.4.3]{golub13}. The pseudoinverse maps the singular values of $[\matr{U}]_{\set{S},\set{K}}$ onto the singular values of $([\matr{U}]_{\set{S},\set{K}})^{\dagger}$ in the following way \cite[Section 5.5.2]{golub13}:
    \begin{equation}
    \sigma_{i}(([\matr{U}]_{\set{S},\set{K}})^{\dagger}) =
        \begin{cases}
            0 &\textrm{if }\sigma_{i}([\matr{U}]_{\set{S},\set{K}}) = 0 \\
            \sigma_{i}([\matr{U}]_{\set{S},\set{K}})^{-1} &\textrm{otherwise}
        \end{cases}
    \end{equation}
    and the squares of the singular values of $[\matr{U}]_{\set{S},\set{K}}$ are $\lambda_{i}$ \cite[Eq. (8.6.1)]{golub13}. As $[\matr{U}]_{\set{S},\set{K}}[\matr{U}]_{\set{S},\set{K}}^{T} = \projbl[S]$, summing the singular values gives the result.
\end{proof}

\begin{lemma}
\label{lemma:square_to_rect_rank}
\begin{equation}
\rank{\projbl[S]} =  \rank{\matrsubU{S}} \leq k. 
\end{equation}
\end{lemma}
\begin{proof}
Remember that $[\matr{\Pi}_{bl(\set{K}}]_{\set{S}}= [\matr{U}]_{\set{S},\set{K}}[\matr{U}]_{\set{S},\set{K}}^T$.

The equality: The number of strictly positive singular values of a matrix is its rank \cite[Corollary 2.4.6]{golub13} and both $[\matr{\Pi}_{bl(\set{K}}]_{\set{S}}= [\matr{U}]_{\set{S},\set{K}}[\matr{U}]_{\set{S},\set{K}}^T$ and $[\matr{U}]_{\set{S},\set{K}}$ have the same number of strictly positive singular values \cite[Eq. (8.6.2)]{golub13}. 

The inequality: $[\matr{U}]_{\set{S},\set{K}}$ has $k$ columns and so $\textrm{column rank}([\matr{U}]_{\set{S},\set{K}})\leq k$ and rank equals column rank.
\end{proof}

We can now prove the overall result:

\begin{lemma}    \label{lemma:LS_delta_1_improvement_means_delta_2_worse}
    For LS, $\Delta_{1} < 0 $ if and only if $ \Delta_{2} > 0$.
\end{lemma}
\begin{proof}
As $\Delta_{1} \in \{0,-1\}$ (Lemma \ref{lemma:LS_xi_1_is_rank}), we instead prove that $\Delta_{1} = 0$ if and only if $ \Delta_{2} \leq 0$.

     Write the eigenvalues of $[\matr{\Pi}_{bl(\set{K})}]_{\set{S} \backslash \{v\}}$ as $\lambda_{1}, \dots, \lambda_{n}$ and the eigenvalues of $[\matr{\Pi}_{bl(\set{K})}]_{\set{S}}$ as $\mu_{1}, \dots \mu_{n+1}$. As $[\matr{\Pi}_{bl(\set{K})}]_{\set{S} \backslash \{v\}}$ is a principal submatrix of $[\matr{\Pi}_{bl(\set{K})}]_{\set{S}}$, by Cauchy's Interlacing Theorem \cite[p.~59]{bhatia2013matrix},
     \begin{equation}
         0 \leq \mu_{1} \leq \lambda_{1} \leq \dots \leq \lambda_{n} \leq \mu_{n+1} \leq 1 \label{eq:MUUM_nonnegative}
     \end{equation} 
    where the outer bounds come from the fact that both matrices are principal submatrices of $\matr{\Pi}_{bl(\set{K})}$, an orthogonal projection matrix.
    
    \subsubsection{\texorpdfstring{$\Delta_{1}= 0 \implies \Delta_{2} \leq 0$}{\textDelta\textoneinferior = 0 implies \textDelta\texttwoinferior =< 0}}
    $\Delta_{1} = 0$ implies  $\textrm{rank}([\matr{U}]_{\set{S},\set{K}}) = \textrm{rank}(\matrsubUGen{\set{S} \backslash\{v\}})$,  so  $\rank{\projbl[S]} = \rank{\projblgen[\set{S} \backslash \{v\}]}$. As the rank is unchanged, $[\matr{\Pi}_{bl(\set{K})}]_{\set{S}}$ has one more zero-eigenvalue than 
 $[\matr{\Pi}_{bl(\set{K})}]_{\set{S} \backslash \{v\}}$. This means:
    \begin{align}
        \mu_{1} = 0 \label{eq:mu_1_is_zero}\\
        \lambda_{i} = 0 \iff \mu_{i+1} = 0 \label{eq:mu_lambda_biject_zeros}
    \end{align} 
    By Cauchy's Interlacing Theorem, $\lambda_{i} \leq \mu_{i+1}$ and so 
    \begin{equation}
        \frac{1}{\lambda_{i}} \geq \frac{1}{\mu_{i+1}} \text{ if } \lambda_{i} \neq 0 \text{ and } \mu_{i+1} \neq 0.  \label{eq:lambda_mu_ineq}
    \end{equation} Therefore
 \begin{equation}
     \sum_{\lambda_{i}^{\set{S}} \neq 0} { \frac{1}{\lambda_{i}^{\set{S}}}} \geq \sum_{\mu_{i}^{\set{S}} \neq 0} { \frac{1}{\mu_{i}^{\set{S}}}} \label{eq:eig_sum_ineq}
 \end{equation}
 as we have the same number of non-zero terms in each of these terms by (\ref{eq:mu_1_is_zero}) and (\ref{eq:mu_lambda_biject_zeros}), and the inequality is proved by summing over the non-zero terms using (\ref{eq:lambda_mu_ineq}).
Equation (\ref{eq:eig_sum_ineq}) is exactly 
\begin{equation}
   \xi_{2}(\set{S} \backslash \{v\}) \geq \xi_{2}(\set{S}).
\end{equation}
Rearranging gives $\Delta_{2} \leq 0$.

\subsubsection{\texorpdfstring{$\Delta_{1} = 0 \impliedby \Delta_{2} \leq 0$}{\textDelta\textoneinferior = 0 <== \textDelta\texttwoinferior =< 0}} We prove the equivalent statement
\begin{equation}
    \Delta_{1} \neq 0 \implies \Delta_{2} > 0 .
\end{equation}
By Lemma \ref{lemma:LS_xi_1_is_rank} , if $\Delta_{1} \neq 0$ then $ \Delta_{1} = -1$. This means that  $\rank{\matrsubU{S}} -1 = \rank{\matrsubUGen{\set{S} \backslash \{v\}}}$, therefore $\projbl[S]$ has one more non-zero eigenvalue than $[\matr{\Pi}_{bl(\set{K})}]_{\set{S} \backslash \{v\}}$. This means:
\begin{align}
    \mu_{n+1} > 0 \label{eq:mu_nonzero} \\
    \lambda_{i} \neq 0 \iff \mu_{i} \neq 0 \label{eq:lambda_mu_match2}
\end{align}
By Cauchy's interlacing theorem, $\lambda_i \geq \mu_i$ and so
\begin{equation}
        \frac{1}{\lambda_{i}} \leq \frac{1}{\mu_{i}} \text{ if } \lambda_{i} \neq 0 \text{ and } \mu_{i} \neq 0.  \label{eq:lambda_mu_inv_ineq2}
    \end{equation}
    Let $I$ be the number of zero eigenvalues of $\projbl[S]$. Then 
\begin{equation}
    \sum_{I\leq i \leq n} { \frac{1}{\lambda_{i}^{\set{S}}}} \leq \sum_{I\leq i \leq n} { \frac{1}{\mu_{i}^{\set{S}}}} <  \sum_{I \leq i \leq n+1} { \frac{1}{\mu_{i}^{\set{S}}}}. \label{eq:eig_sum_ineq2_intermediate}
\end{equation}
With the left inequality by matching terms via (\ref{eq:lambda_mu_match2}) and then summing over (\ref{eq:lambda_mu_inv_ineq2}), and the right inequality because (\ref{eq:mu_nonzero}) means $\frac{1}{\mu_{n+1}} > 0$. We then note the left and the right terms in this inequality say:
\begin{equation}
     \sum_{\lambda_{i}^{\set{S}} \neq 0} { \frac{1}{\lambda_{i}^{\set{S}}}} < \sum_{\mu_{i}^{\set{S}} \neq 0} { \frac{1}{\mu_{i}^{\set{S}}}} \label{eq:eig_sum_ineq2}
 \end{equation}
or equivalently,
\begin{equation}
 \xi_{2}(\set{S} \backslash \{v\}) < \xi_{2}(\set{S}).   
\end{equation}
Rearranging gives $\Delta_{2} > 0$.
\end{proof}

{%
\section{Proof of Theorem \ref{main_general}}
\label{app:main_thm_proof}
\begin{proof}
    
$\set{S} \backslash \set{T}$ is better than $\set{S}$ if and only if $\textrm{MSE}_{\set{S}\backslash\set{T}} < \textrm{MSE}_{
\set{S}}$. By (\ref{eq:EMSE_decomp_into_delta}) this happens if and only if
\begin{equation}
    \Delta_{1}(\set{S},\set{T}) + \sigma^{2} \cdot \Delta_{2}(\set{S},\set{T}) > 0.
\end{equation}
By substituting in $\sigma^{2} = { \frac{\expect{\sqnormvec{\vect{x}}}}{\expect{\sqnormvec{\vect{\epsilon}}}} }  \cdot \frac{1}{\textrm{SNR}}$ and multiplying both sides by SNR (which does not change the direction of the inequality, as $\textrm{SNR} > 0$), $\set{S} \backslash \set{T}$ is better than $\set{S}$ if and only if
\begin{equation}
\label{eq:proof_main_general_intermediate}
   { \frac{\expect{\sqnormvec{\vect{x}}}}{\expect{\sqnormvec{\vect{\epsilon}}}} } \cdot \Delta_{2}(\set{S},\set{T}) > -\Delta_{1}(\set{S},\set{T}) \cdot \textrm{SNR}. 
\end{equation}
We consider the conditions of (\ref{eq:main_thm_cond}):
\subsection*{\texorpdfstring{(\ref{eq:main_thm_cond:d1neg}): $\Delta_{1}(\set{S},\set{T}) < 0$}{(\ref{eq:main_thm_cond:d1neg}): \textDelta\textoneinferior(S,T) < 0}}
We can divide both sides of (\ref{eq:proof_main_general_intermediate}) by $-\Delta_{1}(\set{S},\set{T})$ without changing the inequality, so (\ref{eq:proof_main_general_intermediate}) holds if and only if
\begin{equation}
  { \frac{\expect{\sqnormvec{\vect{x}}}}{\expect{\sqnormvec{\vect{\epsilon}}}} } \cdot \frac{\Delta_{2}(\set{S},\set{T})}{-\Delta_{1}(\set{S},\set{T})} >  \textrm{SNR}.
\end{equation}

\subsection*{\texorpdfstring{(\ref{eq:main_thm_cond:d1pos}): $\Delta_{1}(\set{S},\set{T}) > 0$}{(\ref{eq:main_thm_cond:d1pos}): \textDelta\textoneinferior(S,T) > 0}}
Dividing both sides of (\ref{eq:proof_main_general_intermediate}) by $-\Delta_{1}(\set{S},\set{T})$ flips the inequality, so (\ref{eq:proof_main_general_intermediate}) holds if and only if
\begin{equation}
  { \frac{\expect{\sqnormvec{\vect{x}}}}{\expect{\sqnormvec{\vect{\epsilon}}}} } \cdot\frac{\Delta_{2}(\set{S},\set{T})}{-\Delta_{1}(\set{S},\set{T})} < \textrm{SNR}.
\end{equation}

\subsection*{\texorpdfstring{(\ref{eq:main_thm_cond:d1zero}): $\Delta_{1}(\set{S},\set{T}) = 0$}{(\ref{eq:main_thm_cond:d1zero}): \textDelta\textoneinferior(S,T) = 0}}
$-\Delta_{1}(\set{S},\set{T}) \cdot \textrm{SNR} = 0$ so (\ref{eq:proof_main_general_intermediate}) holds if and only if $ { \frac{\expect{\sqnormvec{\vect{x}}}}{\expect{\sqnormvec{\vect{\epsilon}}}} }\cdot \Delta_{2}(\set{S},\set{T}) > 0$, if and only if
\begin{equation}
\Delta_{2}(\set{S},\set{T}) > 0.
\end{equation}
\end{proof}
}

\section{Proof of Corollary \ref{main_ls}}
\label{proof_appendix_LS}
\begin{proof}
For brevity, we fix $\set{S}$ and $v$ and write $\Delta_{1} = \Delta_{1}(\set{S},v)$ and $\Delta_{2} = \Delta_{2}(\set{S},v) $.

Rearranging (\ref{eq:EMSE_decomp_into_delta}) gives us that $\set{S} \backslash \{v\}$ is better than $\set{S}$ if and only if 
\begin{equation}
    \Delta_{1} + \sigma^{2} \cdot \Delta_{2} > 0
\end{equation}
or equivalently if and only if 
\begin{equation}
    \Delta_{1} > - \sigma^{2} \cdot \Delta_{2} .
\end{equation}
By definition, $\sigma^{2} = \frac{k}{N \cdot \text{SNR}}$, so this condition is equivalent to
\begin{equation}
    \Delta_{1} > - \frac{k}{N \cdot \text{SNR}} \Delta_{2}
\end{equation}
and as SNR is strictly positive, this is equivalent to
\begin{equation}
    \text{SNR}\cdot \Delta_{1} > - \frac{k}{N} \Delta_{2}. \label{eq:thm_proof_linear_cond}
\end{equation}

We can now use the major lemmas from the previous appendices. By Lemma \ref{lemma:LS_xi_1_is_rank}, we have two possible values of $\Delta_{1}(\set{S},v)$:
\subsection*{$\Delta_{1} = 0$:}
Lemma \ref{lemma:LS_delta_1_improvement_means_delta_2_worse} means $\Delta_{2} < 0$, so
\begin{equation}
    \Delta_{1} + \sigma^{2} \cdot \Delta_{2} = \sigma^{2} \cdot \Delta_{2} < 0
\end{equation}
and so $\set{S} \backslash \{v\}$ is not better than $\set{S}$.

\subsection*{$\Delta_{1} = -1$:}
Eq. (\ref{eq:thm_proof_linear_cond}) simplifies to:
\begin{equation}
    - \text{SNR} > - \frac{k}{N} \Delta_{2}
\end{equation}
which is equivalent to
\begin{equation}
    \text{SNR} < \frac{k}{N} \Delta_{2}. \label{eq:delta_is_minus_1_cond_SNR}
\end{equation}

On the one hand, $v$ improves $\set{S}$ implies $\Delta_{1} = -1$, which implies (\ref{eq:delta_is_minus_1_cond_SNR}). On the other hand, (\ref{eq:delta_is_minus_1_cond_SNR}) implies $\Delta_{2} > 0$ which in turn implies $\Delta_{1}=-1$, which means (\ref{eq:delta_is_minus_1_cond_SNR}) implies (\ref{eq:thm_proof_linear_cond}), which implies $\set{S} \backslash \{v\}$ is better than $\set{S}$. 

Note that the right-hand side of (\ref{eq:delta_is_minus_1_cond_SNR}) is $\tau(\set{S},v)$; this completes the proof.
\end{proof}

\section{Proof of Proposition \ref{propn:main_existence_LS}}
\label{app:LS_satisfied}
We reframe the Proposition to the following equivalent statement:

    Consider any sequence of vertices $v_{1},\dots, v_{N}$ with no repeated vertices, and let $\set{S}_{i} = \{v_{1} , \dots , v_{i} \}$. Then there are exactly $k$ indices $I_{1}, \dots, I_{k}$ such that under LS reconstruction of a noisy $k$-bandlimited signal,
    \begin{equation}
        \forall 1\leq j \leq k: \tau(\set{S}_{I_{j}}, v_{I_{j}}) > 0
    \end{equation}
    and so for some $\text{SNR}>0$, $\set{S}_{I_{j}} \backslash \{v_{I_{j}}\}$ is better than $\set{S}_{I_{j}}$.

\begin{proof}
\label{app:proof_main_existence_LS}
By Lemma \ref{lemma:LS_xi_1_is_rank} in Appendix \ref{proof_appendix_LS}: 
\begin{align}
    &\xi_{1}(\set{S}_{i}) = k -\text{rank}([\matr{U}]_{\set{S}_{i}, \set{K}}), \\
    &\Delta_{1} \in \{0,-1\}
\end{align}
 and as $\text{rank}(\matrsubU{N}) = k$, $\xi_{1}(\set{S}_{N}) = 0$. As $\xi_{1}(\set{S}_{0}) = k$, we must have exactly $k$ indices for which $\Delta_{1}(\set{S}_{i}, v_{i}) = -1$, and by Lemma \ref{lemma:LS_delta_1_improvement_means_delta_2_worse} in Appendix \ref{proof_appendix_LS}  
  we have exactly $k$ indices for which $\Delta_{2}(\set{S}_{i}, v_{i}) > 0$. As $\tau (\set{S}_{i}, v_{i})= \frac{k}{N} \Delta_{2}(\set{S}_{i}, v_{i})$, we're done.
\end{proof}

\section{Proofs for LS Reconstruction with Bandlimited Noise}
\label{app:LS_Bandlimited_Noise_proofs}
\subsection{Proof of Lemma \ref{lemma:LS_bandlimited_noise_MSE}}
\begin{proof}
By Appendix \ref{app:table_delta_proof}, Lemma \ref{lemma:LS_xi_1_is_rank}, under LS reconstruction,
    \begin{equation}
    \xi_{1}(\set{S}) = k - \textrm{rank}([\matr{U}]_{\set{S},\set{K}}).
    \end{equation}
        Assuming LS reconstruction,
    \begin{align}
         \xi_{2}(\set{S}) &=  || \matrsubU{N}[\matr{U}]_{\set{S},\set{K}}^{+}[\matr{U}]_{\set{S},\set{K}}||^{2}_{F} \\
         &= || [\matr{U}]_{\set{S},\set{K}}^{+}[\matr{U}]_{\set{S},\set{K}}||^{2}_{F}.
    \end{align}
    As $\matr{A}^{+}\matr{A}$ is an orthogonal projection matrix, its eigenvalues are 0 or 1. Therefore
    \begin{equation}
         || [\matr{U}]_{\set{S},\set{K}}^{+}[\matr{U}]_{\set{S},\set{K}}||^{2}_{F} =  \textrm{rank}([\matr{U}]_{\set{S},\set{K}}).
    \end{equation}
    Add this times $\sigma^{2}$ to $\xi_{1}(\set{S})$ to get the result.
\end{proof}
\subsection{Proof of Corollary \ref{corr:LS_bandlimited_noise_big_variance}}
\begin{proof}
    As sample size increases, $\textrm{rank}([\matr{U}]_{\set{S},\set{K}})$ is increasing. If $\textrm{SNR} < 1$, then $1<\sigma^{2}$ and the MSE increases with sample size by Lemma \ref{lemma:LS_bandlimited_noise_MSE}.
\end{proof}
\subsection{Proof of Corollary \ref{corr:LS_bandlimited_noise_sample_only_k}}
\begin{proof}
    Under a noiseless-optimal sampling scheme, after sampling $k$ vertices we have perfect reconstruction of any clean $k$-bandlimited signal, and so $\xi_{1}(\set{S}_{k}) = k - \textrm{rank}([\matr{U}]_{\set{S}_{k},\set{K}}) = 0$. 
    
    Let $m \leq k$. As $\msubgen[\set{S}_{k},\set{K}]{\matr{U}}$ is of full rank, for any $\set{S}_{m} \subseteq \set{S}_{k}$, $\msubgen[\set{S},\set{K}]{\matr{U}}$ must also be full rank.
    Therefore
    \begin{align}
        \textrm{MSE}_{\set{S}_{m}} - \textrm{MSE}_{\set{S}_{m} \backslash \{v\}} &= (\sigma^{2} - 1)(m - (m-1)) \\
        &= \sigma^{2} - 1
    \end{align}
    so $\set{S}_{m} \backslash \{v\}$ is better than $\set{S} \iff \sigma^{2} > 1 \iff \textrm{SNR} < 1$.
    
    In the case where $m > k$:
    \begin{align}
        |\set{S}| \geq k &\implies \textrm{rank}([\matr{U}]_{\set{S},\set{K}}) = k \\
        &\implies \mathrm{MSE}_{\set{S}} = \sigma^{2}k.
    \end{align}
    which is constant as sample size increases.
\end{proof}

\section{Proof of Theorem \ref{thm:noiseless_optimality_means_noise_sensitivity}}
\label{app:optimal_noiseless_schemes_immediately_satisfy}
\subsection{Proof of (\ref{eq:greedy_sampling_first_k})}
We first show that if $m \leq k$ then 
\begin{equation}
    \forall m \leq k: \enskip \forall v \in \set{S}_{m}: \enskip \Delta_{1}(\set{S}_{m},v) = -1. \label{eq:proof_LS_exist:delta_1_neg}
\end{equation}
\begin{proof}
By Appendix \ref{proof_appendix_LS}, Lemma \ref{lemma:LS_xi_1_is_rank}, the noiseless error
\begin{equation}
    \xi_{1}(\set{S}) = k -\text{rank}([\matr{U}]_{\set{S},\set{K}})
\end{equation}
must be 0, as we can perfectly reconstruct any $k$-bandlimited signal. Therefore, $\text{rank}([\matr{U}]_{\set{S},\set{K}}) = k$.

$[\matr{U}]_{\set{S},\set{K}}$ is a $k \times k$ matrix of full rank, so its rows must be linearly independent. 
Any subset of linearly independent rows is linearly independent, so for any non-empty $\set{R} \subset \set{S}$, $[\matr{U}]_{\set{R},\set{K}}$ has linearly independent rows.

Greedy schemes pick increasing sample sets: that is, if asked to pick a vertex sample set $\set{S}_{m}$ of size $m$ for $m < k$ and a sample set $\set{S}$ of size $k$, $\set{S}_{m} \subset \set{S}$. Therefore for any sample set $\set{S}_{m}$ of size $m \leq k$ picked by the scheme, $[\matr{U}]_{\set{S}_{m},\set{K}}$ has independent rows.

If $[\matr{U}]_{\set{S}_{m},\set{K}}$ has independent rows, then removal of any row (corresponding to removing any vertex) reduces its rank by 1; which is (\ref{eq:proof_LS_exist:delta_1_neg}).
\end{proof}

We now show that for $m \leq k$, 
\begin{equation}
    \forall m \leq k: \enskip \forall v \in \set{S}_{m}: \enskip \Delta_{2}(\set{S}_{m},v) \geq 1.
\end{equation}
\begin{proof}
By the previous section, we know that $\matrsubU{R}$ is full rank for $\set{R}\subseteq \set{S}_{m}$, so $\matrsubU{R}\matrsubU{R}^{T} = \projbl[R]$ is invertible. By Appendix \ref{app:table_delta_proof}, (\ref{eq:xi_2_is_U_pseudo_frob}), $\xi_{2}(\set{S}_{m}) = \trace{\projblgen[\set{S}_{m}]^{-1}}$. We have
\begin{align}
    &\xi_{2}(\set{S}_{m}) \\
    &= \trace{\projblgen[\set{S}_{m}]^{-1}} \\
    &= \msubgen[\{v\}]{\projblgen[\set{S}_{m}]^{-1}} + \trace{\msubgen[\set{S}_{m} \backslash \{v\}]{\projblgen[\set{S}_{m}]^{-1}}} \\
    &\geq \projblgen[\{v\}]^{-1} + \trace{\projblgen[\set{S}_{m} \backslash \{v\}]^{-1}} \label{eq:proof_LS_exist:bring_submatrix_inverse_in} \\
    &= \frac{1}{\projblgen[\{v\}]} + \xi_{2}(\set{S}_{m} \backslash \{v\}) \\
    &\geq 1 + \xi_{2}(\set{S}_{m} \backslash \{v\})
\end{align}
where the inequality in (\ref{eq:proof_LS_exist:bring_submatrix_inverse_in}) is by \cite[Eq. 5]{zhang2000schur}, and the final inequality is because the diagonal elements of $\projbl$ are bounded above by its maximum eigenvalue, which is 1 as $\projbl$ is a projection. Therefore, for all $v \in \set{S}_{m}$,
\begin{align}
    \Delta_{2}(\set{S}_{m},v) = \xi_{2}(\set{S}_{m}) - \xi_{2}(\set{S}_{m} \backslash \{v\}) \geq 1.
\end{align}

\end{proof}
Finally as $\tau(\set{S}_{m},v) = \frac{k}{N}\Delta_{2}(\set{S}_{m},v)$,
\begin{equation}
     \forall m \leq k: \enskip \forall v \in \set{S}_{m}: \enskip \tau(\set{S}_{m},v) \geq \frac{k}{N}. 
\end{equation}
\subsection{Proof of (\ref{eq:greedy_sampling_over_k})}
\begin{proof}
    
As $[\matr{U}]_{\set{S}_{k},\set{K}}$ has $k$ independent rows, it is of rank $k$. Adding further rows cannot decrease its rank, so for $m' > k$, $\textrm{rank}( [\matr{U}]_{\set{S}_{m'},\set{K}}) \geq k$. As $\matrsubU{N}$ is of rank $k$, $\textrm{rank}([\matr{U}]_{\set{S}_{m'},\set{K}}) \leq k$. This means for all samples sizes $m' > k$, $\textrm{rank}([\matr{U}]_{\set{S}_{m'},\set{K}}) = k$. This says that further additions of rows do not change rank; that is:
\begin{equation} 
   \forall m' > k: \enskip \forall v \in \set{S}_{m'} \backslash \set{S}_{k}: \enskip \Delta_{1}(\set{S}_{m'},v) = 0
\end{equation}
Then, by Appendix \ref{proof_appendix_LS}, Lemma \ref{lemma:LS_delta_1_improvement_means_delta_2_worse},
\begin{equation}
    \forall m' > k: \enskip \forall v \in \set{S}_{m'} \backslash \set{S}_{k}: \enskip \Delta_{2}(\set{S}_{m'},v) \leq 0 
\end{equation}
and, like for (\ref{eq:greedy_sampling_first_k}), as $\tau(\set{S}_{m},v) = \frac{k}{N}\Delta_{2}(\set{S}_{m},v)$ and $\frac{k}{N} > 0$,
\begin{equation}
     \forall m' > k: \enskip \forall v \in \set{S}_{m'}\backslash \set{S}_{k} : \enskip \tau(\set{S}_{m'},v) \leq 0. 
\end{equation}
\end{proof}

\section{Proof of Remark \ref{remark:ADE_are_noiseless_optimal}}
\label{app:proof_of_remark_ADE_are_noiseless_optimal}

\subsection*{A-optimality}
A-optimality depends on the existence of the inverse of $\projbl[S]$  existing, which requires it to be of full rank (at $|\set{S}| = k$ our definition of $\matr{P}$ requires $\matrsubU{S}^{T}\matrsubU{S}$ invertible; this implies $\projbl[S]$ invertible as $0 < \text{det}(\matrsubU{S}^{T}\matrsubU{S}) = \text{det}(\matrsubU{S}^{T})\text{det}(\matrsubU{S}) = \text{det}(\matrsubU{S}\matrsubU{S}^{T})$ ). By Appendix \ref{proof_appendix_LS}, Lemma \ref{lemma:square_to_rect_rank}, if an A-optimal scheme picks a set $\set{S}$ of size $k$, then $\text{rank}([\matr{U}]_{\set{S},\set{K}}) = k$. Therefore, $\set{S}$ is a uniqueness set \cite{anis2016efficient} and can perfectly reconstruct any $k$-bandlimited signal.

\subsection*{D- and E-optimality}
We show that for sample sizes less than $k$ we can always pick a row which keeps $\projbl[S]$ full rank (of rank $|\set{S}|$), and that D- and E-optimal schemes do so.

By Appendix \ref{proof_appendix_LS}, Lemma \ref{lemma:square_to_rect_rank}, $\text{rank}(\projbl[S]) = \text{rank}(\matrsubU{S})$, so we only need to ensure $\text{rank}(\matrsubU{S}) = |\set{S}|$.

We proceed by induction: given $\set{S}_{1}$ with $|\set{S}_{1}| = 1$, $\text{rank}(\matrsubUGen{\set{S}_{1}}) = 1$. Assume that for $\set{S}_{i}$ with $|\set{S}_{i}| = i < k$, $\text{rank}(\matrsubUGen{\set{S}_{i}}) = i$. As $\text{rank}(\matrsubU{N}) = k$ and $i < k$, we can find a row to add to $\matrsubUGen{\set{S}_{i}}$ which will increase its rank (else all other rows would lie in the $i$-dimensional space spanned by the rows of $\matrsubUGen{\set{S}_{i}}$, which would imply $\text{rank}(\matrsubU{N}) = i$, which is a contradiction as $i < k$). Adding the vertex which corresponds to the row to $\set{S}_{i}$ gives $\set{S}_{i+1}$ with $\text{rank}(\matrsubUGen{\set{S}_{i+1}}) = i+1$.

We have shown that we can greedily choose to keep $\text{rank}(\matrsubU{S}) = |\set{S}|$. We now show that D- and E-optimal schemes do so. The eigenvalues of $\projbl[S]$ are non-negative (see Appendix \ref{proof_appendix_LS}, Eq. (\ref{eq:MUUM_nonnegative})), so any invertible $\projbl[S]$ will have a strictly positive determinant and minimum eigenvalue, which are preferable under the D- and E- optimality criterion respectively to a non-invertible $\projbl[S]$, which has a determinant and minimum eigenvalue of 0. Therefore, greedy D- and E- optimal sampling schemes will make sure $\projbl[S]$ is invertible, and thus keep $\text{rank}(\matrsubU{S}) = |\set{S}|$ for $|\set{S}| \leq k$. Therefore when D- and E- optimal schemes pick a set $\set{S}$ of size $k$, $\text{rank}(\matrsubU{S}) = k$. Therefore, $\set{S}$ is a uniqueness set \cite{{anis2016efficient}} and can perfectly reconstruct any $k$-bandlimited signal.

\section{Proof of Corollary \ref{corr:main_GLR_iff}}
\label{app:proof_main_GLR_iff}
We first simplify $\xi_{1}(\set{S})$:
\begin{align}
    &\matrsubU{N} - \matr{R}_{\set{S}}\matrsubU{S} \\
    &= \left(\matr{I} - \left(\proj{S} + \mu\matr{L} \right)^{-1}\proj{S}\right)\matrsubU{N}  \\
    &=  \left(\proj{S} + \mu\matr{L} \right)^{-1}\mu\matr{L}\matrsubU{N} \\
    &=  \left(\proj{S} + \mu\matr{L} \right)^{-1}\matrsubU{N} \left(\mu\matr{\Lambda}_{k}\right)
\end{align}
where $\matr{\Lambda}_{k}$ is a $k \times k$ diagonal matrix with the corresponding graph frequencies to $\matrsubU{N}$ as its diagonal. Write $u_{i}$ for the $i^{\textrm{th}}$ column of $\matrsubU{N}$, so $\vect{u}_{i}$ is an eigenvector of $\matr{L}$. 
\begin{align}
    \xi_{1}(\set{S}) &= \sqfrob{\matrsubU{N} - \matr{R}_{\set{S}}\matrsubU{S}}  \\
    &= \sum_{i=2}^{k} \left(\mu\lambda_{i}\right)^{2}\vect{u}_{i}^{T}\left(\proj{S} + \mu\matr{L} \right)^{-2}\vect{u}_{i} \label{eq:GLR_corr_proof:xi_1_def}
\end{align}

Note that the condition is equivalent to the following:
\begin{equation}
    \msubgen[\set{S}^{c},\{2,\ldots,k\}]{\matr{U}} = \matr{0} \iff \forall i \in [2,k]: \proj{S}\vect{u}_{i} = \vect{u}_{i}
\end{equation}
that is, the projection is idempotent on each of the $k-1$ non-constant eigenvectors. We consider the cases where this is and is not true and correlate them to cases in Theorem \ref{main_general}.

\subsection{The projection is idempotent}
\label{subapp:GLR_iff_proof:idempotent}
As $\proj{S}\vect{u}_{i} = \vect{u}_{i}$,
\begin{equation}
   (\proj{S} + \mu \matr{L})\vect{u}_{i} = (1 + \mu\lambda_{i})\vect{u}_{i} 
\end{equation}
therefore $\vect{u}_{i}$ is an eigenvector of $(\proj{S} + \mu \matr{L})$ with eigenvalue $1 + \mu\lambda_{i}$ and
\begin{align}
    \vect{u}_{i}^{T}(\proj{S} + \mu \matr{L})^{-2}\vect{u}_{i} = (1 + \mu\lambda_{i})^{-2}.
\end{align}
By Lemma \ref{lemma:GLR_full_observation_MSE}, in this case $\xi_{1}(\set{S}) = \xi_{1}(\set{N})$, i.e. that $\Delta_{1}(\set{N},\set{S}^{c}) = 0$. This corresponds to condition (\ref{eq:main_thm_cond:d1zero}), and gives us condition (\ref{eq:GLR_corol_iff_weird}) in our Corollary.
\subsection{The projection is not idempotent}
Applying Cauchy-Schwartz to $\vect{x} = (\proj{S} + \mu \matr{L})^{-1}\vect{u}_{i}$ and  $\vect{y} = (\proj{S} + \mu \matr{L})\vect{u}_{i}$ gives, as $\vect{u}_{i}^{T}\vect{u}_{i} = 1$,
\begin{align}
    1 \leq \vect{u}_{i}^{T}(\proj{S} + \mu \matr{L})^{-2}\vect{u}_{i} \vect{u}_{i}^{T}(\proj{S} + \mu \matr{L})^{2}\vect{u}_{i}. \label{eq:cs_for_non_idempotent_GLR_proj}
\end{align}

We note that
\begin{align}
    \sqnormvec{\proj{S}\vect{u}_{i}} = \sum_{j \in \set{S}} \left(\vect{u}_{i}\right)_{j}^{2} \leq \sum_{j =1}^{N} \left(\vect{u}_{i}\right)_{j}^{2} = \sqnormvec{\vect{u}_{i}} = 1. \label{eq:proj_idempotent_ineq_for_special_i}
\end{align}
and see that \eqref{eq:proj_idempotent_ineq_for_special_i} must be strict for at least one $i\in [2,k]$, because some component of $\vect{u}_{i}$ in $\set{S}^{c}$ is nonzero for at least one $i \in [2,k]$. Therefore
\begin{align}
    \vect{u}_{i}^{T}(\proj{S} + \mu \matr{L})^{2}\vect{u}_{i} &= (\mu\lambda_{i})^{2} + (1 + 2\mu\lambda_{i}) \sqnormvec{\proj{S}\vect{u}_{i}} \\
    &\leq (1 + \mu\lambda_{i})^{2}.
\end{align}
and substituting into \eqref{eq:cs_for_non_idempotent_GLR_proj}:
\begin{align}
    \left(\frac{1}{1 + \mu\lambda_{i}}\right)^{2} \leq \vect{u}_{i}^{T}(\proj{S} + \mu \matr{L})^{-2}\vect{u}_{i}
\end{align}
with these inequality being strict for at least one $i \in [2,k]$. Therefore, by Lemma \ref{lemma:GLR_full_observation_MSE} and (\ref{eq:GLR_corr_proof:xi_1_def}),
\begin{equation}
    \xi_{1}(\set{S}) > \xi_{1}(\set{N})
\end{equation}
so $\Delta_{1}(\set{N},\set{S}^{c}) < 0$.

\subsection{Simplifying Theorem \ref{main_general}}
We see that the projection is idempotent on $(\vect{u}_{i})_{i=2}^{k}$ implies $\Delta_{1}(\set{N},\set{S}^{c}) = 0$, and the projection is not idempotent implies $\Delta_{1}(\set{N},\set{S}^{c}) < 0$. As the projection must either be idempotent or not idempotent, these implications must be `if and only if' statements. We therefore rule out (\ref{eq:main_thm_cond:d1pos}) in Theorem \ref{main_general}. 

\section{Proof of Theorem \ref{thm:main_GLR_exist}}
\label{app:Proof_thm_main_GLR_exist}
\noindent See Appendix \ref{app:proof_of_remark_GLR_mopt} for a proof that $ m_{opt} < \frac{N+1}{2}$ if $B(m)<N$.

We restate the following definitions from Lemmas \ref{lemma:GLR_xi_2_bound_main} and \ref{lemma:GLR_xi_2_bound_main_bl}:
\begin{align}
        r &= \omega\left( \frac{\lambda_{N}}{\lambda_{2}}\right) \\
        B(m) &= r \frac{N}{m} + \sum^{m}_{i=1}\omega\left(\max\left[1,\frac{\lambda_{N+2-i}}{\lambda_{i}}\right]\right) \\
        r_{bl} &= \omega\left( \frac{\lambda_{k}}{\lambda_{2}}\right) \\
        B_{k}(m) &= r_{bl} \frac{N}{m} + \sum^{m}_{i=1}\omega\left(\max\left[1,\frac{\lambda_{k+2-i}}{\lambda_{i}}\right]\right)
\end{align}

We define the following:
\begin{align}
    \bar{\lambda} &= \frac{1}{N}\sum_{i=1}^{N}\lambda_{i}\\
    \mu_{ub} &= \bar{\lambda}^{-1}\left(\sqrt{\frac{N}{B(m_{opt})}} -1\right)  \\
    \tau_{GLR}(\mu) &= \frac{k}{N} \cdot \frac{\left(\sum^{N}_{i=1}\left(1 + \mu\lambda_{i} \right)^{-2}\right) - B(m_{opt})}{k+B_{k}(m_{opt}) - \sum_{i=1}^{k} \left(1 - (1+\mu\lambda_{i})^{-1}\right)^{2}} 
\end{align}

We start by calculating $\xi_{i}(\set{N})$.

\begin{lemma}
\label{lemma:GLR_full_observation_MSE}
  Under GLR reconstruction with parameter $\mu$,
    \begin{align}
        \xi_{1}(\set{N}) &= \sum_{i=1}^{k} \left(1 - \frac{1}{1+\mu\lambda_{i}}\right)^{2} \\
         \xi_{2}(\set{N}) &= \sum_{i=1}^{N} \left( \frac{1}{1+\mu\lambda_{i}}\right)^{2}
    \end{align}
\end{lemma}
\begin{proof}
    Set $\matr{R}_{\set{N}}= (\matr{I} + \mu \matr{L})^{-1}$ in (\ref{eq:xi_1_def}) and (\ref{eq:xi_2_def}), noting $\matrsubU{N}$ are eigenvectors for $\matr{R}_{\set{N}}$.
\end{proof}

Let $\set{S}$ be any sample set of size $m$. We show that under our conditions $\Delta_{2}(\set{N},\set{S}^{c}) > 0$ and then apply Corollary \ref{corr:main_GLR_iff}. 
To do so, we use the following bounds:
\begin{align}
    \xi_{2}(\set{S}) &\leq B(m) \\
    \xi_{1}(\set{S}) &\leq k + B_{k}(m) \\
    \xi_{1}(\set{N}) &= \sum_{i=1}^{k} \left(1 - \frac{1}{1+\mu\lambda_{i}}\right)^{2} \\
    \xi_{2}(\set{N}) & = \sum_{i=1}^{N} \left( \frac{1}{1+\mu\lambda_{i}}\right)^{2}
\end{align}
These are proven in Lemma \ref{lemma:GLR_xi_2_bound_main}, Appendix \ref{app:proof_unif_ub_xi_1_MSE} Lemma \ref{lemma:unif_ub_xi_1_GLR}, and Lemma \ref{lemma:GLR_full_observation_MSE}.
We therefore see that, as $\Delta_{i}(\set{N},\set{S}^{c}) = \xi_{i}(\set{N}) - \xi_{i}(\set{S})$,
\begin{align}
    \Delta_{2}(\set{N},\set{S}^{c}) &\geq \sum_{i=1}^{N} \left( \frac{1}{1+\mu\lambda_{i}}\right)^{2} - B(m) \label{eq:proof_main_GLR_exist:delta_2_lb} \\
    \Delta_{1}(\set{N},\set{S}^{c}) &\geq \sum_{i=1}^{k} \left(1 - \frac{1}{1+\mu\lambda_{i}}\right)^{2} - (k + B_{k}(m)) \label{eq:proof_main_GLR_exist:delta_1_lb}
\end{align}

We now consider a sets $\set{S}$ of size $m_{opt}$ and show that $\Delta_{2}(\set{N},\set{S}^{c}) > 0$. We have that $r > 0$ so $B(m) > 0$ and by assumption $B(m_{opt}) < N$. Therefore $\mu_{ub} > 0$ and is real and it therefore possible to pick $0 < \mu < \mu_{ub}$.
By assumption $0<\mu<\mu_{ub}$, so by Jensen's Inequality,
\begin{align}
    \sum_{i=1}^{N}\left(\frac{1}{1 + \mu\lambda_{i}}\right)^{2} &\geq \frac{N}{\left(1+\mu\frac{\trace{ \matr{L}
 }}{N}  \right)^{2}}\\
    &> \frac{N}{\left(1+\mu_{ub}\frac{\trace{ \matr{L}
 }}{N} 
  \right)^{2}}\\ 
    &= B(m_{opt})
\end{align}
And therefore by (\ref{eq:proof_main_GLR_exist:delta_2_lb}), $\Delta_{2}(\set{N},\set{S}^{c}) > 0$. 

We now apply Corollary \ref{corr:main_GLR_iff}. We case-split on whether $ \projgen{\set{S}^{c}}\msubgen[\set{N},\{2,\ldots,k\}]{\matr{U}}$ is or is not $\vect{0}$, and show $\set{S}$ is better than $\set{N}$ in both cases.
\subsubsection{(\ref{eq:GLR_corol_iff}) - is not 0}
We assume $ \projgen{\set{S}^{c}}\msubgen[\set{N},\{2,\ldots,k\}]{\matr{U}} \neq \matr{0}$. By (\ref{eq:proof_main_GLR_exist:delta_2_lb}) and (\ref{eq:proof_main_GLR_exist:delta_1_lb}),
\begin{align}
    \tau(\set{N},\set{S}^{c}) = \frac{k}{N} \cdot \frac{\Delta_{2}(\set{N},\set{S}^{c})}{-\Delta_{1}(\set{N},\set{S}^{c})} \geq \tau_{GLR}(\mu) > \text{SNR}
\end{align}
where the last inequality is by our assumption. Therefore Corollary \ref{corr:main_GLR_iff}  (\ref{eq:GLR_corol_iff}) holds and  $\set{S}$ is better than $\set{N}$.
\subsubsection{(\ref{eq:GLR_corol_iff_weird}) - is 0}
We assume $ \projgen{\set{S}^{c}}\msubgen[\set{N},\{2,\ldots,k\}]{\matr{U}} = \matr{0}$. We have that $\Delta_{2}(\set{N},\set{S}^{c}) > 0$, so Corollary \ref{corr:main_GLR_iff} (\ref{eq:GLR_corol_iff_weird}) holds and $\set{S}$ is better than $\set{N}$.

Therefore $\set{S}$ is better than $\set{N}$ regardless of whether $\projgen{\set{S}^{c}}\msubgen[\set{N},\{2,\ldots,k\}]{\matr{U}}$ is or is not $\matr{0}$ and we are done.

\section{Bounding $\xi_{2}(\set{S})$ under GLR -- Proof of Lemma \ref{lemma:GLR_xi_2_bound_main}}
\label{app:proof_unif_ub_xi_2}
In this Appendix, we aim to prove Lemma \ref{lemma:GLR_xi_2_bound_main} which states that under GLR
\begin{align}
    \xi_{2}(\set{S}) &\leq r\frac{N}{m} + \sum_{i=2}^{m} \omega\left(\max\left[1,\frac{\lambda_{N+2-i}}{\lambda_{i}}\right]\right)  \label{eq:xi_2_GLR_bound_strong_appendix}
        \\ 
        &\leq r\frac{N}{m} + r(m -1) . \label{eq:xi_2_GLR_bound_weak_appendix}
\end{align}
\subsection{Preliminaries and notation}
We assume that $\matr{L}$ is the combinatorial Laplacian. We write the eigenvalues of $\matr{L}$ as $0 = \lambda_{1} \leq \ldots \leq \lambda_{N}$.
We note, for any $\matr{X}\in \mathbb{R}^{x \times N}, \matr{Y}\in \mathbb{R}^{N \times y }$,
\begin{align}
   [\matr{X}]_{a,N}[\matr{Y}]_{N,c} &= [\matr{X}\matr{Y}]_{ac}.
\end{align}
We pick a basis suited to our proof. Let the standard basis for $\mathbb{R}^{N}$ be $(\vect{e}_{i})_{i=1}^{N}$. Let
\begin{equation}
    \vect{v}_{1} = \frac{1}{\sqrt{m}}\sum_{i \in \set{S}} \vect{e}_{i} = \frac{1}{\sqrt{m}}\proj{S}\vect{1}_{N}.
\end{equation}
so $||v_{1}||_{2} = 1$. Pick $\{\vect{v}_{2} , \ldots, \vect{v}_{m}\}$ so that $(\vect{v}_{i})_{i=1}^{m}$ is an orthonormal basis for $(\vect{e}_{i})_{i \in \set{S}}$. Finally let $(\vect{v}_{i})_{i=m+1}^{N} = (\vect{e}_{i})_{i \in \set{S}^{c}}$. Now $(\vect{v}_{i})_{i=1}^{N}$ is a basis for $\mathbb{R}^{N}$.

For the rest of this Appendix, we will write out matrices in this new basis. In our new basis, the top left entry when writing out $\matr{L}^{\dagger}$ is 
\begin{equation}
    \matr{L}^{\dagger}_{1,1} = \frac{1}{m}\vect{1}_{m}^{T}[\matr{L}^{+}]_{\set{S}}\vect{1}_{m} \in \mathbb{R}
\end{equation}
and our frequently used projection $\proj{S}$ looks like:
\begin{align}
    \proj{S} = \begin{pmatrix}
        \matr{I} & \matr{0} \\
        \matr{0} & \matr{0}
    \end{pmatrix}.
\end{align}

We define the set $\set{\Theta} = \{2, \ldots, m\}$, and note that
\begin{equation}
     \proj{\Theta} = \matr{I}_{m} - \frac{1}{m}\matr{1}_{m \times m} .
\end{equation}
We have $\{1\} \cup \Theta = \set{S}$ and $ \{1\} \cup \Theta \cup \set{S}^{c} = \set{N}$. 

Finally, we define a useful matrix:
\begin{align}
    \matr{P} &= \matr{I} - \frac{1}{m}\matrsub[N,S]{I}\matr{1}_{m \times N}.
\end{align}
\subsection{Proof overview}
We decompose $\xi_{2}(\set{S}) = \sqfrob{\matr{R}_{\set{S}}}$ row-wise in our new basis (Subsection \ref{subapp:GLR_Proof:row_decomp}).
\begin{align}
 \sqfrob{\matr{R}_{\set{S}}} &= ||\msubgen[\set{N}, \{1\}]{\matr{R}_{\set{S}}}||^{2}_{2} &&+ \sqfrob{\msub[N, \Theta]{\matr{R}_{\set{S}}}} \\
    &= \frac{N}{m} &&+ \sqfrob{\msub[N, \Theta]{\matr{R}_{\set{S}}}}
\end{align}
We explicitly write out $\msub[N,\Theta]{\matr{R}_{\set{S}}}$ (Subsection \ref{subapp:GLR_Proof:Explicit_submatrix}),
\begin{align}
    \msub[N,\Theta]{\matr{R}_{\set{S}}} &= \matr{P}^{T}\msub[N,\Theta]{\frac{1}{\mu}\matr{L}^{\dagger}}\msub[\Theta]{\matr{I} + \frac{1}{\mu}\matr{L}^{\dagger}}^{-1} \\
    &= \matr{P}^{T}\msub[N,\Theta]{\matr{L}^{\dagger}}\msub[\Theta]{\mu \matr{I} + \matr{L}^{\dagger}}^{-1}
\end{align}
and use this to remove the dependence on $\mu$ in our bound (Subsection \ref{subapp:GLR_Proof:no_more_mu}):
\begin{align}
     \sqfrob{\msub[N,\Theta]{\matr{R}_{\set{S}}}} \leq \sqfrob{\matr{P}^{T}\msub[N,\Theta]{\matr{L}^{\dagger}}\msub[\Theta]{\matr{L}^{\dagger}}^{-1}} \label{eq:GLR_Proof_intro:muless_ineq}
\end{align}

We exactly calculate the effects of $\matr{P}^{T}$ on the Frobenius norm (which yields the $\frac{N}{m}$ term in (\ref{eq:GLR_Proof_intro:small1}))(Subsection \ref{subapp:GLR_Proof:column_decomp}):
\begin{align}
    &\sqfrob{\matr{P}^{T}\msub[N,\Theta]{\matr{L}^{\dagger}}\msub[\Theta]{\matr{L}^{\dagger}}^{-1}} \\
    &= \left(\frac{N}{m} \right) \sqfrob{\msubgen[\{1\},\Theta]{\matr{L}^{\dagger}}\msub[\Theta]{\matr{L}^{\dagger}}^{-1}} \label{eq:GLR_Proof_intro:small1}\\
    &+ \sqfrob{\msubgen[\set{N},\Theta]{\matr{L}^{\dagger}}\msub[\Theta]{\matr{L}^{\dagger}}^{-1}} \label{eq:GLR_Proof_intro:small2}
\end{align}

As 
\begin{equation}
\sqfrob{\matrsubpow[\Theta]{L}{\dagger}\msub[\Theta]{\matr{L}^{\dagger}}^{-1}} = \sqfrob{\matr{I}_{m-1}} = m-1,
\end{equation} we get that

\begin{equation}
    \xi_{2}(\set{S}) = \left(\frac{N}{m} + m - 1\right) + \textrm{error}
\end{equation}
where the error term is quantified in (\ref{eq:GLR_Proof_intro:small1}) and (\ref{eq:GLR_Proof_intro:small2}).

Finally, we bound (\ref{eq:GLR_Proof_intro:small1}) and (\ref{eq:GLR_Proof_intro:small2}) using variants of the Kantorovich Inequality. We have that \cite[Eq. 20-23]{householder1965kantorovich} gives (Subsection \ref{subapp:GLR_Proof:adjustment_term})
\begin{equation}
    \sqfrob{\msubgen[\{1\},\Theta]{\matr{L}^{\dagger}}\msub[\Theta]{\matr{L}^{\dagger}}^{-1}} \leq (r-1)
\end{equation}
We use another variant to show that (Subsection  \ref{subapp:GLR_Proof:off_diag_err})
\begin{equation}
    \sqfrob{\msubgen[\set{N},\Theta]{\matr{L}^{\dagger}}\msub[\Theta]{\matr{L}^{\dagger}}^{-1}} \leq \sum^{m}_{i=2}\omega \left( \max\left[1,\frac{\lambda_{N+2-i}}{\lambda_{i}}\right] \right)
\end{equation}

Combine these bounds with (\ref{eq:GLR_Proof_intro:muless_ineq} - \ref{eq:GLR_Proof_intro:small2}) to prove the first inequality in the Lemma.

The second, weaker bound follows as $\omega$ is increasing so 
\begin{equation}
    \omega \left( \max\left[1,\frac{\lambda_{N+2-i}}{\lambda_{i}}\right] \right) \leq \omega\left(\frac{\lambda_{N}}{\lambda_{2}}\right) =r.
\end{equation}

\subsection{Row decomposition}
\label{subapp:GLR_Proof:row_decomp}
As $(\vect{v}_{i})^{m}_{i=1}$ are orthogonal and span $(\vect{e}_{i})_{i \in \set{S}}$, and $\matr{R}_{\set{S}}\vect{1}_{m} = \vect{1}_{N}$,
\begin{align}
 \sqfrob{\matr{R}_{\set{S}}} &= \sum_{i=1}^{m} \left|\left| \matr{R}_{\set{S}}\vect{v}_{i}\right|\right|^{2}_{2} \\
 &=  \sum_{i=1}^{m} \left|\left| \left[\matr{R}_{\set{S}}\right]_{\set{N},\{i\}}\right|\right|^{2}_{2} \\
 &=\left|\left|\matr{R}_{\set{S}}\frac{1}{\sqrt{m}}\vect{1}_{m}\right|\right|^{2}_{2} &&+ \sqfrob{\msub[N, \Theta]{\matr{R}_{\set{S}}}} \\ 
    &= \frac{N}{m} &&+ \sqfrob{\msub[N, \Theta]{\matr{R}_{\set{S}}}}.
\end{align}

\subsection{Explicit submatrix computation}
\label{subapp:GLR_Proof:Explicit_submatrix}
We explicitly compute that 

\begin{align}
    \msub[N,\Theta]{\matr{R}_{\set{S}}} &= \left(\proj{S} + \mu\matr{L}\right)^{-1}\matrsub[N,\Theta]{I}  \\
    &= \matr{P}^{T}\msub[N,\Theta]{\frac{1}{\mu}\matr{L}^{\dagger}}\msub[\Theta]{\matr{I} + \frac{1}{\mu}\matr{L}^{\dagger}}^{-1} \label{eq:GLR_Proof:explicit_submatrix}
\end{align}
\begin{proof}
    We show the equivalent statement,
    \begin{align}
        \matrsub[N,\Theta]{I} \msub[\Theta]{\matr{I} + \frac{1}{\mu}\matr{L}^{\dagger}} &= \left( \proj{S} + \mu\matr{L}\right)\matr{P}^{T}\msub[N,\Theta]{\frac{1}{\mu}\matr{L}^{\dagger}}.
    \end{align}
We have that
\begin{align}
    \mu\matr{L}\matr{P}^{T} &= \mu\matr{L} \\
    \proj{S}\matr{P}^{T} &= \matrsub[N,S]{I} \left( \matr{I} - \frac{1}{m}\matr{1}_{m \times m}\right)\matrsub[S,N]{I} \\
    &= \matrsub[N,S]{I} \msub[S]{\proj{\Theta}}\matrsub[S,N]{I} \\
    &= \matrsub[N,\Theta]{I} \matrsub[\Theta,N]{I}
\end{align}
so 
\begin{align}
\mu\matr{L}\matr{P}^{T} \msub[N,\Theta]{\frac{1}{\mu}\matr{L}^{\dagger}} &= \mu\matr{L} \msub[N,\Theta]{\frac{1}{\mu}\matr{L}^{\dagger}} \\
&= \matr{L}\matr{L}^{\dagger}\matrsub[N,\Theta]{I} \\
&= \left( \matr{I} - \frac{1}{N}\matr{1}_{N \times N}\right)\matrsub[N,\Theta]{I} \\
&= \matrsub[N,\Theta]{I} 
\end{align}
and
\begin{align}
    \proj{S} \matr{P}^{T}\msub[N,\Theta]{\frac{1}{\mu}\matr{L}^{\dagger}} &= \matrsub[N,\Theta]{I} \matrsub[\Theta,N]{I}\msub[N,\Theta]{\frac{1}{\mu}\matr{L}^{\dagger}} \\
    &= \matrsub[N,\Theta]{I} \msub[\Theta]{\frac{1}{\mu}\matr{L}^{\dagger}}
\end{align}
Sum these two terms for the result.
\end{proof}
\subsection{Removing the dependency on $\mu$}
\label{subapp:GLR_Proof:no_more_mu}
By multiplying out the constant in \ref{eq:GLR_Proof:explicit_submatrix}, we see that
\begin{align}
    \msub[N,\Theta]{\matr{R}_{\set{S}}} = \matr{P}^{T}\msub[N,\Theta]{\matr{L}^{\dagger}}\msub[\Theta]{\mu\matr{I} + \matr{L}^{\dagger}}^{-1} \label{eq:GLR_thm_proof:RSNTheta}
\end{align}
We prove that
\begin{align}
    \forall \mu > 0: \sqfrob{\msub[N,\Theta]{\matr{R}_{\set{S}}}} \leq \sqfrob{\matr{P}^{T}\msub[N,\Theta]{\matr{L}^{\dagger}}\msub[\Theta]{\matr{L}^{\dagger}}^{-1}} \label{eq:GLR_Proof:no_more_mu}
\end{align}
\begin{proof}
In this proof, all matrix inequalities are in the standard Loewner ordering \cite[pg. 112]{bhatia2013matrix}. We first show $\msub[\Theta]{\matr{L}^{\dagger}}^{-1}$ is well defined. Principal submatrices of positive definite matrices are positive definite \cite[Corollary III.1.5]{bhatia2013matrix}, so
    \begin{equation}
         \matr{L}^{\dagger} + \vect{1}_{N \times N} >0 \implies \msub[\Theta]{\matr{L}^{\dagger} + \matr{1}_{N \times N}} =\msub[\Theta]{\matr{L}^{\dagger}} >0.
    \end{equation}
    Therefore $\msub[\Theta]{\matr{L}^{\dagger}}$ is invertible and   $\msub[\Theta]{\matr{L}^{\dagger}}^{-1} >0$. Next, as $\mu >0$, $\mu^{2}\matrsub[\Theta]{I} > 0$. Therefore, by expanding the quadratic,
    \begin{equation}    
         \msub[\Theta]{\mu\matr{I} + \matr{L}^{\dagger}}^{2} > \msub[\Theta]{\matr{L}^{\dagger}}^{2}.
    \end{equation}
    By \cite[Proposition V.1.6, page 114]{bhatia2013matrix},
    \begin{equation}
        \msub[\Theta]{\mu\matr{I} + \matr{L}^{\dagger}}^{-2} \leq \msub[\Theta]{\matr{L}^{\dagger}}^{-2}.
    \end{equation}
    As trace is the sum of inner products, by the definition of the Loewner ordering, $\forall \matr{H}\in \mathbb{R}^{N \times m-1}$, $\forall \mu>0$,
    \begin{equation}
        \trace{\matr{H}\msub[\Theta]{\mu\matr{I} + \matr{L}^{\dagger}}^{-2} \matr{H}^{T}} \leq \trace{\matr{H}\msub[\Theta]{\matr{L}^{\dagger}}^{-2}\matr{H}^{T}}
    \end{equation}
    which is the same as
    \begin{equation}
     \sqnormvec{\matr{H}\msub[\Theta]{\mu\matr{I} + \matr{L}^{\dagger}}^{-1} } \leq \sqnormvec{\matr{H}\msub[\Theta]{\matr{L}^{\dagger}}^{-1}}.
    \end{equation}
    Substitute $\matr{H} = \matr{P}^{T}\msub[\set{N},\Theta]{\matr{L}^{\dagger}}$ to finish.
    
\end{proof}
\subsection{Column-wise decomposition}
\label{subapp:GLR_Proof:column_decomp}
We first show that 
\begin{align}
    \sqfrob{\matr{P}^{T}\msub[N,\Theta]{\matr{L}^{\dagger}}\msub[\Theta]{\matr{L}^{\dagger}}^{-1}} &= \frac{N}{m}\sqfrob{\msubgen[\{1\},\Theta]{\matr{L}^{\dagger}}\msub[\Theta]{\matr{L}^{\dagger}}^{-1}}  \\
    &+ \sqfrob{\msub[N,\Theta]{\matr{L}^{\dagger}}\msub[\Theta]{\matr{L}^{\dagger}}^{-1}} \label{eq:GLR_Proof:extract_P_frob}
\end{align}
\begin{proof} 

Let $\matr{K} = \matrsub[N,\Theta]{I}\msub[\Theta]{\matr{L}^{\dagger}}^{-1}$, then 
\begin{align} \sqfrob{\matr{P}^{T}\msub[N,\Theta]{\matr{L}^{\dagger}}\msub[\Theta]{\matr{L}^{\dagger}}^{-1}} &= \textrm{tr}\left(\matr{K}^{T}\matr{L}^{\dagger} \matr{P}\matr{P}^{T} \matr{L}^{\dagger}\matr{K}
 \right)
\end{align}

    As $\matr{1}_{m \times N}\matr{L}^{\dagger} = 0$ cross terms in $\matr{P}\matr{P}^{T}$ disappear in $\matr{L}^{\dagger}\matr{P}\matr{P}^{T}\matr{L}^{\dagger}$, 
    \begin{align}
    \matr{L}^{\dagger}\matr{P}\matr{P}^{T}\matr{L}^{\dagger}  &= \matr{L}^{\dagger}\left(\matr{I} + \frac{N}{m} \begin{pmatrix}
        \frac{1}{m}\matr{1}_{m \times m} & \matr{0}\\
        \matr{0} & \matr{0}
    \end{pmatrix} \right)\matr{L}^{\dagger} \\
    &= \left(\matr{L}^{\dagger}\right)^{2} + \frac{N}{m} \matr{L}^{\dagger}\vect{v}_{1}\vect{v}_{1}^{T}\matr{L}^{\dagger}.
    \label{eq:P_rearrangement_with_extra}
    \end{align}
Therefore
\begin{align}
    &\sqfrob{\matr{P}^{T}\msub[N,\Theta]{\matr{L}^{\dagger}}\msub[\Theta]{\matr{L}^{\dagger}}^{-1}} \\
    &= \trace{\matr{K}^{T}\matr{L}^{\dagger}\matr{L}^{\dagger}\matr{K}} + \frac{N}{m}\trace{\matr{K}^{T}\matr{L}^{\dagger}\vect{v}_{1}\vect{v}_{1}^{T}\matr{L}^{\dagger}\matr{K}} \\
    &=  \sqfrob{\msub[N,\Theta]{\matr{L}^{\dagger}}\msub[\Theta]{\matr{L}^{\dagger}}^{-1}} + \frac{N}{m}\sqfrob{\msubgen[\{1\},\Theta]{\matr{L}^{\dagger}}\msub[\Theta]{\matr{L}^{\dagger}}^{-1}}.
\end{align}
\end{proof}

\subsection{Bounding (\ref{eq:GLR_Proof_intro:small1})}
\label{subapp:GLR_Proof:adjustment_term}
We follow the idea in \cite[Eq. 20-23]{householder1965kantorovich}, but use newer work to account for $\matr{L}^{\dagger}$ not being positive definite.
\begin{align}
    &1 + \sqfrob{\msubgen[\{1\},\Theta]{\matr{L}^{\dagger}}\msub[\Theta]{\matr{L}^{\dagger}}^{-1}} \label{eq:proof_GLR_fb_kantorovich_rminus1_first}\\
    &= 1 + \lambda_{max}\left( \msub[\Theta]{\matr{L}^{\dagger}}^{-1}\msub[\Theta]{\matr{L}^{\dagger}\vect{v}_{1}\vect{v}_{1}^{T}\matr{L}^{\dagger}}\msub[\Theta]{\matr{L}^{\dagger}}^{-1}\right) \\
    &= \lambda_{max}\left(\matrsub[\Theta]{I} + \msub[\Theta]{\matr{L}^{\dagger}}^{-1}\msub[\Theta]{\matr{L}^{\dagger}\vect{v}_{1}\vect{v}_{1}^{T}\matr{L}^{\dagger}}\msub[\Theta]{\matr{L}^{\dagger}}^{-1}\right)\\
    &= \lambda_{max}(\msub[\Theta]{\matr{L}^{\dagger}}^{-1}\msub[\Theta]{\matr{L}^{\dagger}}\msub[\Theta]{\matr{L}^{\dagger}}\msub[\Theta]{\matr{L}^{\dagger}}^{-1}\\
    &\qquad\qquad+ \msub[\Theta]{\matr{L}^{\dagger}}^{-1}\msub[\Theta]{\matr{L}^{\dagger}\vect{v}_{1}\vect{v}_{1}^{T}\matr{L}^{\dagger}}\msub[\Theta]{\matr{L}^{\dagger}}^{-1}) \\
    &\leq \lambda_{max}(\msub[\Theta]{\matr{L}^{\dagger}}^{-1}\msub[\Theta]{\matr{L}^{\dagger}\proj{\Theta}\matr{L}^{\dagger}}\msub[\Theta]{\matr{L}^{\dagger}}^{-1} \label{eq:proof_GLR_fb_kantorovich_proj_expand_rminus1}\\
    &\qquad\qquad+ \msub[\Theta]{\matr{L}^{\dagger}}^{-1}\msub[\Theta]{\matr{L}^{\dagger}\vect{v}_{1}\vect{v}_{1}^{T}\matr{L}^{\dagger}}\msub[\Theta]{\matr{L}^{\dagger}}^{-1}) \\
    &\qquad\qquad+ \msub[\Theta]{\matr{L}^{\dagger}}^{-1}\msub[\Theta]{\matr{L}^{\dagger}\projgen{\set{S}^{c}}\matr{L}^{\dagger}}\msub[\Theta]{\matr{L}^{\dagger}}^{-1}) \label{eq:proof_GLR_fb_kantorovich_add_posdef_rminus1} \\
    &= \lambda_{max}\left(\msub[\Theta]{\matr{L}^{\dagger}}^{-1}\msub[\Theta]{\left(\matr{L}^{\dagger}\right)^{2}}\msub[\Theta]{\matr{L}^{\dagger}}^{-1})\right) \label{eq:proof_GLR_fb_kantorovich_rminus1_last}
\end{align}
where in (\ref{eq:proof_GLR_fb_kantorovich_proj_expand_rminus1}) we use $ \msub[\Theta]{\matr{L}^{\dagger}}\msub[\Theta]{\matr{L}^{\dagger}}= \msub[\Theta]{\matr{L}^{\dagger}\proj{\Theta}\matr{L}^{\dagger}}$. We weaken the inequality in step (\ref{eq:proof_GLR_fb_kantorovich_add_posdef_rminus1}) 
 by adding a positive semi-definite matrix  ($\projgen{\set{S}^{c}}$ is positive semidefinite so for any $\matr{X}$, $\matr{X}\projgen{\set{S}^{c}}\matr{X}^{T}$ is positive semidefinite) -- this is valid by Weyl's inequality \cite[Corollary III.2.2, page 63]{bhatia2013matrix}. The final step comes from $\{1\} \cup \Theta \cup \set{S}^{c} = \set{N}$ so $\matr{X}\vect{v}_{1}\vect{v}_{1}^{T}\matr{X}^{T} + \matr{X}\proj{\Theta}\matr{X}^{T} + \matr{X}\projgen{\set{S}^{c}}\matr{X}^{T} = \matr{X}\matr{X}^{T}$ for any matrix $\matr{X}$.

 We now apply a matrix variant of Kantorovich which holds for positive semi-definite matrices -- \cite[Theorem 3]{liu1997kantorovich} with $\matr{X} = \matrsub[N,\Theta]{I}$,  $\matr{H} = \matr{L}\matr{L}^{\dagger} =\matr{I}_{N} - \frac{1}{N}\vect{1}_{N \times N}$ and $\matr{G} = \matr{L}^{\dagger}$. 
 
 This version has extra constraints on the ranges of $\matr{X}, \matr{G}$ and $\matr{H}$ to handle the semi-definiteness, which we now check. First, $\matr{G}\matr{H} = \matr{H}\matr{G}  = \matr{L}^{\dagger} \geq 0$.   Now $\text{range}(\matr{G}) = \text{range}(\matr{L}^{\dagger}) = \text{range}(\matr{L}\matr{L}^{\dagger}) = \text{range}(\matr{I}-\frac{1}{N}\vect{1}_{N \times N}) =\text{range}(\matr{H})$. As $\vectsub[\Theta]{1_{N}} = \matr{X}^{T}\vect{1}_{N} = \vect{0}$,  so $\text{range}(\vect{1}_{N}\vect{1}_{N}^{T}) \subseteq \text{ker}(\matr{X}^{T})$, and by the Fundamental Theorem of Linear Algebra $\text{ker}(
 X^{T})^{\perp} = \text{range}(\matr{X})$, so we see that $\text{range}(\matr{X})\subseteq \text{range}(\matr{G})$ and $\text{range}(\matr{X})\subseteq \text{range}(\matr{H})$. 

 Applying \cite[Theorem 3]{liu1997kantorovich} gives
 \begin{equation}
     \msub[\Theta]{\matr{L}^{\dagger}}^{-1}\msub[\Theta]{\left(\matr{L}^{\dagger}\right)^{2}}\msub[\Theta]{\matr{L}^{\dagger}}^{-1} \leq r \proj{\Theta}
 \end{equation}
 and therefore
 \begin{equation}
     \lambda_{max}\left( \msub[\Theta]{\matr{L}^{\dagger}}^{-1}\msub[\Theta]{\left(\matr{L}^{\dagger}\right)^{2}}\msub[\Theta]{\matr{L}^{\dagger}}^{-1}\right) \leq r 
 \end{equation}
By (\ref{eq:proof_GLR_fb_kantorovich_rminus1_first}-\ref{eq:proof_GLR_fb_kantorovich_rminus1_last}) we have

\begin{equation}
    \sqfrob{\msubgen[\{1\},\Theta]{\matr{L}^{\dagger}}\msub[\Theta]{\matr{L}^{\dagger}}^{-1}} \leq r-1.
\end{equation}
therefore, 
\begin{equation}
\left(\frac{N}{m}\right) \sqnormvec{\msubgen[\{1\}, \Theta]{\matr{L}^{\dagger}}\msub[\Theta]{\matr{L}^{\dagger}}^{-1}}
\leq \left(\frac{N}{m} \right)(r-1).
\end{equation}

\subsection{Bounding  (\ref{eq:GLR_Proof_intro:small2})}
\label{subapp:GLR_Proof:off_diag_err}
We use the Bloomfield-Watson-Knott extension of Kantorovich's inequality \cite[Thm 5]{khatri1982some}. Let $\matr{X} = \matr{Y} = \msub[\set{N},\Theta]{I}$, $\matr{G}=\matr{L}^{\dagger}$ and $\matr{H} = \matr{L}\matr{L}^{\dagger} =\matr{I}_{N} - \frac{1}{N}\vect{1}_{N \times N}$ in \cite[Thm 5]{khatri1982some}. As $\msub[\Theta]{\matr{H}^{2}} = \msub[\Theta]{\matr{I}}$,

\begin{align} 
    &\sqfrob{\msub[N,\Theta]{\matr{L}^{\dagger}}\msub[\Theta]{\matr{L}^{\dagger}}^{-1}} \leq \sum_{i=2}^{m}\omega\left(\max\left[1,\frac{\lambda_{N+2-i}}{\lambda_{i}}\right]\right)
\end{align}
The same checks we ran in Section \ref{subapp:GLR_Proof:adjustment_term} apply to \cite[Thm 5]{khatri1982some}, and so the constraints hold.

In \cite[Thm 5]{khatri1982some}, in the first part of (2.14), $\omega_{i} = \omega\left(\frac{\lambda_{N+2-i}}{\lambda_{i}}\right)$ -- these correspond to the case where $\frac{\lambda_{N+2-i}}{\lambda_{i}} \geq 1$. We have weakened the second part of (2.14) to be increasing by $1$ each time (as $\omega(1)=1$) for ease of presentation; this does not affect the downstream effectiveness of the bound, and lets us have the same form as the bandlimited bound.

\section{Bounding $\xi_{1}(\set{S})$}
\label{app:proof_unif_ub_xi_1_MSE}
\begin{lemma}
\label{lemma:unif_ub_xi_1_GLR}
Under GLR reconstruction,
\begin{equation}
    \xi_{1}(\set{S}) \leq k + B_{k}(m)
\end{equation}
    where $m=|\set{S}|$ and $B_{k}(m)$ is as defined in Lemma \ref{lemma:GLR_xi_2_bound_main_bl}.
\end{lemma}
\begin{proof}
    
Note that 
\begin{align}
    \matr{R}_{\set{S}}\matrsub[S,N]{I} &= (\matr{\Pi}_{\set{S}} + \mu\matr{L})^{-1}\matr{\Pi}_{\set{S}}  \\
    &= \matr{I} - (\matr{\Pi}_{\set{S}} + \mu\matr{L})^{-1}\mu \matr{L}
\end{align}
we have
\begin{align}
    \xi_{1}(\set{S}) &= \sqfrob{\matrsubU{N} - \matr{R}_{\set{S}}\matrsubU{S}} \\
    &= \textrm{tr}(\matr{I}_{k} - 2\matrsubU{N}^{T}\matr{R}_{\set{S}}\matrsubU{S}) + \sqfrob{\matr{R}_{\set{S}}\matrsubU{S}} \\
    &= \textrm{tr}\left(2\left(\matrsubU{N}^{T}(\matr{\Pi}_{\set{S}} + \mu\matr{L})^{-1}\mu \matr{L}\matrsubU{N}\right) - \matr{I}_{k}\right)  \\
    &\quad + 
 \sqfrob{\matr{R}_{\set{S}}\matrsubU{S}} \\
    &= 2\left(\sum_{i=1}^{k}\vect{u}_{i}^{T}(\matr{\Pi}_{\set{S}} + \mu\matr{L})^{-1}\vect{u}_{i}\mu\lambda_{i}))\right) - k  \\
    &\quad + \sqfrob{\matr{R}_{\set{S}}\matrsubU{S}}
\end{align}
By \cite[Eq. (1.7)]{nordstrom2011convexity} for $i>1, \vect{u}_{i}^{T}(\matr{\Pi}_{\set{S}} + \mu\matr{L})^{-1}\vect{u}_{i} \leq \vect{u}_{i}^{T}(\mu\matr{L})^{+}\vect{u}_{i} = (\mu\lambda_{i})^{-1}$, and for $i=1, \lambda_{1} = 0$, so
\begin{equation}
    \left(\sum_{i=1}^{k}\vect{u}_{i}^{T}(\matr{\Pi}_{\set{S}} + \mu\matr{L})^{-1}\vect{u}_{i}\mu\lambda_{i}))\right) \leq k-1 
\end{equation}
and therefore
\begin{equation}
    \xi_{1}(\set{S}) \leq (k-1) + \sqfrob{\matr{R}_{\set{S}}\matrsubU{S}}.
\end{equation}
As $\sqfrob{\matr{R}_{\set{S}}\matrsubU{S}} = \xi_{2,bl}$, by Lemma \ref{lemma:GLR_xi_2_bound_main_bl},
\begin{equation}
    \xi_{1}(\set{S}) \leq k + B_{k}(m).
\end{equation}
\end{proof}

\section{Proof of Proposition \ref{propn:GLR_simple}}
\label{app:proof_propn_GLR_simple}
As in the proof of Theorem \ref{thm:main_GLR_exist}, we show that $\Delta_{2} > 0$ and then apply Corollary \ref{corr:main_GLR_iff}. We use the following bounds and equalities:

\begin{align}
B_{weak}(m) &= r\left(\frac{N}{m} + m-1 \right) \\
    \xi_{2}(\set{S}) &\leq B_{weak}(m)\\
    \xi_{1}(\set{S}) &\leq k +  B_{weak}(m)\\
    \xi_{1}(\set{N}) &= \sum_{i=1}^{k} \left(1 - \frac{1}{1+\mu\lambda_{i}}\right)^{2} \\
    \xi_{2}(\set{N}) & = \sum_{i=1}^{N} \left( \frac{1}{1+\mu\lambda_{i}}\right)^{2}
\end{align}
These are proven in Lemma \ref{lemma:GLR_xi_2_bound_main}, Appendix \ref{app:proof_unif_ub_xi_1_MSE} Lemma \ref{lemma:unif_ub_xi_1_GLR}, and Lemma \ref{lemma:GLR_full_observation_MSE}. We also use that $\xi_{2,bl}(\set{S}) = \sqfrob{\matr{R}_{\set{S}}\matrsubU{S}} \leq \sqfrob{\matr{R}_{\set{S}}} = \xi_{2}(\set{S})$, which follows because $\matrsubU{N}\matrsubU{N}^{T} = \projbl$ is a projection and thus reduces the Frobenius norm.
We in fact use the following lower bounds:
\begin{align}
    \xi_{2}(N) &\geq N (1 + \mu\bar{\lambda})^{-2} \\
    \xi_{1}(N) &\geq 0
\end{align}
where the first comes via applying Jensen's inequality, and the second by positivity of quadratics.

We therefore see that, as $\Delta_{i}(\set{N},\set{S}^{c}) = \xi_{i}(\set{N}) - \xi_{i}(\set{S})$,
\begin{align}
    \Delta_{2}(\set{N},\set{S}^{c}) &\geq N(1 + \mu\bar{\lambda})^{-2} - B_{weak}(m) \label{eq:proof_GLR_simple:delta_2_lb}\\
    \Delta_{1}(\set{N},\set{S}^{c}) &\geq 0 - (k + B_{weak}(m)).
\end{align}
We now consider a sets $\set{S}$ of size $\lceil\sqrt{N}\rceil$ and show that $\Delta_{2}(\set{N},\set{S}^{c}) > 0$. By assumption $2r\sqrt{N} < N$ so $\sqrt[4]{N}\cdot(2r)^{-\frac{1}{2}} > 1$ and therefore $\mu_{ub\_weak} > 0$ and is real and it therefore possible to pick $0 < \mu < \mu_{ub\_weak}$.
By assumption $0<\mu<\mu_{ub\_weak}$, so 
\begin{align}
     \frac{N}{\left(1+\mu\ \bar{\lambda} \right)^{2}}
    &> \frac{N}{\left(1+\mu_{ub\_weak}\bar{\lambda} 
  \right)^{2}}\\ 
    &= 2r\sqrt{N} \\
    &= r \left(\frac{N}{\sqrt{N}} + \sqrt{N} \right) \\
    &\geq r\left(\frac{N}{\lceil\sqrt{N}\rceil} + \lceil\sqrt{N}\rceil - 1 \right) = B_{weak}(\lceil\sqrt{N}\rceil)
\end{align}
And therefore by (\ref{eq:proof_GLR_simple:delta_2_lb}) (\ref{eq:proof_main_GLR_exist:delta_2_lb}). We then applying Corollary \ref{corr:main_GLR_iff} in the same was as in the proof of Theorem \ref{thm:main_GLR_exist}, which mostly finishes the proof.

The remainder is computing $\tau_{GLR\_weak}$; this is done in the same manner as in Theorem \ref{thm:main_GLR_exist}, except using the bounds for $\Delta_{1}$ and $\Delta_{2}$ presented earlier in this proof, along with $B_{weak}(\lceil\sqrt{N}\rceil) \leq 2r\sqrt{N}$. which was presented above.

\section{Proof of Remark \ref{remark:mopt}}
\label{app:proof_of_remark_GLR_mopt}
We first show $m_{opt} \leq \lceil\frac{N+1}{2}\rceil$. 

First note that if $x > 0$ then $ \frac{1}{2}\left(\sqrt{x} + \frac{1}{\sqrt{x}}\right) \geq \sqrt{\sqrt{x}\frac{1}{\sqrt{x}}} = 1$ by the AM-GM inequality, and so $\forall x, \, \omega(x) \geq 1$. Secondly, $\omega$ is increasing. Therefore
\begin{equation}
    r\frac{N}{m} + m - 1 \leq B(m) \leq r\left(\frac{N}{m} + m - 1\right).
\end{equation}
If $f(m) < g(m)$ then $\min f(m) < \min g(m)$ so
\begin{align}
    2\sqrt{rN} - 1 \leq B(m_{opt}) < N
\end{align}
the LHS has a minimum at $m=\sqrt{rN}$ with value $2\sqrt{rN} - 1$ and so because $\min(\text{LHS}) < \min B(m)$, under the assumption $B(m_{opt}) < N$ we have
\begin{equation}
    2\sqrt{rN} -1 \leq N \label{eq:proof_Bmopt_r_bound}
\end{equation}
For $m \geq \left\lfloor \frac{N}{2}\right\rfloor$,
\begin{equation}
    B(m) = r\frac{N}{m} + m - 1 + c
\end{equation}
for a constant $c \geq 0$ not dependent on $m$. The function $r\frac{N}{m} + m - 1 + c$ is increasing for $m \geq \sqrt{rN}$ and by \ref{eq:proof_Bmopt_r_bound} $\lceil\sqrt{rN}\rceil \leq \lceil\frac{N+1}{2}\rceil$. Therefore, considering that $m$ is discrete, $m_{opt} \leq \lceil\frac{N+1}{2}\rceil$.

We now show that $m_{opt} \in \left[ \left\lfloor \sqrt{N} \right\rfloor, \left\lceil \sqrt{rN} \right\rceil \right]$. Write $r_{i} = \omega\left(\max\left[1,\frac{N+2-i}{\lambda_{i}}\right]\right)$. Then 
\begin{equation}
    B(m) = r\frac{N}{m} + \sum_{i=2}^{m} r_{i}
\end{equation}
and
\begin{equation}
    B(m+1) - B(m) = r_{m+1} - \frac{rN}{m(m+1)}
\end{equation}
As $B(m_{opt})$ is a global minimum it is a local minimum, so
$B(m_{opt} + 1) \geq B(m_{opt})$ and $B(m_{opt} - 1) \geq B(m_{opt})$. This can be written:
\begin{align}
     r_{m_{opt}} m_{opt}(m_{opt} + 1)&\geq rN \\
    r_{m_{opt}-1}m_{opt}(m_{opt} - 1) &\leq rN
\end{align}
As $1 \leq r_{i} \leq r$ for all $i$, we get
\begin{align}
    m_{opt}(m_{opt} + 1)&\geq N  \\
    m_{opt}(m_{opt} - 1) &\leq rN 
\end{align}
and as $(m+1)^{2} > m(m+1)$ and $(m-1)^{2} < m(m-1)$,
\begin{equation}
    \sqrt{N} - 1 < m_{opt} < \sqrt{rN} + 1
\end{equation}
as these inequalities are strict and $m_{opt}$ is an integer,
\begin{equation}
\left\lfloor\sqrt{N}\right\rfloor \leq m_{opt} \leq \left\lceil\sqrt{rN}\right\rceil.
\end{equation}

\section{Sensitivity Analysis of Theorem \ref{thm:main_GLR_exist} and Proposition \ref{propn:GLR_simple}}
\label{app:GLR_sensitivity}

Proposition \ref{propn:GLR_simple} is centered around the graph property $r$\footnote{We give some intuition for $r$: low $r$, equivalent to low $\frac{\lambda_{N}}{\lambda_{2}}$, is a known condition in the Network Sychronisation literature which allows for dynamic oscillators on a network to synchronise \cite{barahona2002synchronization}.} which measures the spread of the eigenvalues of $\matr{L}$; $r$ is increasing in $\frac{\lambda_{N}}{\lambda_{2}}$. $r$ can also be understood as a bound on degree spread; let $d_{min}$ and $d_{max}$ be the minimum and maximum node degrees on the graph, then $4r \geq \frac{d_{max}}{{d_{min}}}$ (via the min-max theorem).
From condition (\ref{eq:weak_GLR_constraint}) we require $r$ to be small for the theorem to hold -- which means the degree spread needs to be narrow. %

To better understand $\mu_{ub\_weak}$ and $\tau_{ub\_weak}$, we present a sensitivity analysis; this analysis varies one parameter of $\bar{\lambda}$, $N$, $r$ and $\mu$ while holding the others constant.

\begin{table}[h]
\caption{Sensitivity analysis for Proposition \ref{propn:GLR_simple} .}
\centering
\begin{tabular}{|l|c|c|c|c|} \hline
     \rule{0pt}{8pt}& $\bar{\lambda}$ & $N$ & $r$ & $\mu$ \\ \hline
    \rule{0pt}{8pt} $\mu_{ub\_weak}$ & $\mathcal{O}(\bar{\lambda}^{-1})$ & $\mathcal{O}(\sqrt[4]{N})$ & $\mathcal{O}(r^{-\frac{1}{2}})$ &  N/A \\
    \rule{0pt}{8pt}$\tau_{GLR\_weak}$ & $\mathcal{O}(\bar{\lambda}^{-2})$ & $\mathcal{O}(N^{-\frac{1}{2}})$ & $\mathcal{O}(r^{-1})$ & $\mathcal{O}(\mu^{-2})$ \\ \hline
\end{tabular}
\label{tbl:GLR_sensitivity_simple}
\end{table}

We see both $\mu_{ub\_weak}$ and $\tau_{GLR\_weak}$ are decreasing in $r$, and furthermore the constraint (\ref{eq:weak_GLR_constraint}) also requires low $r$. This shows that graphs with low $r$ are more amenable to reconstruction with fewer observations via GLR.

Proposition \ref{propn:GLR_simple} summarises the full spectrum of $\matr{L}$ via $\lambda_{2}$, $\bar{\lambda}$ and $\lambda_{N}$ and thus can fail to show that reducing sample size reduces MSE for graph models like the Barabasi-Albert graph model which have a few very large eigenvalues with an otherwise generally concentrated spectrum. Theorem \ref{thm:main_GLR_exist} overcomes this by using the full spectrum; its broader applicability makes it worth understanding.

We briefly talk directly about the parameters of to Theorem \ref{thm:main_GLR_exist}; as $r$ decreases, $m_{opt}$ decreases to $\sqrt{N}$ while $\mu_{ub}$ and $\tau_{GLR}$ increase, which is in line with our suggestion of $r$ as a measure of amenability to reconstruction.

\section{Asymptotic Results for GLR under Full-band Noise}
\label{app:Proof_GLR_big_N}

\begin{propn}
\label{propn:GLR_big_N_app}
Fix $p \in (0,1]$. Consider graphs drawn from the distribution of random connected Erdős–Rényi graphs with $N$ vertices and edge probability $p$. Then as $N \to \infty$, condition (\ref{eq:GLR_exist_thm_B_constraint}) in Theorem \ref{thm:main_GLR_exist} holds w.h.p. . Furthermore, as $N \to \infty$,
    
        \begin{equation}    
        r \overset{p}{\to} 1, \quad
         {m_{opt}}\cdot {{N}}^{-\frac{1}{2}} \overset{p}{\to} 1, \quad
       \mu_{ub}\lambda_{2} \overset{p}{\to} +\infty. 
    \end{equation}
    Assume $\frac{k}{N}$ is fixed and choose $\mu = \frac{c}{\lambda_{2}}$,  or $\frac{c}{\lambda_{N}}$, or $\frac{c}{\sqrt{\lambda_{2}\lambda_{N}} }$ for optimal bias-variance trade-off 
 at $\set{S} = \set{N}$ \cite{chen2017GLRbias}, then
    \begin{align}
        \tau_{GLR} \to \left(1+2c\right)^{-1}.
    \end{align}
\end{propn}
\begin{proof}[Proof Sketch]
By \cite[Theorem 1]{jiang2012low}, $\lambda_{2} \approx Np - \sqrt{2N\log N}$ and $\lambda_{N} \approx Np + \sqrt{2N\log N}$ as $N \to \infty$ so $\frac{\lambda_{N}}{\lambda_{2}} \overset{p}{\to} 1$ and $r \overset{p}{\to} 1$ . Approximately, $B(m) \to \frac{N}{m} + m - 1$, which lets us bound $\mu_{ub}$ and $\tau_{GLR}$.  See Appendix \ref {app:Proof_GLR_big_N} for a full proof.
\end{proof}

\begin{proof}
We prove the proposition by proving some bounds relating terms to $r$, $\lambda_{2}$ and $\lambda_{N}$, then showing that $r \to 1$.
\subsection{Notation and terms}
While we use standard terminology in probability theory, for convenience we define some of it here.

An event $E_{n}$ happens \emph{with high probability} (abbreviated w.h.p.) if $\lim_{n \to \infty} \mathbb{P}(E_{n}) = 1$.
A sequence of random variables $X_{n}$ \emph{converges in probability} to a random variable or constant $X$ if $\forall \epsilon>0:\, \mathbb{P}(\left\lvert X_{n} - X \right\rvert > \epsilon) \to 0$. We write this as
\begin{align}
    X_{n} \overset{p}{\to} X.
\end{align}
Similarly, we write
\begin{align}
    X_{n} \overset{p}{\to} +\infty.
\end{align}
to mean that $\forall c>0:\, \mathbb{P}( X < c) \to 0$
\subsection{Limits in probability}
In this section, we show $r \overset{p}{\to} 1$ in our setting. Morally our argument is that $\frac{\lambda_{N}}{\lambda_{2}} \approx \frac{Np + \sqrt{2N\log N}}{Np - \sqrt{2N\log N}} \overset{p}{\to} 1$. We now prove this formally, starting with statements including disconnected graphs, and using them to derive results about connected graphs.
By \cite[Theorem 1 (i) \& (ii)]{jiang2012low}, across all Erdős–Rényi graphs with edge probability $p$
\begin{align}  
    \frac{ Np - \lambda_{2}}{\sqrt{N\log N}} \overset{p}{\to} \sqrt{2} \\
    \frac{ \lambda_{N} - Np}{\sqrt{N\log N}} \overset{p}{\to} \sqrt{2}.
\end{align}
By Slutsky's Theorem, and as convergence in distribution to a constant implies convergence in probability, we add and square the ratios to get
\begin{align}
    \frac{\left(\lambda_{N} - \lambda_{2}\right)^{2}}{N\log N} \overset{p}{\to} 8. \label{eq:proof_by_slutsky_sq_diff}
\end{align}
We now bound $\frac{N\log N}{\lambda_{2}\lambda_{N}}$. 
By the definition of convergence in probability, because $\frac{Np}{N \log N} \to \infty$ and by the triangle inequality,
\begin{align}
    &\forall \epsilon > 0, \lim_{N \to \infty}\mathbb{P}\left( \left\lvert \frac{ Np - \lambda_{2}}{\sqrt{N\log N}} - \sqrt{2} \right\rvert > \epsilon \right) = 0 \\
    \implies &\forall c > 0, \lim_{N \to \infty}\mathbb{P}\left( \left\lvert \frac{ \lambda_{2}}{\sqrt{N\log N}} \right\rvert < c \right) = 0 \label{eq:Proof_limit_prob_connected} \\
    \implies &\forall \epsilon > 0, \lim_{N \to \infty}\mathbb{P}\left( \left\lvert \frac{\sqrt{N\log N}}{ \lambda_{2}} \mathbbm{1}\{\lambda_{2} > 0\}\right\rvert > \epsilon  \right) = 0
\end{align}
therefore
\begin{align}
    \frac{\sqrt{N\log N}}{ \lambda_{2}} \mathbbm{1}\{ \lambda_{2} > 0 \} \overset{p}{\to} 0
\end{align}
As $0 < \frac{\sqrt{N\log N}}{ \lambda_{N}} < \frac{\sqrt{N\log N}}{ \lambda_{2}}$,
\begin{align}
    \frac{{N\log N}}{ \lambda_{2} \lambda_{N}} \mathbbm{1}\{ \lambda_{2} > 0 \} \overset{p}{\to} 0.
\end{align}
we use Slutsky to multiply this with (\ref{eq:proof_by_slutsky_sq_diff}) to get
\begin{align}
    (r-1)\mathbbm{1}\{ \lambda_{2} > 0 \} =  \frac{\left(\lambda_{N} - \lambda_{2}\right)^{2}}{ 4\lambda_{2} \lambda_{N}} \mathbbm{1}\{ \lambda_{2} > 0 \} \overset{p}{\to} 0
\end{align}

By (\ref{eq:Proof_limit_prob_connected}) we know $\lambda_{2} > 0$ w.h.p. . Therefore, for any $\epsilon > 0$
\begin{align}
    &\mathbb{P}\left( \lvert r-1 \rvert > \epsilon \,\big\vert\, \lambda_{2}  > 0 \right)  \\&= \frac{\mathbb{P}\left( \lvert r-1 \rvert> \epsilon \cap \lambda_{2}  > 0 \right)}{\mathbb{P}\left( \lambda_{2}  > 0 \right)} \\
    &= \frac{\mathbb{P}\left( \lvert r-1 \rvert \mathbbm{1}\{\lambda_{2}  > 0\} > \epsilon \right)}{\mathbb{P}\left( \lambda_{2}  > 0 \right)} \to \frac{0}{1}
\end{align}
Therefore under our setting of considering only connected graphs (i.e. that $\lambda_{2} > 0$), $r \overset{p}{\to} 1$.

\subsection{Parameter bounds}
We first bound $B(m)$ in terms of $r$. Note that $\omega$ is increasing and therefore for $2 \leq i \leq N$,
\begin{align}
   \omega\left(\max\left[1,\frac{\lambda_{N+2-i}}{\lambda_{i}}\right]\right) \leq \omega\left(\frac{\lambda_{N}}{\lambda_{2}}\right) = r.
\end{align}
If $x > 0$ then $ \frac{1}{2}\left(\sqrt{x} + \frac{1}{\sqrt{x}}\right) \geq \sqrt{\sqrt{x}\frac{1}{\sqrt{x}}} = 1$ by the AM-GM inequality, and so $\forall x, \, \omega(x) \geq 1$. Therefore
\begin{align}
    1 &\leq \omega\left(\max\left[1,\frac{\lambda_{N+2-i}}{\lambda_{i}}\right]\right) \leq r \\
    \frac{N}{m} + m - 1 &\leq B(m) \leq r\left(\frac{N}{m} + m - 1\right).
\end{align}
We also have that
\begin{align}
   \frac{1}{\lambda_{N}} \left( \sqrt{\frac{N}{B(m)}}-1 \right) \leq \mu_{ub}\left(m\right) \leq \frac{1}{\lambda_{2}} \left( \sqrt{\frac{N}{B(m)}}-1 \right).
\end{align}

\subsection{Pulling it together}
We already have that $r \overset{p}{\to} 1$. As $\sqrt{N} \leq m_{opt} \leq \left\lceil\sqrt{rN}\right\rceil$, $\frac{m_{opt}}{\sqrt{N}} \overset{p}{\to} 1$. By the squeeze theorem, for a fixed $m$,
\begin{align}
    \frac{B(m)}{N} \overset{p}{\to} \frac{1}{m}.
\end{align}
Because $\frac{1}{m} < 1$ for $m > 1$, condition (\ref{eq:GLR_exist_thm_B_constraint}) in Theorem \ref{thm:main_GLR_exist} holds w.h.p. as $N \to \infty$.
As $B(m) > 0$,
\begin{align}
    \sqrt{\frac{N}{B(m)}} - 1 \overset{p}{\to} \sqrt{m} - 1
\end{align}

Conditioning on $\lambda_{2} > 0$ does not change (\ref{eq:Proof_limit_prob_connected}), so it applies when we only consider the set of connected Erdős–Rényi graphs as well. By setting $c = \sqrt{m} - 1$ in (\ref{eq:Proof_limit_prob_connected}), we get 
\begin{align}
    \frac{1}{\lambda_{2}} \left( \sqrt{\frac{N}{B(m)}}-1 \right) \overset{p}{\to} 0.
\end{align}
and therefore for fixed $m$, $\mu_{ub}(m) \to 0$. We now calculate $\mu_{ub}(m_{opt})$. First note that as $x \mapsto x + \frac{1}{x}$ invertible and monotone on $x \in (0,1)$, so its inverse is continuous. Consider that $\frac{\lambda_{2}}{\lambda_{N}} \in (0,1)$. As the inverse is continuous, we have that $\forall \epsilon > 0 \,\exists \delta > 0:\, \left\lvert\left( \frac{\lambda_{2}}{\lambda_{N}} + \frac{\lambda_{N}}{\lambda{2}}\right) - 2\right\rvert < \epsilon \implies \left\lvert\frac{\lambda_{2}}{\lambda_{N}} - 1\right\rvert < \delta$. Using this with the definition of convergence in probability and that $r \overset{p}{\to} 1$ gives
\begin{equation}
    \frac{\lambda_{2}}{\lambda_{N}} \overset{p}{\to} 1.
\end{equation}

Next, note that
\begin{align}
    \frac{\sqrt{N}}{m_{opt}} + \frac{m_{opt}}{\sqrt{N}} - \frac{1}{\sqrt{N}} &\leq \frac{B(m_{opt})}{\sqrt{N}} \leq r\left( \frac{\sqrt{N}}{m_{opt}} + \frac{m_{opt}}{\sqrt{N}} \right).
\end{align}

so 
\begin{align}
 \frac{m_{opt}}{\sqrt{N}} \to 1 \implies \frac{B(m_{opt})}{\sqrt{N}} \overset{p}{\to} 2.  \label{eq:Proof_GLR_big_N:B_mopt}
\end{align}
Therefore (\ref{eq:GLR_exist_thm_B_constraint}) in Theorem \ref{thm:main_GLR_exist} holds for $m_{opt}$ w.h.p. as $N \to \infty$.
Also
\begin{align} \frac{\mu_{ub}\left(m_{opt}\right)\lambda_{2}}{\sqrt[4]{N}} \geq &\frac{\lambda_{2}}{\lambda_{N}} \left( \sqrt{\frac{\sqrt{N}}{B(m_{opt})}}-\frac{1}{\sqrt[4]{N}} \right) \\
&\overset{p}{\to} \frac{1}{\sqrt{2}}
\end{align}
and therefore ${\mu_{ub}\left(m_{opt}\right)\lambda_{2}} \overset{p}{\to} \infty$.

Finally, we bound $\tau_{GLR}$. We first take limits of $\left(\frac{1}{1 + \mu\lambda_{i}}\right)^{2}$. Note that
\begin{align}
    1 \leq \frac{\lambda_{i}}{\lambda_{2}} \leq \frac{\lambda_{N}}{\lambda_{2}} \overset{p}{\to} 1, \\
    1 \geq \frac{\lambda_{i}}{\lambda_{N}} \geq \frac{\lambda_{2}}{\lambda_{N}} \overset{p}{\to} 1
\end{align}
and $\frac{\lambda_{i}}{\sqrt{\lambda_{2}\lambda_{N}}} \in \left[ \frac{\lambda_{i}}{\lambda_{N}}, \frac{\lambda_{i}}{\lambda_{2}} \right]$ and so $\frac{\lambda_{i}}{\sqrt{\lambda_{2}\lambda_{N}}} \overset{p}{\to} 1$. Therefore, for all three choices of $\mu$, $\mu\lambda_{i} \overset{p}{\to} c$. As

\begin{align}
    \frac{1}{N} + \frac{N-1}{N}\left(\frac{1}{1 + \mu\lambda_{N}}\right)^{2} &\leq \frac{1}{N}\sum_{i = 1}^{N}\left(\frac{1}{1+  \mu \lambda_{i}} \right)^{2} \\
    \frac{1}{N} + \frac{N-1}{N}\left(\frac{1}{1 + \mu\lambda_{2}}\right)^{2} &\geq \frac{1}{N}\sum_{i = 1}^{N}\left(\frac{1}{1+  \mu \lambda_{i}} \right)^{2}
\end{align}
For all three choices of $\mu$, 
\begin{equation}
    \frac{1}{N}\sum_{i = 1}^{N}\left(\frac{1}{1+  \mu \lambda_{i}} \right)^{2} \to \left(\frac{1}{1 + c}\right)^{2}.
\end{equation}
Similarly,
\begin{equation}
    \frac{1}{N}\sum_{i = 1}^{N}\left(1 - \frac{1}{1+  \mu \lambda_{i}} \right)^{2} \to \left(1 - \frac{1}{1 + c}\right)^{2}.
\end{equation}
As $\frac{B(m_{opt})}{\sqrt{N}} \overset{p}{\to} 2$, we must have that $\frac{B(m_{opt})}{{N}} \overset{p}{\to} 0$ and as $\frac{k}{N}$ is constant, $\frac{B(m_{opt})}{k} \overset{p}{\to} 0$. Therefore, by Slutsky, for the three choices of optimal $\mu$,
\begin{align}
    \tau_{GLR}(\mu,m_{opt}) \to \frac{\left(\frac{1}{1 + c}\right)^{2}}{1 - \left(1 - \frac{1}{1 + c}\right)^{2}} = \frac{1}{1+2c}.
\end{align}
\end{proof}

\section{Proof of Lemma \ref{lemma:GLR_xi_2_bound_main_bl}}
\label{app:GLR_bandlimited_lemma_proof}
We follow the structure and notation of Appendix \ref{lemma:GLR_xi_2_bound_main} and follow the same broad outline. However, there are two additional complexities: firstly, we cannot remove the dependence on $\mu$ at the same stage as we did in the proof of Lemma \ref{lemma:GLR_xi_2_bound_main} because the terms are not of the form $\matr{X}\msub[\Theta]{\mu\matr{I} + \matr{L}^{\dagger}}^{-2}\matr{X}^{T}$, as $\matrsubU{S}$ breaks the symmetry we were using. We overcome this by only removing the dependence on $\mu$ at the end of the proof.

Secondly, we cannot directly apply the matrix versions of Kantorovich's inequality, because $\proj{\Theta}$ and $\projbl$ do not commute, thus breaking the assumptions needed. We prove two bandlimited variants of Kantorovich instead, one which works by directly building the proof from the ground up, and the second of which works by finding a matrix that commutes with $\matr{L}\matr{L}^{\dagger}$ and upper bounds $\projbl$ without weakening the proof too much.

\subsection{Notation and useful matrices}
While we follow the notation described in Appendix \ref{app:GLR_bandlimited_thm}, we also define the following matrices:
\begin{align}
     \matr{L}^{\dagger}_{\mu} &= \mu\matr{I} + \matr{L}^{\dagger} \\
     \projgen{\matr{L}} &= \matr{I} - \frac{1}{N}\vect{1}_{N \times N} = \matr{L}^{\dagger}\matr{L} = \matr{L}\matr{L}^{\dagger}
\end{align}
where the latter can be seen to be a projection operator as $ \projgen{\matr{L}} = \projgen{\matr{L}}^{T} = \projgen{\matr{L}}^{2}$. Futhermore, $\projgen{\matr{L}}\matr{L}^{\dagger} = \matr{L}^{\dagger}$ and because $\vect{1}_{N}\matrsub[N,\Theta]{I} = 0$, $\projgen{\matr{L}}\matrsub[N,\Theta]{I} = \matrsub[N,\Theta]{I}$.

\subsection{Proof overview}
Overall we show that 
\begin{align}
    &\xi_{2,bl}(\set{S}) = \sqfrob{\matr{R}_{\set{S}}\matrsubU{S}} \\
    &= 1 + \sqfrob{\matr{R}_{\set{S}}\msubgen[\set{S},2:k]{\matr{U}}} \\
    &= 1 + \sqfrob{\msubgen[\set{N},\Theta]{\matr{L}^{\dagger}} \msubgen[\Theta]{\mu\matr{I} + \matr{L}^{\dagger}}^{-1}\msubgen[\set{\Theta},2:k]{\matr{U}} } \\
    &+ {\frac{N}{m}}\sqnormvec{\vect{v}_{1}^{T}\left( \matr{I} - \msubgen[\set{N},\Theta]{\matr{L}^{\dagger}} \msubgen[\Theta]{ \matr{L}_{\mu}^{\dagger}}^{-1}\msubgen[\Theta,\set{N}]{\matr{I}}\right)\msubgen[\set{N},2:k]{\matr{U}}} \\
    &\leq  1 + B_{k}(m)
\end{align}
by proving that
\begin{align}
    &{\frac{N}{m}}\sqnormvec{\vect{v}_{1}^{T}\left( \matr{I} - \msubgen[\set{N},\Theta]{\matr{L}^{\dagger}} \msubgen[\Theta]{ \matr{L}_{\mu}^{\dagger}}^{-1}\msubgen[\Theta,\set{N}]{\matr{I}}\right)\msubgen[\set{N},2:k]{\matr{U}}} \\
    &\leq \frac{N}{m}r_{bl}
\end{align}
and
\begin{align}
   &\sqfrob{\msubgen[\set{N},\Theta]{\matr{L}^{\dagger}} \msubgen[\Theta]{\mu\matr{I} + \matr{L}^{\dagger}}^{-1}\msubgen[\set{\Theta},2:k]{\matr{U}} }  \\
   &\leq \sum^{m}_{i=1}\omega\left(\max\left[1,\frac{\lambda_{k+2-i}}{\lambda_{i}}\right]\right).
\end{align}

\subsection{Row decomposition}
Noting that $\matr{R}_{\set{S}}\vectsub[S]{{u}_{1}} = \frac{1}{\sqrt{N}}\matr{R}_\set{S}\vect{1}_{S} = \frac{1}{\sqrt{N}}\vect{1}_{N}$,
\begin{align}
    \xi_{2,bl}(\set{S}) &= \sqfrob{\matr{R}_{\set{S}}\matrsubU{S}} \\
    &= \sqnormvec{\matr{R}_{\set{S}}\vectsub[S]{u_{1}}} + \sqfrob{\matr{R}_{\set{S}}\msubgen[\set{S},2:k]{\matr{U}}} \\
    &= 1 + \sqfrob{\matr{R}_{\set{S}}\msubgen[\set{S},2:k]{\matr{U}}}.
\end{align}

\subsection{Explicit submatrix computation}
We first show
\begin{align}
    &\left(\proj{S} + \mu\matr{L}\right)^{-1}\projgen{\matr{L}} \\
    &= \frac{1}{\mu}\matr{P}^{T}\left( \matr{L}^{\dagger} - \msubgen[\set{N},\Theta]{\matr{L}^{\dagger}} \msubgen[\Theta]{ \matr{L}^{\dagger}_{\mu}}^{-1}\msubgen[\Theta,\set{N}]{\matr{L}^{\dagger}}\right)
\end{align}
\begin{proof}
    From Appendix \ref{app:GLR_bandlimited_thm}, $\left(\proj{S} + \mu\matr{L} \right)\matr{P}^{T} = \proj{\Theta} + \mu\matr{L}$, so
    \begin{align}
        \left(\proj{S} + \mu\matr{L}\right)\frac{1}{\mu}\matr{P}^{T}\matr{L}^{\dagger} = \left(\frac{1}{\mu}\proj{\Theta}\matr{L}^{\dagger} + \projgen{\matr{L}} \right)
    \end{align}
    as $\msub[N,\Theta]{\projgen{\matr{L}}} = \matrsub[N,\Theta]{I}$,
    \begin{align}
        &\left(\proj{S} + \mu\matr{L}\right)\frac{1}{\mu}\matr{P}^{T}\msub[N,\Theta]{\matr{L}^{\dagger}}\msub[\Theta]{\matr{L}^{\dagger}_{\mu}}^{-1} \\
        &= \left(\frac{1}{\mu}\proj{\Theta}\matr{L}^{\dagger} + \projgen{\matr{L}} \right)\matrsub[N,\Theta]{I}\msub[\Theta]{\matr{L}^{\dagger}_{\mu}}^{-1} \\
        &= \frac{1}{\mu}\matrsub[N,\Theta]{I} \left(\msub[\Theta]{\matr{L}^{\dagger}} + \mu\matrsub[\Theta]{I} \right)\msub[\Theta]{\matr{L}^{\dagger}_{\mu}}^{-1} \\
        &= \frac{1}{\mu}\matrsub[N,\Theta]{I}
    \end{align}
    Then
    \begin{align}
        &\left(\proj{S} + \mu\matr{L}\right)\frac{1}{\mu}\matr{P}^{T}\msub[N,\Theta]{\matr{L}^{\dagger}}\msub[\Theta]{\matr{L}^{\dagger}_{\mu}}^{-1}\msub[\Theta, N]{\matr{L}^{\dagger}} \\
        &= \frac{1}{\mu}\matrsub[N,\Theta]{I}\msub[\Theta, N]{\matr{L}^{\dagger}} = \frac{1}{\mu}\proj{\Theta}\matr{L}^{\dagger}
    \end{align}
Therefore
\begin{align}
    &\left(\proj{S} + \mu\matr{L}\right) \frac{1}{\mu}\matr{P}^{T}\left( \matr{L}^{\dagger} - \msubgen[\set{N},\Theta]{\matr{L}^{\dagger}} \msubgen[\Theta]{ \matr{L}^{\dagger}_{\mu}}^{-1}\msubgen[\Theta,\set{N}]{\matr{L}^{\dagger}}\right) \\
    &= \left(\frac{1}{\mu}\proj{\Theta}\matr{L}^{\dagger} + \projgen{\matr{L}} \right) - \frac{1}{\mu}\proj{\Theta}\matr{L}^{\dagger} \\
    &= \projgen{\matr{L}}.
\end{align}
Multiplying both sides by $\left(\proj{S} + \mu\matr{L}\right)$ finishes the proof.
\end{proof}

We use this to compute $\left(\proj{S} + \mu\matr{L} \right)^{-1}\proj{S}$. 
\begin{equation}
 \left(\proj{S} + \mu\matr{L} \right)^{-1}\left(\proj{S} + \mu\matr{L} \right) = \matr{I}
\end{equation}
 means
\begin{align}
 &\left(\proj{S} + \mu\matr{L} \right)^{-1}\proj{S} 
 \\&= \matr{I} - \mu\left(\proj{S} + \mu\matr{L} \right)^{-1}\matr{L} \\
 &= \matr{I} - \mu\left(\proj{S} + \mu\matr{L} \right)^{-1}\projgen{\matr{L}}\matr{L} \\
 &= \matr{I} - \matr{P}^{T}\left( \matr{L}^{\dagger} - \msubgen[\set{N},\Theta]{\matr{L}^{\dagger}} \msubgen[\Theta]{ \matr{L}_{\mu}^{\dagger}}^{-1}\msubgen[\Theta,\set{N}]{\matr{L}^{\dagger}}\right)\matr{L}\\
 &= \matr{I} - \matr{P}^{T}\left( \projgen{\matr{L}} - \msubgen[\set{N},\Theta]{\matr{L}^{\dagger}} \msubgen[\Theta]{ \matr{L}_{\mu}^{\dagger}}^{-1}\msubgen[\Theta,\set{N}]{\matr{I}}\right)
\end{align}
As 
\begin{align}
\projgen{\matr{L}}\msubgen[\set{N},2:k]{\matr{U}} = \msubgen[\set{N},2:k]{\matr{U}}, \\
\matr{P}^{T} = \matr{I} - \frac{1}{m}\matr{1}_{N \times m}\matrsub[S,N]{I} = \matr{I} - \sqrt{\frac{N}{m}}\vect{u}_{1}\matr{v}_{1}^{T}
\end{align}
where $\vect{u}_{1} = \frac{1}{\sqrt{N}}\vect{1}_{N}$ is the null eigenvector of $\matr{L}$ and $\vect{v}_{1} = \frac{1}{\sqrt{m}}\vect{1}_{\set{S}}$. We have
\begin{align}
    &\matr{R}_{\set{S}}\msubgen[\set{S},2:k]{\matr{U}} \\ 
    &=  \msubgen[\set{N},\Theta]{\matr{L}^{\dagger}} \msubgen[\Theta]{\mu\matr{I} + \matr{L}^{\dagger}}^{-1}\msubgen[\set{\Theta},2:k]{\matr{U}} \\
    &+ \sqrt{\frac{N}{m}}\vect{u}_{1}\vect{v}_{1}^{T}\left( \matr{I} - \msubgen[\set{N},\Theta]{\matr{L}^{\dagger}} \msubgen[\Theta]{ \matr{L}_{\mu}^{\dagger}}^{-1}\msubgen[\Theta,\set{N}]{\matr{I}}\right)\msubgen[\set{N},2:k]{\matr{U}}
\end{align}

Because $\matr{X}^{T}\matr{L}^{\dagger}\vect{u}_{1}\vect{y}^{T} = \vect{0}$, we have $\sqfrob{\matr{L}^{\dagger}\matr{X} + \vect{u}_{1}\matr{y}} = \sqfrob{\matr{L}^{\dagger}\matr{X}} + \sqnormvec{\vect{y}}$ for any matrix $\matr{X}$ and any vector $\vect{y}$.
This gives that
\begin{align} 
    &\sqfrob{\matr{R}_{\set{S}}\msubgen[\set{S},2:k]{\matr{U}}} \\ 
    &=  \sqfrob{\msubgen[\set{N},\Theta]{\matr{L}^{\dagger}} \msubgen[\Theta]{\matr{L}^{\dagger}_{\mu}}^{-1}\msubgen[\set{\Theta},2:k]{\matr{U}}} \label{eq:GLR_bandlimited_proof:wide_term} \\
    &+ {\frac{N}{m}}\sqfrob{\vect{v}_{1}^{T}{\left( \matr{I} - \msubgen[\set{N},\Theta]{\matr{L}^{\dagger}} \msubgen[\Theta]{ \matr{L}_{\mu}^{\dagger}}^{-1}\msubgen[\Theta,\set{N}]{\matr{I}}\right)\msubgen[\set{N},2:k]{\matr{U}}}} \label{eq:GLR_bandlimited_proof:v1_term}
\end{align}

\subsection{Bounding (\ref{eq:GLR_bandlimited_proof:wide_term})}
We show that
\begin{align}
    \sqfrob{\msubgen[\set{N},\Theta]{\matr{L}^{\dagger}} \msubgen[\Theta]{\matr{L}^{\dagger}_{\mu}}^{-1}\msubgen[\set{\Theta},2:k]{\matr{U}}} \leq 
\sum^{m}_{i=2}\omega\left(\max\left[1,\frac{\lambda_{k+2-i}}{\lambda_{i}}\right]\right).
\end{align}

\begin{proof}
    We seek to apply \cite[Theorem 5]{khatri1982some}.
    
    For this, we first construct some matrices which are diagonal in the eigenbasis of $\matr{L}$. We present them by their values on $\vect{u}_{1}$, $\msubgen[\set{N},2:k]{\matr{U}}$ and $\matr{I} - \projbl$ in Table \ref{tbl:bandlimited_kantorovich_matrices}.We write $$\projgen{bl(2:k)} = \msubgen[\set{N},2:k]{\matr{U}} \msubgen[\set{N},2:k]{\matr{U}}^{T}.$$

\begin{table}[h!]
\caption{Matrices diagonal in the eigenbasis of $\matr{L}$.}
    \begin{center}
        \begin{tabular}{|l|c|c|c|}
     \hline
       & $\vect{u}_{1}$ & $\vect{u}_{2},\ldots,\vect{u}_{k}$ & $\vect{u}_{k+1},\ldots,\vect{u}_{N}$ \\ 
     \hline
     $\matr{G}$ & 0 & $\sqrt{\lambda_{i}(\mu + \frac{1}{\lambda_{i}})}$ & $\sqrt{\lambda_{\lceil \frac{k}{2} \rceil}(\mu + \frac{1}{\lambda_{i}})}$ \\
     \hline
     $\matr{H}$ & 0 & $\sqrt{\frac{1}{\lambda_{i}}(\mu + \frac{1}{\lambda_{i}})}$ & $\sqrt{\frac{1}{\lambda_{\lceil \frac{k}{2} \rceil}}(\mu + \frac{1}{\lambda_{i}})}$ \\
     \hline
     $\matr{G}\matr{H}$ & 0 & $\mu + \frac{1}{\lambda_{i}}$ & $\mu + \frac{1}{\lambda_{i}}$ \\
     \hline
     $\matr{H}^{\dagger}\matr{G}$ & 0 & ${\lambda_{i}}$ &  $\lambda_{\lceil \frac{k}{2} \rceil}$\\
     \hline
     $\matr{G}^{2}$ & 0 & $1 + \mu\lambda_{i}$ & $\lambda_{\lceil \frac{k}{2} \rceil}(\mu + \frac{1}{\lambda_{i}})$ \\
     \hline
     $\projgen{bl(2:k)}$ & 0 & $1$ & $0$ \\
     \hline
     $\matr{H}^{2}$ & 0 & $\frac{1}{\lambda_{i}}(\mu + \frac{1}{\lambda_{i}})$ & $\frac{1}{\lambda_{\lceil \frac{k}{2} \rceil}}(\mu + \frac{1}{\lambda_{i}})$ \\
     \hline
     $\left(\matr{L}^{\dagger}\right)^{2}$ & 0 & $\frac{1}{\lambda_{i}}^{2}$ & $\frac{1}{\lambda_{i}}^{2}$  \\
     \hline
        \end{tabular}
    \end{center}
\label{tbl:bandlimited_kantorovich_matrices}
\end{table}

From Table \ref{tbl:bandlimited_kantorovich_matrices} we have constructed $\matr{G}^{2}$ and $\matr{H}^{2}$ such that
\begin{align}
    \projgen{bl(2:k)} &\leq \matr{G}^{2} \\
    \left(\matr{L}^{\dagger}\right)^{2} &\leq \matr{H}^{2}
\end{align}

By the way we constructed ${\Theta}$,
    \begin{equation}
        \msubgen[\Theta]{\matr{G}\matr{H}} = \msubgen[\Theta]{\left(\matr{I} - \frac{1}{N}\vect{1}_{N\times N}\right)\left(\mu\matr{I} + \matr{L}^{\dagger} \right)} = \msubgen[\Theta]{\mu\matr{I} + \matr{L}^{\dagger} }. \label{eq:bandlimited_kantorovich:GH_eq}
    \end{equation}

We now note that $\matr{A}^{2} \leq \matr{B}^{2}$ means $\forall \vect{x}: \vect{x}^{T}\matr{A}^{2}\vect{x} \leq \vect{x}^{T}\matr{B}^{2}\vect{x}$ so  for all compatible $\matr{V}$, $\trace{\matr{V}^{T}\matr{A}^{2}\matr{V}} \leq \trace{\matr{V}^{T}\matr{B}^{2}\matr{V}}$. Therefore
\begin{align}
    &\sqfrob{\msubgen[\set{N}, \Theta]{\matr{L}^{\dagger}}\msubgen[\Theta]{\mu\matr{I} + \matr{L}^{\dagger}}^{-1}\msubgen[\Theta, 2:k]{\matr{U}}} \\
    &=\sqfrob{\msubgen[\set{N}, \Theta]{\matr{L}^{\dagger}}\msubgen[\Theta]{\matr{G}\matr{H}}^{-1}\msubgen[\Theta, 2:k]{\matr{U}}} \\
    &= \trace{\msubgen[\Theta, 2:k]{\matr{U}}^{T}\msubgen[\Theta]{\matr{G}\matr{H}}^{-1}\msubgen[ \Theta]{\left(\matr{L}^{\dagger}\right)^{2}}\msubgen[\Theta]{\matr{G}\matr{H}}^{-1}\msubgen[\Theta, 2:k]{\matr{U}}} \\
    &\leq \trace{\msubgen[\Theta, 2:k]{\matr{U}}^{T}\msubgen[\Theta]{\matr{G}\matr{H}}^{-1}\msubgen[ \Theta]{\matr{H}^{2}}\msubgen[\Theta]{\matr{G}\matr{H}}^{-1}\msubgen[\Theta, 2:k]{\matr{U}}} \\
    &= \trace{\msubgen[\Theta]{\matr{G}\matr{H}}^{-1}\msubgen[ \Theta]{\matr{H}^{2}}\msubgen[\Theta]{\matr{G}\matr{H}}^{-1}\msubgen[\Theta]{\projgen{bl(2:k)}}} \\
    &= \trace{\matrsub[N,\Theta]{H}\msubgen[\Theta]{\matr{G}\matr{H}}^{-1}\msubgen[\Theta]{\projgen{bl(2:k)}} \msubgen[\Theta]{\matr{G}\matr{H}}^{-1}\matrsub[\Theta,N]{H} } \\
    &\leq \trace{\matrsub[N,\Theta]{H}\msubgen[\Theta]{\matr{G}\matr{H}}^{-1}\msubgen[\Theta]{\matr{G}^{2}} \msubgen[\Theta]{\matr{G}\matr{H}}^{-1}\matrsub[\Theta,N]{H} } \\
    &= \trace{\msubgen[\Theta]{\matr{G}\matr{H}}^{-1}\msubgen[\Theta]{\matr{G}^{2}} \msubgen[\Theta]{\matr{G}\matr{H}}^{-1}\msubgen[\Theta]{\matr{H}^{2}} } \\
    &\leq \sum^{m}_{i=2}\omega\left(\max\left[1,\frac{\lambda_{k+2-i}}{\lambda_{i}}\right]\right).
\end{align}
where the equalities are via the cyclic properties or via (\ref{eq:bandlimited_kantorovich:GH_eq}).

The final inequality is by applying \cite[Theorem 5]{khatri1982some} with $\matr{X} = \matr{Y} = \matrsub[N,\Theta]{I}$. One can see that the conditions of the Theorem are fulfilled by the way we constructed $\matr{G}$ and $\matr{H}$ in Table \ref{tbl:bandlimited_kantorovich_matrices}, that $\omega_{i} = \omega\left(\max\left[1,\frac{\lambda_{k+2-i}}{\lambda_{i}}\right]\right)$ in the $t \geq k+s$ case and the  $t < k+s$ case can be weakened to our result.
\end{proof}

\subsection{A bandlimited Kantorovich lemma}
We require the following lemma to bound (\ref{eq:GLR_bandlimited_proof:v1_term}).

\begin{lemma}
\label{lemma:bandlimited_kantorovich}
Let $\matr{A}, \matr{B}$ be positive semi-definite, and let $\projgen{k}$ be a projection to the eigenbasis of $\matr{A}$ s.t.
\begin{equation}
    \exists\, 0<m\leq M: \quad m\projgen{k} \leq \matr{A}\projgen{k} \leq M\projgen{k}.
\end{equation}
 i.e. it bandlimits $\matr{A}$ to between $m$ and $M$. Furthermore, assume $\matr{A}$ has no eigenvalues less than $m$.
    
    Assume $\matr{B}\matr{A} = \matr{A}\matr{B}$ and $\matr{B} \leq \matr{I}$. Then for any projection $\proj{X}$ where $\proj{X}\matr{A}\matr{A}^{\dagger}\proj{X}$ is idempotent,
    \begin{equation}
\sigma_{max}\left(\sqrt{\matr{B}\matr{A}^{\dagger}}\proj{X}\sqrt{\matr{A}} \projgen{k} \right) \leq \sqrt{\frac{(M + m)^{2}}{4Mm}}.
    \end{equation}
\end{lemma}
\begin{proof}
    We abuse our notation to write $\proj{X} = \matrsub[N,X]{I}\matrsub[X,N]{I}$. Then
    \begin{equation}        \sqrt{\matr{A}^{\dagger}}\proj{X}\sqrt{\matr{A}} = \msubgen[\set{N},\set{X}]{\sqrt{\matr{A}^{\dagger}}}\msubgen[\set{N},\set{X}]{\sqrt{\matr{A}}}.
    \end{equation}

    Throughout this proof we will use the fact that $\proj{A}= 
 \matr{A}\matr{A}^{\dagger}$ is a projection to the full range of $\matr{A}$ and thus $\proj{A} \geq \projgen{k}$ (consider diagonalizing both, and note all eigenvalues are 1 or 0).
    Note that
    \begin{equation}
        m\matr{A}^{\dagger}(\matr{I} - \projgen{k}) \leq (\matr{I} - \projgen{k})
    \end{equation}
    because $\matr{A}$ has no non-zero eigenvalues less than $m$, so $\matr{A}^{\dagger}$ has no non-zero eigenvalues greater than $\frac{1}{m}$. Thus 
    \begin{align}
        &Mm\matr{A}^{\dagger}(\matr{I} - \projgen{k}) \leq M(\matr{I} - \projgen{k}) \quad \text{and so} \\
        & (M+m)(\matr{I} - \projgen{k}) - Mm\matr{A}^{\dagger}(\matr{I} - \projgen{k})\geq 0. \label{eq:bandlimited_kantorovich_difference}
    \end{align}
    From \cite[Theorem 2, (3.3)]{baksalary1991generalized},
    \begin{align}
        \matr{A}\projgen{k} &\leq (m+M)\projgen{k} - Mm\matr{A}^{\dagger}\projgen{k} \\
        &\leq (m+M)\proj{A} - Mm\matr{A}^{\dagger}
    \end{align}
    where the second inequality comes from adding on (\ref{eq:bandlimited_kantorovich_difference})

    Conjugating by $\matrsub[N,X]{I}$ and $\matrsub[X,N]{I}$ gives:
    \begin{align}
        &\msubgen[\set{X}]{\matr{A}\projgen{k}} \leq (M+m)\msubgen[\set{X}]{\proj{A}} - Mm\msubgen[\set{X}]{\matr{A}^{\dagger}}
    \end{align}
    We assumed $\proj{X}\proj{A}\proj{X}$ is idempotent, and it implies that $\msubgen[\set{X}]{\proj{A}}$ is idempotent so we can complete the square giving:
    \begin{equation}
        \msubgen[\set{X}]{\matr{A}\projgen{k}} \leq \frac{(M+m)^{2}}{4Mm} \msubgen[\set{X}]{\matr{A}^{\dagger}}^{\dagger}. \label{eq:bandlimited_kantorovich_proof_inner}
    \end{equation}
    In more detail: this is the final step of the proof of \cite[Theorem 2]{baksalary1991generalized} along with the observation that the constant is invariant to the transformation  $(m,M) \mapsto \left(\frac{1}{m},\frac{1}{M} \right)$.

    Conjugating (\ref{eq:bandlimited_kantorovich_proof_inner}) by $\msubgen[\set{N},\set{X}]{\sqrt{\matr{A}^{\dagger}}}$ and $\msubgen[\set{N},\set{X}]{\sqrt{\matr{A}^{\dagger}}}^{T}$:
    
    \begin{align}
        &\msubgen[\set{N},\set{X}]{\sqrt{\matr{A}^{\dagger}}}\msubgen[\set{X}]{\matr{A}\projgen{k}} \msubgen[\set{X},\set{N}]{\sqrt{\matr{A}^{\dagger}}} \\
        &\leq \frac{(M+m)^{2}}{4Mm} \msubgen[\set{N},\set{X}]{\sqrt{\matr{A}^{\dagger}}}\msubgen[\set{X}]{\matr{A}^{\dagger}}^{\dagger}\msubgen[\set{N},\set{X}]{\sqrt{\matr{A}^{\dagger}}}.
    \end{align}

    Note that 
    \begin{equation}
        \proj{!} = \msubgen[\set{N},\set{X}]{\sqrt{\matr{A}^{\dagger}}}\msubgen[\set{X}]{\matr{A}^{\dagger}}^{\dagger}\msubgen[\set{N},\set{X}]{\sqrt{\matr{A}^{\dagger}}}
    \end{equation}
    is a projection because $\proj{!} = \proj{!}^{T} = \proj{!}^{2}$, and thus its eigenvalues are 0 or 1 (i.e the real roots of $\lambda = \lambda^{2}$). Thus, for any vector $\vect{v}$ with $\sqnormvec{v} = 1$,
    \begin{align}
    &\sqnormvec{\vect{v}^{T}\sqrt{\matr{B}\matr{A}^{\dagger}}\proj{X}\sqrt{\matr{A}}\projgen{k} } \label{eq:bandlimited_kantorovich_proof_end1}\\
        &=\vect{v}^{T}\msubgen[\set{N},\set{X}]{\sqrt{\matr{B}\matr{A}^{\dagger}}}\msubgen[\set{X}]{\matr{A}\projbl} \msubgen[\set{X},\set{N}]{\sqrt{\matr{A}^{\dagger}\matr{B}}}\vect{v}\\
        &\leq \frac{(M+m)^{2}}{4Mm} \vect{v}^{T}\sqrt{\matr{B}}\proj{!}\sqrt{\matr{B}}\vect{v} \\
        &\leq \frac{(M+m)^{2}}{4Mm} \sqnormvec{\sqrt{\matr{B}\vect{v}}} \\
        &\leq \frac{(M+m)^{2}}{4Mm} \label{eq:bandlimited_kantorovich_proof_end2}
    \end{align}
    where we have $\sqrt{\matr{A}} \projgen{k} \sqrt{\matr{A}} = \matr{A}\proj{k}$ because $\sqrt{\matr{A}}$ commutes with $\projgen{k}$ (as they are both diagonal in the same basis).
    
    As (\ref{eq:bandlimited_kantorovich_proof_end1} - \ref{eq:bandlimited_kantorovich_proof_end2}) holds true for all normalised $\vect{v}$, taking square roots gives the proof.
\end{proof}

\subsection{Bounding (\ref{eq:GLR_bandlimited_proof:v1_term})}
We show that 
\begin{equation}
    {\frac{N}{m}}\sqnormvec{\vect{v}_{1}^{T}{\left( \matr{I} - \msubgen[\set{N},\Theta]{\matr{L}^{\dagger}} \msubgen[\Theta]{ \matr{L}_{\mu}^{\dagger}}^{-1}\msubgen[\Theta,\set{N}]{\matr{I}}\right)\msubgen[\set{N},2:k]{\matr{U}}}} \leq \frac{N}{m} r_{bl}.
\end{equation}

\begin{proof} 
First note that $\vect{v}_{1}^{T}\matrsub[N,\Theta]{I} = \vect{0}$, , $\projgen{\matr{L}}\matrsub[N,\Theta]{I} = \matrsub[N,\Theta]{I}$ and so $\vect{v}_{1}^{T}\msub[N,\Theta]{\mu\projgen{\matr{L}}} = \mu \vect{v}_{1}^{T}\matrsub[N,\Theta]{I} = 0$ so 
\begin{align}
\vect{v}_{1}^{T}\msubgen[\set{N},\Theta]{\matr{L}^{\dagger}} &= \vect{v}_{1}^{T} \msubgen[\set{N},\Theta]{\projgen{\matr{L}}\matr{L}^{\dagger}} \\
&= \vect{v}_{1}^{T} \msubgen[\set{N},\Theta]{\mu\projgen{\matr{L}} + \projgen{\matr{L}}\matr{L}^{\dagger}\projgen{\matr{L}}} \\
&= \vect{v}_{1}^{T} \msubgen[\set{N},\Theta]{\projgen{\matr{L}}\matr{L}_{\mu}^{\dagger}\projgen{\matr{L}}} 
\end{align}
Write $\projgen{\matr{L}}\matr{L}^{\dagger}_{\mu}\projgen{\matr{L}} = \matr{L}^{\dagger}_{\mu-}$. This zeros the smallest eigenvalue of $\matr{L}^{\dagger}_{\mu}$ which is $\mu$, meaning the eigenvalues of $\matr{L}^{\dagger}_{\mu-}$ are between $\mu + \frac{1}{\lambda_{N}}$ and $\mu + \frac{1}{\lambda_{2}}$. Then
\begin{align}
    &{\vect{v}_{1}^{T}{\left( \matr{I} - \msubgen[\set{N},\Theta]{\matr{L}^{\dagger}} \msubgen[\Theta]{ \matr{L}_{\mu}^{\dagger}}^{-1}\msubgen[\Theta,\set{N}]{\matr{I}}\right)\msubgen[\set{N},2:k]{\matr{U}}}} \\
    &= {\vect{v}_{1}^{T}{\left( \matr{I} - \msubgen[\set{N},\Theta]{\matr{L}_{\mu -}^{\dagger}} \msubgen[\Theta]{ \matr{L}_{\mu}^{\dagger}}^{-1}\msubgen[\Theta,\set{N}]{\matr{I}}\right)\msubgen[\set{N},2:k]{\matr{U}}}}.
\end{align}
Note that because $\projgen{\matr{L}}\matrsub[N,\Theta]{I} = \matrsub[N,\Theta]{I}$,
\begin{equation}
    \msub[\Theta]{\matr{L}^{\dagger}_{\mu -}}
     = \msub[\Theta]{\matr{L}^{\dagger}_{\mu}}
\end{equation}
    We note that the following are projection operators:
    \begin{align}
        \proj{\mu} &= \msub[N,\Theta]{\sqrt{\matr{L}_{\mu-}^{\dagger}}}\msub[\Theta]{\matr{L}^{\dagger}_{\mu}}^{-1} \msub[\Theta,N]{\sqrt{\matr{L}_{\mu-}^{\dagger}}} \\
        \projgen{\mu^{c}} &= \matr{I} - \proj{\mu}
    \end{align}
    which can be checked by seeing that $\proj{\mu} = \proj{\mu}^{T} = \proj{\mu}^{2}$.
As $\projgen{\matr{L}}\msubgen[\set{N},2:k]{\matr{U}} = \msubgen[\set{N},2:k]{\matr{U}}$,
\begin{align} 
    &{\vect{v}_{1}^{T}{\left( \matr{I} - \msubgen[\set{N},\Theta]{\matr{L}_{\mu -}^{\dagger}} \msubgen[\Theta]{ \matr{L}_{\mu}^{\dagger}}^{-1}\msubgen[\Theta,\set{N}]{\matr{I}}\right)\msubgen[\set{N},2:k]{\matr{U}}}} \\
    &=\vect{v}_{1}^{T}\sqrt{\matr{L}^{\dagger}_{\mu-}}\projgen{\mu^{c}}\sqrt{\left(\matr{L}^{\dagger}_{\mu-}\right)^{\dagger}} \msubgen[\set{N},2:k]{\matr{U}}
\end{align}
Therefore
\begin{align}
    &\sqnormvec{\vect{v}_{1}^{T}{\left( \matr{I} - \msubgen[\set{N},\Theta]{\matr{L}^{\dagger}} \msubgen[\Theta]{ \matr{L}_{\mu}^{\dagger}}^{-1}\msubgen[\Theta,\set{N}]{\matr{I}}\right)\msubgen[\set{N},2:k]{\matr{U}}}} \\
    &= \sqnormvec{ \vect{v}_{1}^{T}\sqrt{\matr{L}^{\dagger}_{\mu-}}\projgen{\mu^{c}}\sqrt{\left(\matr{L}^{\dagger}_{\mu-}\right)^{\dagger}} \msubgen[\set{N},2:k]{\matr{U}}} \\
    &\leq \sigma_{max}^{2}\left(\sqrt{\matr{L}^{\dagger}_{\mu-}}\projgen{\mu^{c}}\sqrt{\left(\matr{L}^{\dagger}_{\mu-}\right)^{\dagger}} \msubgen[\set{N},2:k]{\matr{U}} \right)
\end{align}
We now apply Lemma \ref{lemma:bandlimited_kantorovich} with $\matr{A} = {\left(\matr{L}^{\dagger}_{\mu-}\right)^{\dagger}}$, $\matr{B}=\matr{I}$ and $\projgen{k} = \msubgen[\set{N},2:k]{\matr{U}}\msubgen[\set{N},2:k]{\matr{U}}^{T}$. Then
\begin{align}
    M &= \frac{1}{\mu + \frac{1}{\lambda_{k}}} \\
    m &= \frac{1}{\mu + \frac{1}{\lambda_{2}}}
\end{align}
so
\begin{align}
\sigma_{max}^{2}\left(\sqrt{\matr{L}^{\dagger}_{\mu-}}\projgen{\mu^{c}}\sqrt{\left(\matr{L}^{\dagger}_{\mu-}\right)^{\dagger}} \msubgen[\set{N},2:k]{\matr{U}} \right) \leq \frac{(M+m)^{2}}{4Mm} 
\end{align}
Note that 
\begin{align}
    \frac{(M+m)^{2}}{4Mm} = \omega\left(\frac{M}{m}\right) = \omega\left(\frac{\mu + \frac{1}{\lambda_{2}}}{\mu + \frac{1}{\lambda_{k}}}\right)
\end{align}
as $\frac{1}{\lambda_{2}} \geq \frac{1}{\lambda_{k}}$ we know that $\frac{\mu + \frac{1}{\lambda_{2}}}{\mu + \frac{1}{\lambda_{k}}}$ is decreasing as a function of $\mu$ so
\begin{equation}
    \frac{\mu + \frac{1}{\lambda_{2}}}{\mu + \frac{1}{\lambda_{k}}} \leq \frac{ \frac{1}{\lambda_{2}}}{ \frac{1}{\lambda_{k}}} \leq \frac{\lambda_{k}}{\lambda_{2}}
\end{equation}
as $\omega$ is increasing this means that
\begin{equation}
    \frac{(M+m)^{2}}{4Mm} \leq \omega\left(\frac{\lambda_{k}}{\lambda_{2}}\right) = r_{bl}
\end{equation}
combining these results gives
\begin{equation}
    \sqnormvec{\vect{v}_{1}^{T}{\left( \matr{I} - \msubgen[\set{N},\Theta]{\matr{L}^{\dagger}} \msubgen[\Theta]{ \matr{L}_{\mu}^{\dagger}}^{-1}\msubgen[\Theta,\set{N}]{\matr{I}}\right)\msubgen[\set{N},2:k]{\matr{U}}}} \leq r_{bl}
\end{equation}
and multiplying both sides by $\frac{N}{m}$ gives the result.
\end{proof}

\section{Proof of Theorem \ref{thm:main_GLR_bl}}
\label{app:GLR_bandlimited_thm}
\begin{proof}
To get the bounds on $m_{opt\_bl}$ given $B_{k}(m_{opt\_bl}) < k-1$, we can follow Appendix \ref{app:proof_of_remark_GLR_mopt} replacing $r$ with $r_{bl}$. Let
\begin{align} 
    \mu_{ub\_bl} &= {\lambda_{k}}^{-1}\left(\sqrt{\frac{k}{1 +B_{k}(m_{opt\_bl})}} -1\right)  \\
    \tau_{GLR\_bl}(\mu) &=  \frac{\left(\sum^{k}_{i=1}\left(1 + \mu\lambda_{i} \right)^{-2}\right) - B_{k}(m_{opt\_bl})}{k+B_{k}(m_{opt\_bl}) - \sum_{i=1}^{k} \left(1 - (1+\mu\lambda_{i})^{-1}\right)^{2}} 
\end{align}

We follow the structure of Appendix \ref{app:Proof_thm_main_GLR_exist}. Firstly, we note that by the same arguments as for Lemma \ref{lemma:GLR_full_observation_MSE},
\begin{align}
    \xi_{2}(\set{N}) = \sum_{i=1}^{k} \left( \frac{1}{1+\mu\lambda_{i}} \right)^{2}.
\end{align}

Using \ref{lemma:GLR_xi_2_bound_main_bl} to get bounds on $\xi_{2}(\set{S})$, we continue to follow Appendix \ref{app:Proof_thm_main_GLR_exist}. Note that $\xi_{1}$ is not changed by changing our noise model. Therefore

\begin{align}
    \Delta_{2}(\set{N},\set{S}^{c}) &\geq \sum_{i=1}^{k} \left( \frac{1}{1+\mu\lambda_{i}}\right)^{2} - (1+B_{k}(m)) \label{eq:proof_main_GLR_exist:delta_2_lb_bl} \\
    \Delta_{1}(\set{N},\set{S}^{c}) &\geq \sum_{i=1}^{N} \left(1 - \frac{1}{1+\mu\lambda_{i}}\right)^{2} - (k + B_{k}(m)) \label{eq:proof_main_GLR_exist:delta_1_lb_bl}.
\end{align}

We now show that $\Delta_{2} > 0$. We have that $\omega(x) > 0$ so $B(m) > 0$ and by assumption $B(m_{opt\_bl}) < k$. Therefore $\mu_{ub\_bl} > 0$ and is real and it is therefore possible to pick $0 < \mu < \mu_{ub\_bl}$.
By assumption $0<\mu<\mu_{ub\_bl}$ and so
\begin{align}
    \sum_{i=1}^{k}\left(\frac{1}{1 + \mu\lambda_{i}}\right)^{2} &\geq \frac{k}{\left(1+\mu\frac{1}{\lambda_{k}}  \right)^{2}}\\
    &> \frac{k}{\left(1+\mu_{ub\_bl}\frac{1}{\lambda_{k}} 
  \right)^{2}}\\ 
    &= 1 + B_{k}(m_{opt\_bl})
\end{align}
And therefore by (\ref{eq:proof_main_GLR_exist:delta_2_lb_bl}), $\Delta_{2}(\set{N},\set{S}^{c}) > 0$. Note that $\textrm{SNR} = \frac{1}{\sigma^{2}}$ and apply Corollary \ref{corr:main_GLR_iff} in the same manner as Appendix \ref{app:Proof_thm_main_GLR_exist} and we are done.
\end{proof}

\section{Proof of Proposition \ref{propn:GLR_big_N_bl}}
\label{app:Proof_GLR_big_N_bl}
The proof in Appendix \ref{app:Proof_GLR_big_N} adapts immediately to the bandlimited case, by noting that $1 \leq r_{bl} \leq r$ and $B_{k}(m) \leq B(m)$ and using the squeeze theorem. 
\section{Additional Results}
\label{plot_appendix}
Under LS reconstruction, we show the threshold $\tau(\set{S},v)$ for the BA and SBM graphs with 1000, 2000 and 3000 vertices (Figs. 
\ref{LS_SNR_Threshold_plots_big_BA} and  \ref{LS_SNR_Threshold_plots_big_SBM}). We also present MSE plots for BA and SBM graphs under LS with full-band noise (Figs. \ref{LS_BA_MSE_fig} and \ref{LS_SBM_MSE_fig}), LS with bandlimited noise (Figs. \ref{LS_BA_MSE_bandlimited_fig} and  \ref{LS_SBM_MSE_bandlimited_fig}),  GLR with full-band noise (Figs. \ref{GLR_BA_MSE_fig} and  \ref{GLR_SBM_MSE_fig}) and GLR with bandlimited noise (Figs. \ref{bandlimited_GLR_BA_MSE_fig} and  \ref{bandlimited_GLR_SBM_MSE_fig}).

\begin{figure*}%
    \centering
    \begin{subfigure}{0.6\columnwidth}
    \includegraphics[width=\columnwidth]{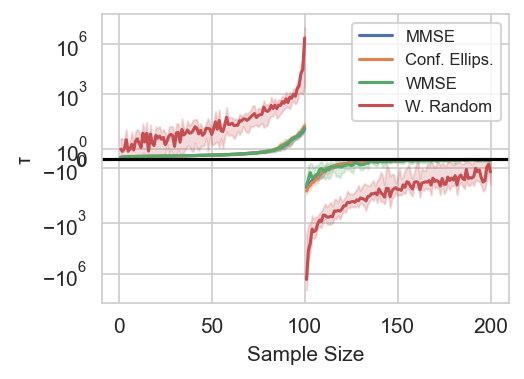}
    \caption{1000 vertices} 
    \label{snr_BA_1000}
    \end{subfigure}
    \hfill
    \begin{subfigure}{0.6\columnwidth}
    \includegraphics[width=\columnwidth]{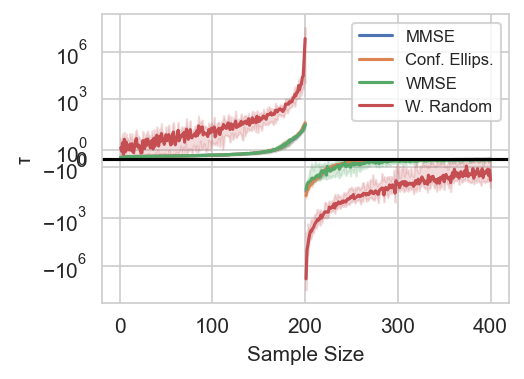}%
    \caption{2000 vertices}%
    \label{snr_BA_2000}%
    \end{subfigure}
    \hfill%
    \begin{subfigure}{0.6\columnwidth}
    \includegraphics[width=\columnwidth]{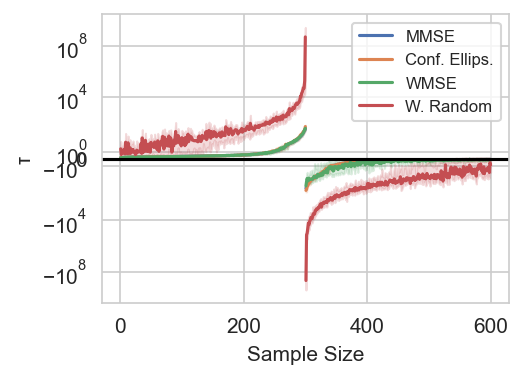}%
    \caption{3000 vertices}%
    \label{snr_BA_3000}%
    \end{subfigure}%
    \caption{$\tau(\set{S},v)$ for different sized BA graphs under LS reconstruction (bandwidth = $\frac{\# \text{ vertices}}{10}$).}
\label{LS_SNR_Threshold_plots_big_BA}
\end{figure*}

\begin{figure*}%
    \centering
    \begin{subfigure}{0.6\columnwidth}
    \includegraphics[width=\columnwidth]{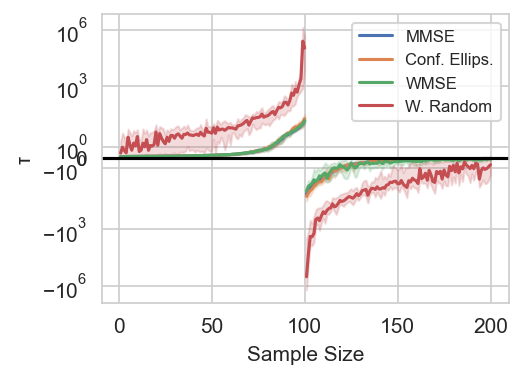}
    \caption{1000 vertices} 
    \label{snr_SBM_1000}
    \end{subfigure}
    \hfill
    \begin{subfigure}{0.6\columnwidth}
    \includegraphics[width=\columnwidth]{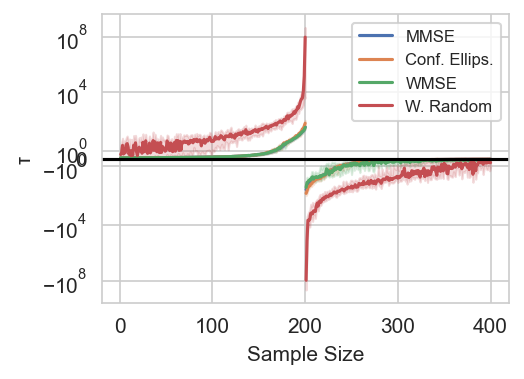}%
    \caption{2000 vertices}%
    \label{snr_SBM_2000}%
    \end{subfigure}
    \hfill%
    \begin{subfigure}{0.6\columnwidth}
    \includegraphics[width=\columnwidth]{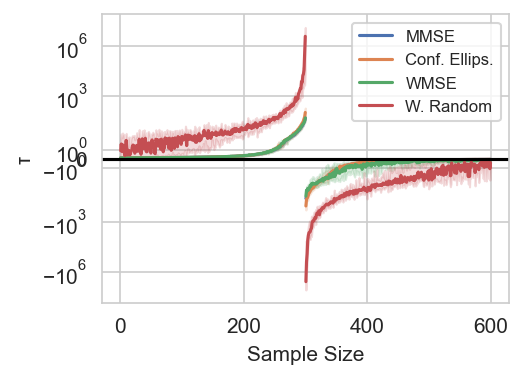}%
    \caption{3000 vertices}%
    \label{snr_SBM_3000}%
    \end{subfigure}%
    \caption{$\tau(\set{S},v)$ for different sized SBM graphs under LS reconstruction (bandwidth = $\frac{\# \text{ vertices}}{10}$).}
\label{LS_SNR_Threshold_plots_big_SBM}
\end{figure*}

\begin{figure*}%
    \centering
    \begin{subfigure}{0.6\columnwidth}
    \includegraphics[width=\columnwidth]{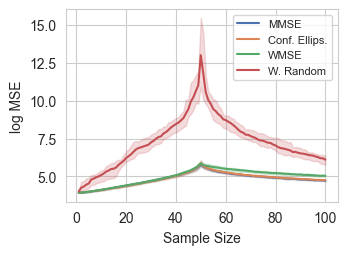}
    \caption{SNR = $10^{-1}$}
    \label{BA_MSE_subfiga}
    \end{subfigure}\hfill
    \begin{subfigure}{0.6\columnwidth}
    \includegraphics[width=\columnwidth]{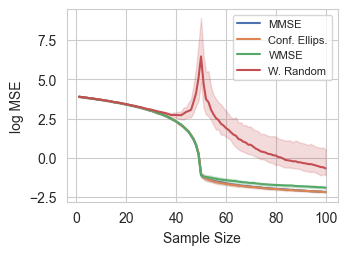}%
    \caption{SNR = $10^{2}$}%
    \label{BA_MSE_subfigb}%
    \end{subfigure}\hfill%
    \begin{subfigure}{0.6\columnwidth}
    \includegraphics[width=\columnwidth]{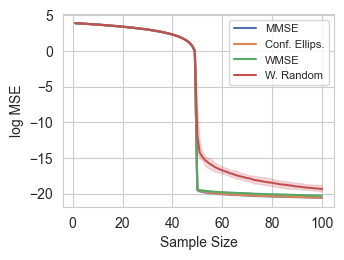}%
    \caption{SNR = $10^{10}$}%
    \label{BA_MSE_subfigc}%
    \end{subfigure}%
    \caption{Average MSE for LS reconstruction on BA graphs (\#vertices=500, bandwidth = 50) with different SNRs.}
\label{LS_BA_MSE_fig}
\end{figure*}

\begin{figure*}%
    \centering
    \begin{subfigure}{0.6\columnwidth}
    \includegraphics[width=\columnwidth]{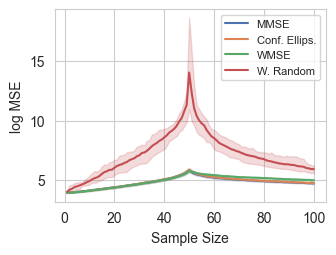}
    \caption{SNR = $10^{-1}$}
    \label{bandlimited_BA_MSE_subfiga}
    \end{subfigure}\hfill
    \begin{subfigure}{0.6\columnwidth}
    \includegraphics[width=\columnwidth]{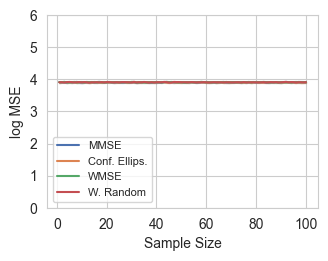}
    \caption{SNR = $1$}%
    \label{bandlimited_BA_MSE_subfigb}%
    \end{subfigure}\hfill%
    \begin{subfigure}{0.6\columnwidth}
    \includegraphics[width=\columnwidth]{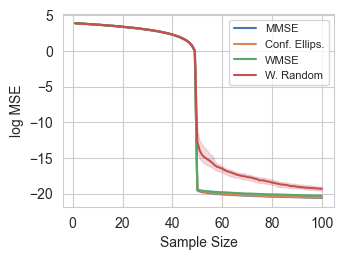}
    \caption{SNR = $10^{10}$}%
    \label{bandlimited_BA_MSE_subfigc}%
    \end{subfigure}%
    \caption{Average MSE for LS reconstruction on BA graphs with bandlimited noise (\#vertices=500, bandwidth = 50) with different SNRs.}
\label{LS_BA_MSE_bandlimited_fig}
\end{figure*}

\begin{figure*}%
    \centering
    \begin{subfigure}{0.6\columnwidth}
    \includegraphics[width=\columnwidth]{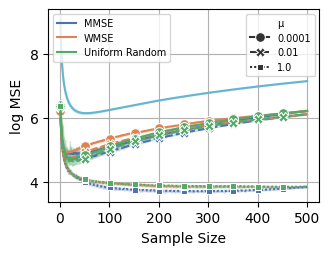}
    \caption{SNR = $10^{-1}$}
    \label{GLR_BA_MSE_subfiga}
    \end{subfigure}\hfill
    \begin{subfigure}{0.6\columnwidth}
    \includegraphics[width=\columnwidth]{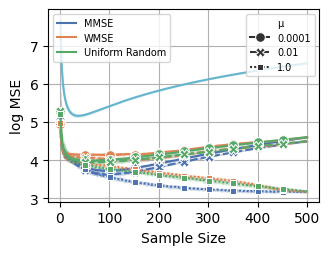}
    \caption{SNR = $\frac{1}{2}$}%
    \label{GLR_BA_MSE_subfigb}%
    \end{subfigure}\hfill%
    \begin{subfigure}{0.6\columnwidth}
    \includegraphics[width=\columnwidth]{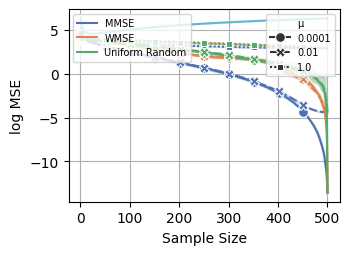}
    \caption{SNR = $10^{10}$}%
    \label{GLR_BA_MSE_subfigc}%
    \end{subfigure}%
    \caption{Average MSE for GLR reconstruction on BA graphs (\#vertices=500, bandwidth = 50) with different SNRs. Line without markers is an upper bound.}
\label{GLR_BA_MSE_fig}
\end{figure*}

\begin{figure*}%
    \centering
    \begin{subfigure}{0.6\columnwidth}
    \includegraphics[width=\columnwidth]{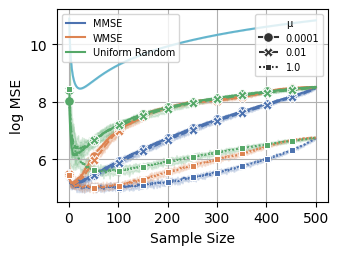}
    \caption{SNR = $10^{-2}$}
    \label{bandlimited_GLR_BA_MSE_subfiga}
    \end{subfigure}\hfill
    \begin{subfigure}{0.6\columnwidth}
    \includegraphics[width=\columnwidth]{gm_BA_sn-3_fb.png}
    \caption{SNR = $\frac{1}{2}$}%
    \label{bandlimited_GLR_BA_MSE_subfigb}%
    \end{subfigure}\hfill%
    \begin{subfigure}{0.6\columnwidth}
    \includegraphics[width=\columnwidth]{gm_BA_sn100_fb.png}
    \caption{SNR = $10^{10}$}%
    \label{bandlimited_GLR_BA_MSE_subfigc}%
    \end{subfigure}%
    \caption{Average MSE for GLR reconstruction on BA graphs under bandlimited noise (\#vertices=500, bandwidth = 50) with different SNRs. Line without markers is an upper bound.}
\label{bandlimited_GLR_BA_MSE_fig}
\end{figure*}

\begin{figure*}%
    \centering
    \begin{subfigure}{0.6\columnwidth}
    \includegraphics[width=\columnwidth]{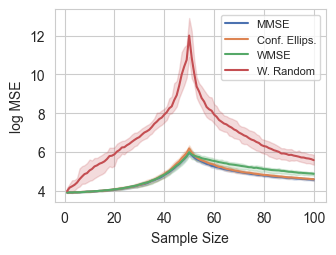}
    \caption{SNR = $10^{-1}$}
    \label{SBM_MSE_subfiga}
    \end{subfigure}\hfill
    \begin{subfigure}{0.6\columnwidth}
    \includegraphics[width=\columnwidth]{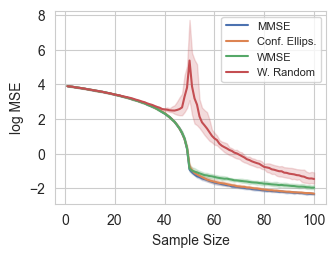}%
    \caption{SNR = $10^{2}$}%
    \label{SBM_MSE_subfigb}%
    \end{subfigure}\hfill%
    \begin{subfigure}{0.6\columnwidth}
    \includegraphics[width=\columnwidth]{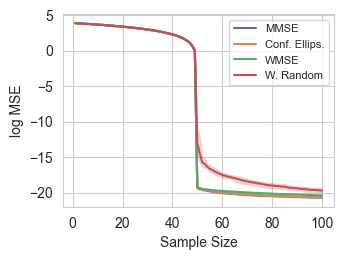}%
    \caption{SNR = $10^{10}$}%
    \label{SBM_MSE_subfigc}%
    \end{subfigure}%
    \caption{Average MSE for LS reconstruction on SBM graphs (\#vertices=500, bandwidth = 50) with different SNRs.}
\label{LS_SBM_MSE_fig}
\end{figure*}

\begin{figure*}%
    \centering
    \begin{subfigure}{0.6\columnwidth}
    \includegraphics[width=\columnwidth]{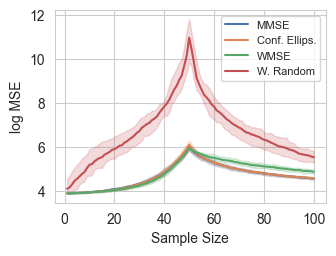}
    \caption{SNR = $10^{-1}$}
    \label{bandlimited_SBM_MSE_subfiga}
    \end{subfigure}\hfill
    \begin{subfigure}{0.6\columnwidth}
    \includegraphics[width=\columnwidth]{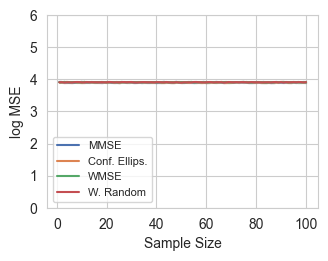}
    \caption{SNR = $1$}%
    \label{bandlimited_SBM_MSE_subfigb}%
    \end{subfigure}\hfill%
    \begin{subfigure}{0.6\columnwidth}
    \includegraphics[width=\columnwidth]{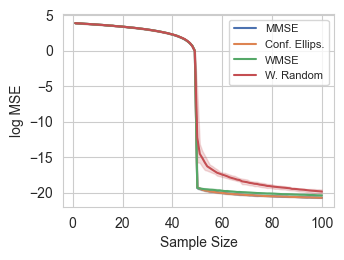}
    \caption{SNR = $10^{10}$}%
    \label{bandlimited_SBM_MSE_subfigc}%
    \end{subfigure}%
    \caption{Average MSE for LS reconstruction on SBM graphs with bandlimited noise (\#vertices=500, bandwidth = 50) with different SNRs.}
\label{LS_SBM_MSE_bandlimited_fig}
\end{figure*}

\begin{figure*}%
    \centering
    \begin{subfigure}{0.6\columnwidth}
    \includegraphics[width=\columnwidth]{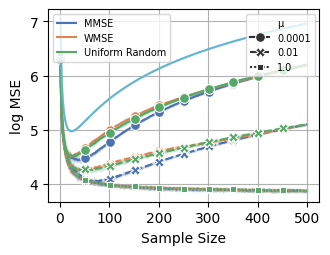}
    \caption{SNR = $10^{-1}$}
    \label{GLR_SBM_MSE_subfiga}
    \end{subfigure}\hfill
    \begin{subfigure}{0.6\columnwidth}
    \includegraphics[width=\columnwidth]{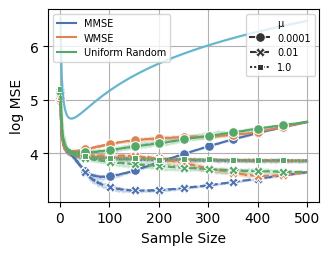}
    \caption{SNR = $\frac{1}{2}$}%
    \label{GLR_SBM_MSE_subfigb}%
    \end{subfigure}\hfill%
    \begin{subfigure}{0.6\columnwidth}
    \includegraphics[width=\columnwidth]{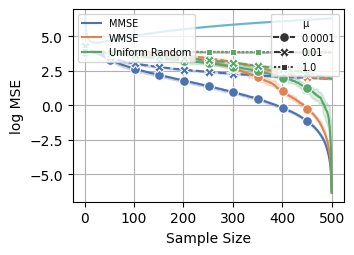}
    \caption{SNR = $10^{10}$}%
    \label{GLR_SBM_MSE_subfigc}%
    \end{subfigure}%
    \caption{Average MSE for GLR reconstruction on SBM graphs (\#vertices=500, bandwidth = 50) with different SNRs. Line without markers is an upper bound.}
\label{GLR_SBM_MSE_fig}
\end{figure*}

\begin{figure*}%
    \centering
    \begin{subfigure}{0.6\columnwidth}
    \includegraphics[width=\columnwidth]{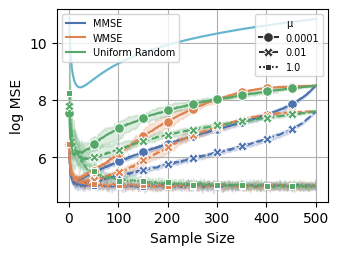}
    \caption{SNR = $10^{-2}$}
    \label{bandlimited_GLR_SBM_MSE_subfiga}
    \end{subfigure}\hfill
    \begin{subfigure}{0.6\columnwidth}
    \includegraphics[width=\columnwidth]{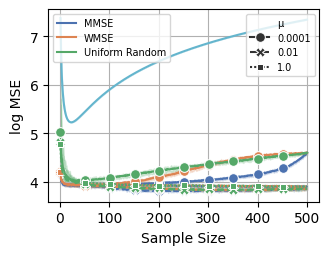}
    \caption{SNR = $\frac{1}{2}$}%
    \label{bandlimited_GLR_SBM_MSE_subfigb}%
    \end{subfigure}\hfill%
    \begin{subfigure}{0.6\columnwidth}
    \includegraphics[width=\columnwidth]{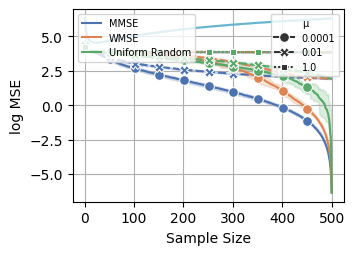}
    \caption{SNR = $10^{10}$}%
    \label{bandlimited_GLR_SBM_MSE_subfigc}%
    \end{subfigure}%
    \caption{Average MSE for GLR reconstruction on SBM graphs under bandlimited noise (\#vertices=500, bandwidth = 50) with different SNRs. Line without markers is an upper bound.}
\label{bandlimited_GLR_SBM_MSE_fig}
\end{figure*}

\end{document}